\PassOptionsToPackage{unicode}{hyperref}
\PassOptionsToPackage{hyphens}{url}
\PassOptionsToPackage{dvipsnames,svgnames,x11names}{xcolor}
\documentclass[
  12pt]{article}

\usepackage{amsmath,amssymb}
\usepackage{iftex}
\ifPDFTeX
  \usepackage[T1]{fontenc}
  \usepackage[utf8]{inputenc}
  \usepackage{textcomp} 
\else 
  \usepackage{unicode-math}
  \defaultfontfeatures{Scale=MatchLowercase}
  \defaultfontfeatures[\rmfamily]{Ligatures=TeX,Scale=1}
\fi
\usepackage{lmodern}
\ifPDFTeX\else  
\fi
\IfFileExists{upquote.sty}{\usepackage{upquote}}{}
\IfFileExists{microtype.sty}{
  \usepackage[]{microtype}
  \UseMicrotypeSet[protrusion]{basicmath} 
}{}
\makeatletter
\@ifundefined{KOMAClassName}{
  \IfFileExists{parskip.sty}{%
    \usepackage{parskip}
  }{
    \setlength{\parindent}{0pt}
    \setlength{\parskip}{6pt plus 2pt minus 1pt}}
}{
  \KOMAoptions{parskip=half}}
\makeatother
\usepackage{xcolor}
\makeatletter
\ifx\paragraph\undefined\else
  \let\oldparagraph\paragraph
  \renewcommand{\paragraph}{
    \@ifstar
      \xxxParagraphStar
      \xxxParagraphNoStar
  }
  \newcommand{\xxxParagraphStar}[1]{\oldparagraph*{#1}\mbox{}}
  \newcommand{\xxxParagraphNoStar}[1]{\oldparagraph{#1}\mbox{}}
\fi
\ifx\subparagraph\undefined\else
  \let\oldsubparagraph\subparagraph
  \renewcommand{\subparagraph}{
    \@ifstar
      \xxxSubParagraphStar
      \xxxSubParagraphNoStar
  }
  \newcommand{\xxxSubParagraphStar}[1]{\oldsubparagraph*{#1}\mbox{}}
  \newcommand{\xxxSubParagraphNoStar}[1]{\oldsubparagraph{#1}\mbox{}}
\fi
\makeatother

\usepackage{longtable,booktabs,array}
\usepackage{calc} 
\usepackage{etoolbox}
\makeatletter
\patchcmd\longtable{\par}{\if@noskipsec\mbox{}\fi\par}{}{}
\makeatother
\IfFileExists{footnotehyper.sty}{\usepackage{footnotehyper}}{\usepackage{footnote}}
\makesavenoteenv{longtable}
\usepackage{graphicx}
\makeatletter
\def\maxwidth{\ifdim\Gin@nat@width>\linewidth\linewidth\else\Gin@nat@width\fi}
\def\maxheight{\ifdim\Gin@nat@height>\textheight\textheight\else\Gin@nat@height\fi}
\makeatother
\setkeys{Gin}{width=\maxwidth,height=\maxheight,keepaspectratio}
\makeatletter
\def\fps@figure{htbp}
\makeatother

\makeatletter
\@ifpackageloaded{caption}{}{\usepackage{caption}}
\AtBeginDocument{%
\ifdefined\contentsname
  \renewcommand*\contentsname{Table of contents}
\else
  \newcommand\contentsname{Table of contents}
\fi
\ifdefined\listfigurename
  \renewcommand*\listfigurename{List of Figures}
\else
  \newcommand\listfigurename{List of Figures}
\fi
\ifdefined\listtablename
  \renewcommand*\listtablename{List of Tables}
\else
  \newcommand\listtablename{List of Tables}
\fi
\ifdefined\figurename
  \renewcommand*\figurename{Figure}
\else
  \newcommand\figurename{Figure}
\fi
\ifdefined\tablename
  \renewcommand*\tablename{Table}
\else
  \newcommand\tablename{Table}
\fi
}
\@ifpackageloaded{float}{}{\usepackage{float}}
\floatstyle{ruled}
\@ifundefined{c@chapter}{\newfloat{codelisting}{h}{lop}}{\newfloat{codelisting}{h}{lop}[chapter]}
\floatname{codelisting}{Listing}

\makeatother
\makeatletter
\@ifpackageloaded{caption}{}{\usepackage{caption}}
\@ifpackageloaded{subcaption}{}{\usepackage{subcaption}}
\makeatother

\ifLuaTeX
  \usepackage{selnolig}  
\fi
\usepackage[]{natbib}
\usepackage{bookmark}

\IfFileExists{xurl.sty}{\usepackage{xurl}}{} 
\hypersetup{
  pdftitle={Title},
  pdfauthor={Author 1; Author 2},
  pdfkeywords={3 to 6 keywords, that do not appear in the title},
  colorlinks=true,
  linkcolor={blue},
  filecolor={Maroon},
  citecolor={Blue},
  urlcolor={Blue},
  pdfcreator={LaTeX via pandoc}}

\newcommand{\anon}{1}

\usepackage[english]{babel}
\usepackage{epstopdf}

\usepackage{amscd,amsfonts,amsmath,amssymb,amsthm,bbm,bm,latexsym,mathrsfs}
\usepackage{graphicx}
\usepackage{enumerate}
\usepackage{natbib}

\usepackage{tikz}
\usepackage{tkz-berge}
\usetikzlibrary{shapes,arrows,bayesnet}	

\usepackage{verbatim}
\usepackage{makecell}

\newtheorem{lemma}{Lemma}[]
\newtheorem{corollary}{Corollary}[]
\newtheorem{proposition}{Proposition}

\usepackage{xr} 
\usepackage{xr-hyper} 
\newsavebox{\linewithcirclesbox}
\savebox{\linewithcirclesbox}{%
    \tikz[baseline=-0.5ex]{%
        \draw[thick] (0,0) -- (0.5cm,0); 
        \foreach \x in {0,0.2,0.4} {
            \fill (\x cm,0) circle (1.2pt); 
        }
    }%
}

\definecolor{customblue}{HTML}{ADD2E2}
\newsavebox{\linewithcirclesboxBlue}
\savebox{\linewithcirclesboxBlue}{%
    \tikz[baseline=-0.5ex]{%
        \draw[thick, customblue] (0,0) -- (0.5cm,0); 
        \foreach \x in {0,0.2,0.4} {
            \fill[customblue] (\x cm,0) circle (1.2pt); 
        }
    }%
}

\definecolor{customyellow}{HTML}{F1F834}

\newsavebox{\yellowbigdot}
\savebox{\yellowbigdot}{%
    \tikz[baseline=-0.5ex]{%
        \fill[customyellow, draw=yellow!60!black, line width=0.5pt] (0,0) circle (4pt);
    }%
}

\definecolor{custompurple}{HTML}{241E91}

\newsavebox{\purplesmalldot}
\savebox{\purplesmalldot}{%
    \tikz[baseline=-0.5ex]{%
        \fill[custompurple, draw=purple!60!black, line width=0.5pt] (0,0) circle (2pt);
    }%
}

\newsavebox{\lineRed}
\savebox{\lineRed}{%
    \tikz[baseline=-0.5ex]{%
        \draw[thick, red!60] (0,0) -- (0.5cm,0); 
    }%
}

\newsavebox{\lineRedDash}
\savebox{\lineRedDash}{%
    \tikz[baseline=-0.5ex]{%
        \draw[thick, red!60, dashed] (0,0) -- (0.5cm,0); 
}
}

\newsavebox{\lineBlue}
\savebox{\lineBlue}{%
    \tikz[baseline=-0.5ex]{%
        \draw[thick, blue!60] (0,0) -- (0.5cm,0); 
    }%
}

\newsavebox{\lineBlueDash}
\savebox{\lineBlueDash}{%
    \tikz[baseline=-0.5ex]{%
        \draw[thick, blue!60, dashed] (0,0) -- (0.5cm,0); 
}
}

\newsavebox{\lineGreen}
\savebox{\lineGreen}{%
    \tikz[baseline=-0.5ex]{%
        \draw[thick, green!60] (0,0) -- (0.5cm,0); 
}
}

\newsavebox{\lineGreenDash}
\savebox{\lineGreenDash}{%
    \tikz[baseline=-0.5ex]{%
        \draw[thick, green!60, dashed] (0,0) -- (0.5cm,0); 
}
}

\newsavebox{\lineBlack}
\savebox{\lineBlack}{%
    \tikz[baseline=-0.5ex]{%
        \draw[thick, black] (0,0) -- (0.5cm,0); 
    }%
}

\newsavebox{\lineBlackDash}
\savebox{\lineBlackDash}{%
  \tikz[baseline=-0.5ex]{%
    \draw[black, thick, dashed] (0,0) -- (0.5,0);
  }%
}

\newsavebox{\DotRed}
\savebox{\DotRed}{%
    \tikz[baseline=-0.5ex]{%
        \fill[red!60] (0.2,0) circle (1.2pt); 
    }%
}

\newsavebox{\grayrectanglebox}
\savebox{\grayrectanglebox}{%
    \tikz[baseline={(0.1cm,0)}]{%
        \fill[gray!50] (0,0) rectangle (0.5cm,0.15cm); 
    }%
}

\newsavebox{\bluerectanglebox}
\savebox{\bluerectanglebox}{%
    \tikz[baseline={(0.1cm,0)}]{%
        \fill[customblue] (0,0) rectangle (0.5cm,0.15cm); 
    }%
}

\usepackage{url}

\graphicspath{{Fig/}}

\theoremstyle{thmstyleone}%
\theoremstyle{thmstyletwo}%
\theoremstyle{thmstylethree}%
\newtheorem{definition}{Definition}

\begin{document}

\def\spacingset#1{\renewcommand{\baselinestretch}%
{#1}\small\normalsize} \spacingset{1}


\if1\anon
{
  \title{\bf Generalized Poisson Dynamic Network Models}
  \author{Giulia Carallo\thanks{ University of Salento, Italy, \color{blue}\texttt{giulia.carallo@unisalento.it}},
    Roberto Casarin\thanks{Ca' Foscari University of Venice, Italy,
\color{blue}\texttt{r.casarin@unive.it}}, and
    Antonio Peruzzi\thanks{Ca' Foscari University of Venice, Italy,
\color{blue}\texttt{antonio.peruzzi@unive.it}} \\}

  \maketitle
} \fi

\if0\anon
{
  \bigskip
  \bigskip
  \bigskip
  \begin{center}
    {\LARGE\bf Generalized Poisson Dynamic Network Models}
\end{center}
  \medskip
} \fi

\bigskip
\begin{abstract}
Count-weighted temporal networks often exhibit unequal dispersion in the edge weights, which cannot be fully explained by modelling observational heterogeneity through latent factors in the conditional mean.  Therefore, we propose new dynamic network model classes exploiting the Generalized Poisson distribution to capture both under- and overdispersion. We consider three different dynamic specifications: latent factor dynamics, autoregressive dynamics, and latent position dynamics, and study some theoretical properties of the random networks, showing the impact of the dispersion parameter on the random network's connectivity. After discussing the parameter identification strategy, we present a Bayesian inference procedure along with a posterior sampling algorithm. A numerical illustration demonstrates the effectiveness of the designed algorithm and provides estimates of the misspecification bias when unequal dispersion is neglected. Our new models are then applied to two relevant dynamic datasets considered in previous studies: a set of bike-sharing dynamic networks and a set of dynamic media networks. Our results highlight the importance of explicitly modeling overdispersion for both an accurate in-sample fit and out-of-sample performance.
\end{abstract}

\noindent%
{\it Keywords:} Overdispersion, Concentration Inequalities, Bayesian Inference, Latent Space, Identifiability. 

\vfill

\newpage

\section{Introduction}

Temporal networks have attracted considerable attention across a wide range of disciplines \citep{holme2012temporal}, including biology \citep{pastor2015}, neuroscience \citep{betzel2017generative,vavsa2022null}, economics \citep{jackson2002evolution, friel2016interlocking}, and the social sciences \citep{barbera2015birds, casarin2025media}. The distinctive feature of temporal networks is that connections between nodes are not static, but change over time.

In recent years, an increasing number of studies have focused on count-weighted temporal networks, that is, dynamic networks in which the weight associated with each edge represents a count \citep{sewell2016latent}. Examples include online media interaction networks \citep{casarin2025media}, email communication networks \citep{yin2017local}, transportation networks \citep{he2025semiparametric}, and brain connectivity networks \citep{zhang2019nonparametric}. It is well established that count-weighted network data frequently exhibit overdispersion \citep[e.g., see][]{zheng2006many, corsini2022dealing} and underdispersion \citep{lord2010statistical,lux2020distribution}. Nevertheless, many modeling approaches often overlook this feature, which can lead to biased estimates and misleading inferences.

The contribution of this work is manifold. First, we propose a flexible model for weighted networks with integer-valued weights based on the Generalized Poisson (GP) distribution family and derive theoretical properties of the GP network model, including expected strength and node centrality. Second, we introduce three dynamic specifications for our Generalized Poisson network model. Third, a Bayesian inference framework is provided together with an efficient posterior approximation procedure. The proposed approximation procedure is evaluated through an extensive simulation study, followed by two empirical applications. The first application focuses on bike-sharing data from New York City \citep{citibike2019}, while the second examines the evolution of media interaction networks \citep{casarin2025media}. The analysis of these data can have a significant impact across many fields and aspects of society.

Regarding the weight distribution, we assume a GP. This family, introduced by \cite{ConJai1973,consul1973some}, can capture both over- and underdispersion, excess kurtosis, and includes the Poisson distribution as a special case. See \citet[][Ch. 9]{ConFam2006} and \cite{FamLee2020} for an introduction. The GP distribution is well studied in the literature and has found applications across diverse fields, including health and epidemiology \citep{Zam2016}, sport statistics \citep{ShaMoy2016}, as well as economics, finance and insurance \citep{WanFam1997,Ambagaspitiya_Balakrishnan_1994,FemCon1995,FamEtal2004,Lin2004,finner2015some}. 

GP has been successfully used in count time-series analysis \citep{zhu2012modeling,ShaMoy2014, CheLee2016,CARALLO20241359} and naturally emerges in random graph theory from branching processes \citep[e.g., see][]{aldous1998tree,aldous2004tractable,bertoin2012fires}. Unlike other distribution families that allow both overdispersion and underdispersion, it preserves tractability in deriving relevant random network properties. Other popular models for count data, such as the Negative Binomial, Poisson-gamma, and Poisson-lognormal, have limitations: they allow only overdispersion and struggle to handle low sample means and small samples. The Conway-Maxwell-Poisson is also appealing, since it includes the Poisson as a special case, belongs to the exponential family, and admits conjugate distributions \citep{boatwright2006conjugate, shmueli2005useful}. However, its moments are not tractable, which prevents its use in the analytical study of the network model properties. 

Regarding the dynamic specification of our model, we propose three alternative formulations that capture temporal dependence in distinct ways. The first specification introduces a common dynamic latent factor that affects all edges simultaneously \citep[e.g.,][]{BrauningKoopman2020}. The second specification is more parsimonious and incorporates a lagged measure of network strength. In this setting, past global network features influence current link intensities, producing autoregressive dynamics in the network \citep[e.g.,][]{tsikerdekis2021network,jiang2023autoregressive}. Finally, the third specification assumes a Latent Space (LS) model \citep{hoff2002latent} with time-varying latent node coordinates that drive the connectivity dynamics \citep{sewell2016latent, rastelli2016properties, artico2023dynamic, casarin2025media}. Some theoretical properties of the models, such as node centrality, are derived by building on concentration inequalities \citep{vershynin2018high}, which show how the dispersion parameter affects the random network's connectivity.

We adopt a Bayesian inference framework, as it provides greater flexibility for working with nonlinear and latent-variable models via the data-augmentation principle and the complete-data likelihood \citep{robert2007bayesian} and accommodates uncertainty in the prediction (\citealp{mccabe2011efficient}).
We provide sufficient conditions for latent variable identifiability and an efficient Markov Chain Monte Carlo sampler for the posterior distribution \citep{RobCas2013}.

The Citibike \citep{citibike2019} and Media network datasets \citep{schmidt2018polarization} provide two examples of unequal dispersion and edge persistence in temporal networks. These two applications show that GP network models fit better than Poisson models and capture salient global and local features of the networks considered.

The work is organized as follows. Section \ref{sec:stylized} presents the Citibike and Media Network datasets, and illustrates the dispersion features in the network weights. Section \ref{sec:model} introduces the proposed Generalized Poisson network model, provides some model properties, and a Bayesian inference procedure. Section \ref{sec:sim} presents the results of a simulation study. Section \ref{sec:ill} illustrates the model performance on the two network datasets. Finally, Section 5 concludes. 

\section{Overdispersion in Dynamic Networks}
\label{sec:stylized}
We consider two datasets relevant to many fields and with impact on society: the Citibike dataset \citep{citibike2019} and the Media network dataset. \citep{schmidt2018polarization}. While bike-sharing data, as the former, are of interest in transportation \citep{noland2016bikeshare}, geography \citep{an2019weather}, urban planning \citep{yu2018cost}, and sustainability \citep{chen2022environmental}, 
media networks, as the latter, are relevant to media and communication studies \citep{xu2022evolution}, political science \citep{eady2025news}, and computational social science \citep{del2017mapping, cinelli2021echo}. Both datasets have recently attracted the attention of the network statistical modeling literature \citep[e.g., see][]{he2025semiparametric, casarin2025media}. The Citibike dataset contains information on rides between any two Citibike stations in New York City and allows obtaining count-weighted symmetric temporal networks with nodes representing the neighborhoods of New York City and edges representing the count of bike-sharing connections between neighborhoods. The Media network dataset contains information on the Facebook activity of national and local news outlets in France, Germany, Italy, and Spain, and allows obtaining temporal networks in which nodes represent the news outlets' pages and edge weights denote the number of unique users who commented on a pair of outlets within a given time interval. 

The left panel in Figure \ref{fig:logmean_logvar} reports the comparison between log-mean and log-variance for each edge of the Citibike and Media network dataset for Italy for two sub-periods. Compared with the 45-degree line, we observe some heterogeneity: a few edge weights exhibit underdispersion, while most exhibit overdispersion. Moreover, there is evidence of a relationship between log-dispersion and autocorrelation at lag 1 for the edges (right panel). This behavior is consistent across sub-periods. The preliminary results indicate the need to use a conditional distribution for network edges that accommodates unequal dispersion and dynamic features. We report in the Supplement (Section \ref{app:emp}) further results which corroborate this evidence, not only for the Citybike and Italian Media network datasets but also for the Media networks of France, Germany, and Spain.

\begin{figure}[t]
\centering
\resizebox{0.85\textwidth}{!}{
\begin{tabular}{cc}
  \includegraphics[width=0.48\textwidth]{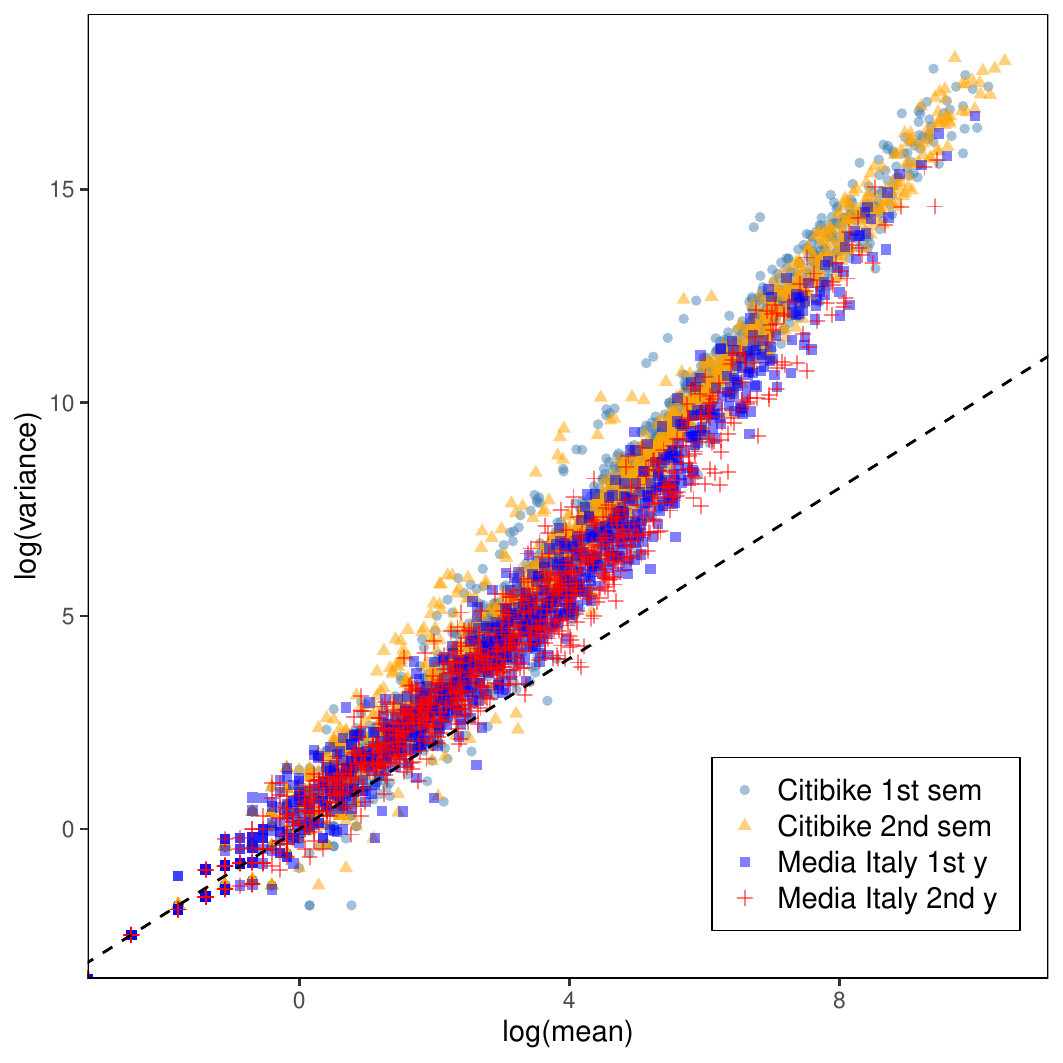} &
  \includegraphics[width=0.48\textwidth]{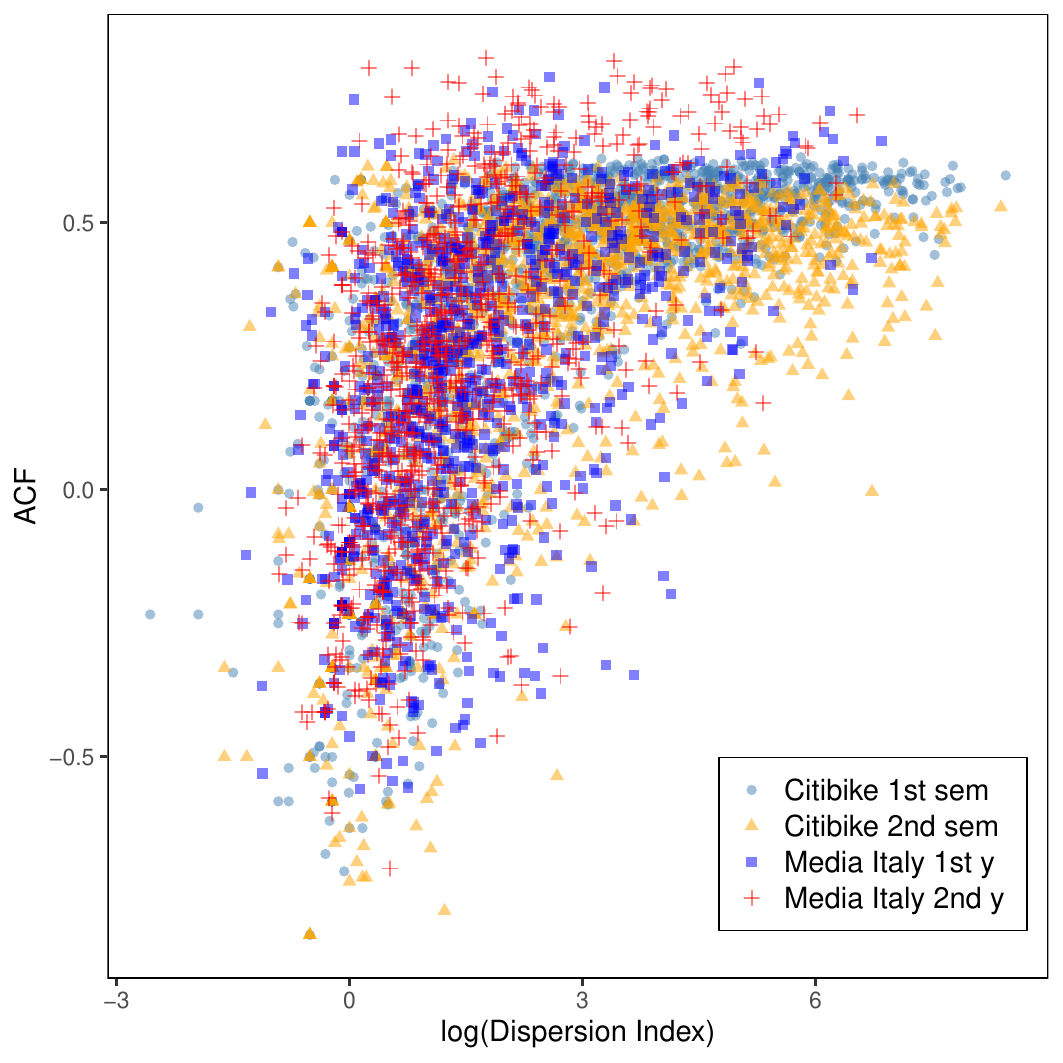}
\end{tabular}}
\caption{Scatter plot of the logarithm of the mean vs. the logarithm of the variance across time for each edge in the dataset after splitting the dataset into two sub-periods: semester~1 and semester~2 for the Citi Bike dataset, and year~1 and year~2 for the Italian media network, right panel. The 45-degree line indicates equal dispersion. Scatter plot of the logarithm of the dispersion index vs. autocorrelation at lag 1 for each edge in the dataset after splitting the dataset into the same two sub-periods.
}
\label{fig:logmean_logvar}
\end{figure}

\section{GP Dynamic Network Models}
\label{sec:model}
The GP is defined, and some properties of the GP networks are derived. GP model classes are introduced together with a suitable inference procedure. The proofs of the results in this section are in the Supplement (Section \ref{app:proofs}).

\subsection{Generalized Poisson networks}
Let $\mathcal{G}=\{\mathcal{G}_{t},\, t=1,2,\ldots, T\}$ be an undirected and count--weighted temporal network with $\mathcal{G}_{t}=(V,E_{t}, Y_{t})$ where $V$ denotes the vertex set, $E_{t}$ the edge set and $Y_t = \{Y_{ijt}\}_{i,j= 1}^N$ the corresponding $N \times N$ count--weighted adjacency matrix at time $t$. We consider a GP specification for modeling the count of interactions $Y_{ijt}$ between nodes $i$ and $j$ at time $t$, and we assume conditional independence given a set of latent parameters. The conditional distribution of $Y_{ijt}\in\mathbb{N}$ follows the GP:
$$
p(y_{ijt} \mid \lambda_{ijt}, \theta)
= 
\frac{\lambda_{ijt} \left( \lambda_{ijt} + \theta y_{ijt} \right)^{y_{ijt}-1} \exp(-(\lambda_{ijt} + \theta y_{ijt}))}{y_{ijt}!},
\,\,
y_{ijt} = 0,1,2,\ldots,
$$
where $\lambda_{ijt} > 0$ governs the mean intensity of the process and $\theta \in (-1,1)$ controls the degree of dispersion \citep{ConJai1973}. We write $Y_{ijt}\sim\mathcal{GP}(\lambda_{ijt},\theta)$. For $\theta = 0$ the model reduces to the standard Poisson distribution, while $\theta > 0$ ($\theta < 0$) allows for overdispersion (underdispersion). Let $\mathcal{F}_{t}=\sigma(\{Y_{ijs},i=1,\ldots,N,j=1,\ldots,N,j\neq i, s=1,\ldots,t\})$ denote the $\sigma$-algebra generated by the collection of edge weights up to time $t$, with $Y_{ijt}\sim\mathcal{GP}(\lambda_{ijt},\theta)$. Following \cite{Ambagaspitiya_Balakrishnan_1994} the conditional moment generating function (mgf) of $Y_{ijt}$, defined as $M_t(u)=\mathbb{E}(Y_{ijt}|\mathcal{F}_{t-1})$, is given by 
\begin{equation}
M_t(u) = \exp(-\frac{\lambda_{ijt}}{\theta}\left(W(z(u))+\theta\right))\label{mgf}
\end{equation}
$u < u_0(\theta)$, where $ z(u) = -\theta \exp(u-\theta)$, and $W(z)$ is the Lambert's W function \citep{mezo2022lambert}. The conditional expected value and variance of $Y_{ijt}$ are given by
\[
\mathbb{E}[Y_{ijt}|\mathcal{F}_{t-1}] = \frac{\lambda_{ijt}}{1 - \theta}, 
\qquad 
\mathbb{V}(Y_{ijt}|\mathcal{F}_{t-1}) = \frac{\lambda_{ijt}}{(1 - \theta)^3}.
\]
Since the probability of a missing edge between two nodes $\mathbb{P}(\{Y_{ijt}=0\}|\mathcal{F}_{t-1})=1-\exp(-\lambda)$ does not depend on $\theta$, then the unweighted random graph, where there is an edge between $(i,j)$ if and only if $Y_{ijt}>0$, has the same properties as an Erd\"os-R\'enyi graph with edge probability $p=1-\exp(-\lambda)$.  The dispersion parameter $\theta$ plays a role when weights are used in graph metrics. For instance, since the GP satisfies the convolution property \citep[][th. 9.1]{Con1989}, the total strength of the network follows a GP with dispersion parameter $\theta$.

\begin{proposition}[Strength]\label{prop1}
Let $\bar{Y}_t= \sum_{i=1}^{N}\sum_{j\neq i }^N Y_{ijt}$ be the network total strength, with $Y_{ijt}\sim\mathcal{GP}(\lambda_{ijt},\theta)$ $i=1,\ldots,N$ $j\neq i$ independently. Then $\bar{Y}_t\sim\mathcal{GP}(\bar{\lambda}_t,\theta)$ with $\bar{\lambda}_t=\sum_{i=1}^{N}\sum_{j\neq i}^N\lambda_{ijt}$ and the expected total strength and its variance are $\bar{\lambda}_t/(1-\theta)$ and $\bar{\lambda}_t/(1-\theta)^3$.
\end{proposition}
From the previous proposition, one can deduce that larger values of $\theta$ provide an increase in the expected strength and consequently in the average node centrality. In the underdispersion setting, that is $\theta\in(-1,0)$, the node strength is smaller than in the oversdisperion case $\theta\in(0,1)$. Also, since the GP is a sub-exponential distribution, Bernstein concentration inequalities can be applied to study how indirect global connectivity is affected by dispersion features \citep{vershynin2018high}. 
\begin{definition}(Sub--exponential)
A random variable $Y$ is sub--exponential if its tails satisfy $\mathbb{P}(\{|X|\geq t\})\leq 2\exp(-t/K)$ for some constant $K$ and for all $t\geq 0$.
\end{definition}
The following proposition states the sub-exponential property of a GP variable and provides a useful bound for the logarithm of the moment generating function. 
\begin{lemma}[Sub--exponential property]\label{propsubexp}
Let $Y\sim\mathcal{GP}(\lambda,\theta)$ then 
\leavevmode\vspace{-0.5\baselineskip}
\begin{enumerate}[i)$\,\,$]
    \item $Y$ is sub-exponential.
    \item \( \log \mathbb{E}(\exp(u(Y-\mathbb{E}(Y)))\le v u^2/(2(1 - b|u|))\), for $|u|\leq r$ where $v = \lambda(1-\theta)^{-3},\quad b = |2\theta+1|(3(1-\theta)^{2})^{-1}
+ r B_r(\theta,\lambda)(1-\theta)^3 (12\lambda)^{-1}$.
\end{enumerate}
\end{lemma}

\par\noindent Before stating the main result on the concentration for the spectral radius $\varrho(Y_t)$ of the random matrix $Y_t$, we introduce some notation and a preliminary result. Denote with  $||X||_{op}=\sup\{||Xz||_2:||z||_2=1\}$ the operator norm of a $(N\times N)$ matrix $X$, with $||X||_2$ the square norm, and with $||X||_{\Psi_1}=\inf\{u>0|\mathbb{E}(\exp(||X||_{op}/u))\leq 2\}$ the Orlicz norm. Note that if $X$ is symmetric then $||X||_{op}=\max\{|\varrho_{i}(X)|,\,i=1,\ldots,N\}=\varrho(X)$, where $\varrho_{i}(X)$ are the eigenvalues of $X$ and $\varrho(X)$ its the spectral radius. The following propositions establish some bounds for the Orlicz norm $||\cdot||_{\psi_1}$ of GP variables.
\begin{proposition}\label{propB}
Let $Y\sim\mathcal{GP}(\lambda,\theta)$, then there exists a constant $C$ such that $||Y-\mathbb{E}(Y)||_{\psi_1}\leq C (\sqrt{v}+b)$ where $v$ and $b$ have been defined in Lemma \ref{propsubexp}.
\end{proposition}

\begin{proposition}[Node Centrality]\label{prop2}
Let $\Lambda_t$ and $Y_t$ be the $N\times N$ matrices with elements $\lambda_{ijt}/(1-\theta)$ and $Y_{ijt}$, respectively and denote with $\varrho(X)$ the spectral radius of a matrix $X$. There exist universal constants $C$ and $c$ such that 
\begin{equation}
|\varrho(Y_t)-\varrho(\Lambda_t)|\leq C(\sqrt{v_t^2 \log N}+K_t \log N)
\end{equation}
with high probability, larger than $1-2 N^{-c}$, where $v_t^2 =\underset{i}{\max} \sum_{j=1,\,j\neq i}^N\lambda_{ijt}/(1-\theta)^3$ and $||Y_{t}||_{\psi_1}<K_t$.
\end{proposition}
This bound says that the spectral radius of the random matrix $Y_{t}$ is sharply concentrated around its expectation’s eigenvalue $\varrho(\Lambda_t)$, implying that some graph properties, such as diffusion rate and eigenvector centrality, can be predicted using the expected connectivity $\Lambda_t$. In addition if $\varrho(\Lambda_t)>> \sqrt{v_t^2 \log N}+K_t \log N$ than $\varrho(Y_t)/\varrho(\Lambda_t)\overset{p}{\rightarrow} 1$ as $N\rightarrow+\infty$. The larger the value of $\theta$, the larger are the expected spectral radius and the absolute fluctuation bound. An increase in $\varrho(\Lambda_t)$ implies highly connected nodes become more central, and every node’s potential to propagate influence (in the eigenvector sense) increases.

The overdispersion parameter is not only central to the properties of the random graph model at each point in time, but also crucial for making inferences about network dynamic properties. As shown in the following sections, assuming equal dispersion, $\theta=0$, a priori, can produce significant misspecification bias. 

\subsection{Dynamic specifications}
For a better interpretability, we reparametrize the Generalized Poisson model in terms of location and dispersion parameters, that is $\mu_{ijt} = \mathbb{E}[Y_{ijt}]$ and $\rho= \mathbb{V}(Y_{ijt})/\mathbb{E}[Y_{ijt}] = (1-\theta)^{-2}$ thus $\lambda_{ijt}=\mu_{ijt}\rho^{-1/2}$ and $\theta=1-\rho^{-1/2}$ \citep[e.g., see][]{CARALLO20241359}. Since $\rho = (1-\theta)^{-2} \in (\frac{1}{4}, \infty)$ for $\theta \in (-1,1)$, we reparametrize further the model and assume the dispersion ratio to be guided by the parameter $\zeta$, i.e. $\rho = \tfrac{1}{4} + \exp(\zeta)$.

The GP parameter is then given by $\theta = 1 - 1/\sqrt{\rho}$, so that overdispersion relative to the Poisson model ($\theta = 0$) corresponds to $\zeta > \log(3/4)$. When the goal is to capture the dynamics in temporal network data, additional features can be incorporated to improve predictions. In the following, we introduce three model classes.

In the first model class, denoted $\mathcal{M}_1$, latent factors are used to explain the formation of edge weights. As an example, consider a dynamic factor $f_{ijt}$ in the log-mean parameter:
$$
\log \mu_{ijt} = \alpha_i + \alpha_j + f_{ijt},
$$
where $\alpha_i$ and $\alpha_j$ are individual-specific node effects that account for systematic heterogeneity in nodes’ propensity to form ties -- i.e., baseline differences in connectivity that translate into higher or lower expected degree. These effects capture node centrality, and are often interpreted as node popularity in the network literature, since larger values correspond to nodes that are more likely to attract links across the network \citep[e.g., see][]{friel2016interlocking}. A common latent trend is useful because many networks are driven not only by node-specific heterogeneity and pairwise structure, but also by system-wide forces that affect all nodes simultaneously (e.g., overall growth or contraction in activity, macroeconomic conditions, regulatory changes, platform-wide shocks). If such global variation is ignored, the model may attempt to absorb broad changes in link density into the node effects $\alpha_i$ or the latent positions, thereby complicating interpretation and potentially biasing inference. Introducing a common time component separates global intensity from relational structure, improving fit and identifiability.

To capture time variation shared across the whole network, we therefore introduce a common latent local--level component $f_t$ in the link-formation process. This term represents unobserved time-specific conditions that increase or decrease the propensity to form ties for all node pairs. Accordingly, we assume $f_{ijt}=f_t$ for all $i,j$, so that pairwise link probabilities co-move over time through a single scalar factor. We model $f_t$ as a random-walk (local--level) process,
\[
f_t = f_{t-1} + \varepsilon_t,\qquad \varepsilon_t \sim \mathcal{N}(0,\sigma_\varepsilon^2),
\]
which yields a flexible evolution over time: the level can drift gradually, while $\sigma_\varepsilon^2 > 0$ controls the smoothness of the trend (small values imply a slowly varying baseline density, whereas larger values allow sharper shifts). Models in this class require a careful inference procedure with a non-negligible computational cost. This specification relates to state-space models commonly used in time-series analysis \citep[e.g.,][]{xing2010statespace, mazzarisi2020dynamic, BrauningKoopman2020, BuccheriMazzarisi2024}.

In the second model class, denoted $\mathcal{M}_2$, past edge weights affect the edge-formation process. We consider a parsimonious autoregressive formulation with $f_{ijt}=f_t$, where
$f_t = \sum_{\ell=1}^{k}\delta_\ell \log(\tilde{y}_{t-\ell}),\,\, \tilde{y}_{t-\ell} = \frac{\bar{S}_{t-\ell}}{(N-1)}
$
with $\delta_\ell$, $\ell=1,\ldots,k$ autoregressive coefficients and $\bar{S}_{t-\ell} = N^{-1}\sum_{i=1}^{N}\sum_{j\neq i}^N y_{ij,t-\ell}$ is the average network strength at time $t-\ell$, $\ell=1,\ldots,k$. This specification shares points with time-varying ERGMs, in which global network features are employed to predict the network's average behavior \citep[e.g., see][]{shang2009exponential,tsikerdekis2021network,jiang2023autoregressive}. The framework can be readily extended to accommodate general INGARCH specifications \citep[e.g.,][]{zhu2012modeling}, in which lagged $f_{t}$ values enter the dynamics.

The third class of models, denoted by $\mathcal{M}_3$, is obtained by augmenting the edge specification with time-varying node-specific features (or latent coordinates) $\mathbf{x}_{it}\in\mathbb{R}^d$, $i=1,\ldots,N$. In this setting, the pairwise propensity to form a link is allowed to depend on how ``close'' two nodes are in the feature space, namely
\[
f_{ijt} = f_t - d(\mathbf{x}_{it},\mathbf{x}_{jt}),
\]
where $d(\cdot,\cdot)$ is a user-chosen function that quantifies similarity or dissimilarity between $\mathbf{x}_{it}$ and $\mathbf{x}_{jt}$.  The choice of $d(\cdot,\cdot)$ directly controls how differences in node features translate into link formation, and hence the flexibility of the model in capturing latent clustering and other geometric network effects. Different distance measures are used because the notion of ``closeness'' between node features depends on the application and on the type of latent representation one wishes to impose. 

For example, one may adopt a similarity measure such as the inner product $d(\mathbf{x},\mathbf{y})=-\mathbf{x}^\top\mathbf{y}$, $\mathbf{x},\mathbf{y}\in\mathbb{R}^d$, so that more aligned feature vectors increase the propensity to link. Similarity-based measures emphasize alignment in direction and allow for assortative patterns, in which ties are more likely between nodes with concordant feature profiles. Alternatively, one may use a dissimilarity measure, such as those based on Euclidean distances, such as $d(\mathbf{x},\mathbf{y})=\|\mathbf{x}-\mathbf{y}\|_2$ or $d(\mathbf{x},\mathbf{y})=\|\mathbf{x}-\mathbf{y}\|_2^2$ with $\|\cdot\|_2$ denoting the Euclidean norm, so that nodes that are closer in the latent space are more likely to be connected.  Euclidean distances encode homophily by penalizing separation in the latent space, whereas squared distances further strengthen the penalty for large separations and are often computationally convenient (they are smooth everywhere and avoid the singularity at zero distance), whereas the Euclidean norm yields a more direct geometric interpretation in terms of straight-line distances. More generally, the chosen metric determines the geometry of the latent space and thus the class of network structures that the model can represent (e.g., clustering, transitivity, or the presence of hubs). In the following, we focus on the squared Euclidean distance $d(\mathbf{x},\mathbf{y})=||\mathbf{x}-\mathbf{y}||_2^2$, nevertheless our modeling strategy and inference procedure can be readily extended to other distance functions. 

The features can be observable or latent, node-specific or common, with a dynamic specification. In the following, we assume a random walk for the common factor $f_t = f_{t-1} + \varepsilon_t$ with $\varepsilon_t \sim \mathcal{N}(0, \sigma_\varepsilon^2)$, and a random walk specification of the latent positions, that is
\[
\mathbf{x}_{it}=\mathbf{x}_{it-1}+\boldsymbol{\nu}_{it},\, \boldsymbol{\nu}_{it}\overset{iid}{\sim}\mathcal{N}_d(\mathbf{0},\Sigma_x),
\]
with $\Sigma_x = \sigma^2_xI_d$. This allows the network to be represented in a latent space  - an approach introduced in the latent position model of \cite{hoff2002latent}, where nodes are embedded as points in a low-dimensional Euclidean space and tie probabilities depend on distances between points, and later extended by \cite{sewell2016latent} to accommodate time-varying latent coordinates for dynamic networks. Although the modeling strategies discussed above could be further combined, we leave this extension to future research, as our goal is to demonstrate the contribution of GP dispersion to network model properties across classes.

\subsection{Prior distribution and model properties}
We propose a Bayesian inference framework that naturally accounts for parameter uncertainty. Our model choice is to assume independent Gaussian prior distributions on $\alpha_i$ and $\delta_\ell$, that is
\begin{align}
    &\alpha_i \overset{iid}{\sim} \mathcal{N}(0,\sigma^2_\alpha),\, i = 1, \ldots, N,\quad \delta_\ell \overset{iid}{\sim} \mathcal{N}(0,\sigma^2_\delta),\,\ell=1,\ldots, k
\end{align}
as it is standard in relational data models, see \cite{hoff2005bilinear}. Following \cite{CARALLO20241359}, a Gaussian prior distribution is assumed for the dispersion parameter $\zeta$, which induces a weakly informative prior on the dispersion ratio $\rho$ and, consequently, on $\theta$
\begin{align}    
    &\zeta \sim \mathcal{N}(\mu_\zeta, \sigma^2_\zeta).
\end{align}
As is common in random network literature \citep[e.g., see][]{rastelli2016properties, casarin2024dynamic, casarin2025media}, we investigate the impact of node random effects on the conditional strength distribution, its expected value, and dispersion index. Let $\bar{S}_t$ be the network average strength at time $t$ and $\mathcal{F}_{t}^{\ast}=\mathcal{F}_{t}\vee \{f_{s},i=1,\ldots,N,j=1,\ldots,N,j\neq i, s=1,\ldots,t\}$ the augmented sigma algebra. Note that for model class $\mathcal{M}_1$, $f_t,f_{t-1}\in\mathcal{F}_{t-1}$ and for $\mathcal{M}_{2}$ and $\mathcal{M}_3$, $f_t\notin\mathcal{F}^{\ast}_{t-1}$ and $f_{t-1}\in\mathcal{F}^{\ast}_{t-1}$.

\begin{proposition}[Average Strength Distribution]\label{prop3K}
Let $K_t(u)=\log\mathbb{E}(\exp(u\bar{S}_t)|\boldsymbol{\alpha}\vee\mathcal{F}_{t-1}^{\ast})$ be the conditional cumulant function and $\kappa_t^{(j)}=\dfrac{d^j}{du^j}K_t(u)|_{u=0}$ the random cumulants. The conditional moments of $\kappa_t^{(j)}$ are
\[
\mathbb{E}((\kappa_t^{(j)})^m|\mathcal{F}_{t-1}^{\ast})=
\begin{cases}
\exp(m^2\frac{\sigma^2_\alpha}{2}+m f_{t})(\zeta^{(j)})^m&\mathrm{if}\,f_{t}\in\mathcal{F}_{t-1}\\
\exp(m^2\frac{\sigma^2_\alpha}{2}+m f_{t-1}+m^2\frac{\sigma^2_\epsilon}{2})(\zeta^{(j)})^m&\mathrm{if}\,f_{t}\notin\mathcal{F}_{t-1}^{\ast}\\
\end{cases}
\]
for $m\in\mathbb{N}$, where $  \zeta^{(j)}
  = \frac{1}{\theta}
    \sum_{\ell=1}^{j}
      S(j,\ell)\,(-\theta)^{\ell}
      \frac{\,p_{\ell}(-\theta)}{(1-\theta)^{2\ell-1}}$ with $S(j,\ell)$ the Stirling's number of the second kind and $p_\ell(x)=(1+x)p'_\ell(x)-(nx+3n-1)p_\ell(x), \ell\geq1$ and $p_1(x)=1$.
\end{proposition}
Proposition \ref{prop3K} allows deriving the following properties for $\mathcal{M}_1$ and $\mathcal{M}_2$; similar properties can be derived for $\mathcal{M}_3$. 
\begin{proposition}[Average Strength]\label{prop3}
The conditional expected average strength $\mathbb{E}(\bar{S}_t|\mathcal{F}_{t-1}^{\ast})$ and the dispersion index, $D_t=\mathbb{V}(\bar{S}_t|\mathcal{F}_{t-1}^{\ast})/\mathbb{E}(\bar{S}_t|\mathcal{F}_{t-1}^{\ast})$,  are
\[
\mathbb{E}(\bar{S}_t|\mathcal{F}_{t-1}^{\ast})=
\begin{cases}
\frac{1}{1-\theta}\gamma\exp(f_{t})&\mathrm{if}\,f_{t}\in\mathcal{F}_{t-1}^{\ast}\\
\frac{1}{1-\theta}\gamma\exp(f_{t-1}+\frac{\sigma_{\epsilon}^2}{2})&\mathrm{if}\,f_{t}\notin\mathcal{F}_{t-1}^{\ast},
\end{cases}
\]
where $\gamma=\exp(\sigma^2_\alpha)$ and
\[
D_t=\begin{cases}
\frac{1}{(1-\theta)^2}+2\gamma(\gamma-1)(2N+\gamma-3)\exp(f_t)&\mathrm{if}\,f_{t}\in\mathcal{F}_{t-1}^{\ast}\\
\frac{1}{(1-\theta)^2}+2\gamma(\gamma-1)(2N+\gamma-3)\exp(f_{t-1}+\sigma^2_{\epsilon})&\mathrm{if}\,f_{t}\notin\mathcal{F}_{t-1}^{\ast}.
\end{cases}
\]
\end{proposition}
The average strength can also be derived for model $\mathcal{M}_3$ following similar arguments as in the proof of Prop . 1 of \cite{casarin2024dynamic} and of Proposition 2.1 and Cor. 2.1 of \cite{casarin2025media}.

The individual effect variance, the factor level, and the variance of the factor level can all contribute to overdispersion. As $\sigma^2_\alpha$ increases, one obtains larger overdispersion for $\theta\in(0,1)$ and reduced overdispersion for $\theta\in(-1,0)$. One obtains equal dispersion when $\theta=0$ and $\sigma^2_\alpha$ goes to zero. If $\theta\in(-1,0)$, the model can also account for underdispersion as illustrated in Figure \ref{fig:theta}. In the left plot, each line shows the values of $\theta$ and $\sigma_\alpha^2$ such that the dispersion index is 0.5 (solid lines) or 1.5 (dashed lines).  Also, negative values of $\theta\in(-1,0)$ are compatible with overdispersion (dashed lines) for larger values of the individual effect variance for different network sizes (left) and factor levels (right). The left plot also shows that underdispersion cannot always be achieved if the value of $\sigma^2_\alpha$ is too large, since one needs $\theta>-1$ to have a well-defined GP distribution. The right plot shows that for larger values of $\sigma^2_\alpha$, $\theta$ must be small to guarantee a given level of average strength across different levels of the factor.

\begin{figure}[t]
    \centering
    \setlength{\tabcolsep}{-3pt}
    \resizebox{0.8\textwidth}{!}{
    \begin{tabular}{cc}
    \includegraphics[width=0.45\linewidth]{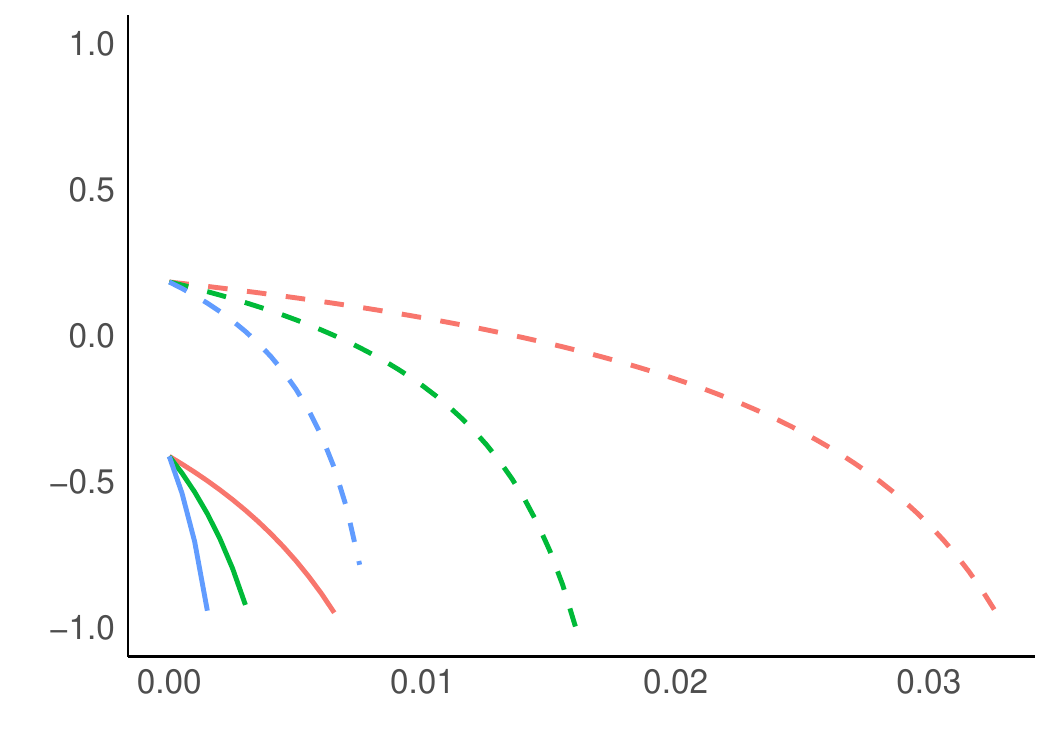}&
    \includegraphics[width=0.45\linewidth]{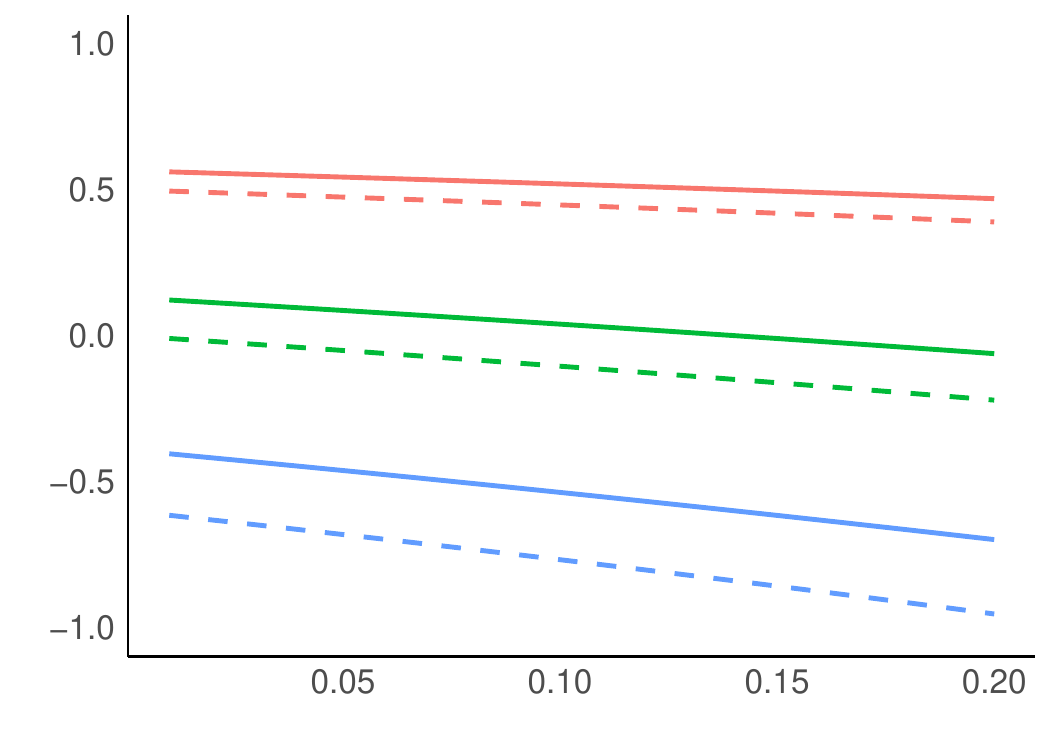}
    \end{tabular}}
    \caption{\textbf{Sensitivty Analysis:} Value of $\theta\in(-1,1)$ (vertical axis) as a function of $\sigma^2_{\alpha}\in(0,0.2)$ (horizontal axis). Left: for a given dispersion index value of 0.5 (\usebox{\lineBlack}) and 1.5 (\usebox{\lineBlackDash}) and network size $N=10$ (\usebox{\lineRed},\usebox{\lineRedDash}), $N=20$ (\usebox{\lineGreen},\usebox{\lineGreenDash}) and $N=22$ (\usebox{\lineBlue},\usebox{\lineBlueDash}). Right: for a given average expected strength of 10 (\usebox{\lineBlack}) and 20 (\usebox{\lineBlackDash}) and factor $f_t=\log(10)$ (\usebox{\lineRed},\usebox{\lineRedDash}), $f_t=\log(20)$ (\usebox{\lineGreen},\usebox{\lineGreenDash}) and $f_t=\log(32)$ (\usebox{\lineBlue},\usebox{\lineBlueDash}).}
    \label{fig:theta}
\end{figure}

\subsection{Parameters Identifiability}
We note that the three proposed models face the same identification issues as other models in the same class. As concerns $\mathcal{M}_1$, the likelihood remains invariant when we sum and subtract the same quantity from $\mu_{ijt}$. This affects the identification of $\alpha_i$ for $i = 1, \ldots, N$ and $f_t$ for $t = 1, \ldots, T$.  This issue is \emph{per-se} mitigated by choosing a prior centered at zero for either $\alpha_i$ or $f_t$, but can be solved by imposing a zero-sum restriction for either of the two sets of parameters. In the following result, we show how a zero-sum restriction on the vector of parameters $\boldsymbol{\alpha}$ is sufficient to guarantee identifiability. 

\begin{proposition}[Identifiability  of $\mathcal{M}_1$]\label{lem:idm1}
Let 
$\log(M(\boldsymbol{\alpha},\mathbf{f}))$ the entry-wise logarithm of the tensor $M(\boldsymbol{\alpha},\mathbf{f})\in\mathbb{R}^{N\times N\times T}$, with frontal slices $\log(M(\boldsymbol{\alpha},\mathbf{f}))_{::t}=\boldsymbol{\alpha}\iota_N'+\iota_N\boldsymbol{\alpha}'+f_t \, \iota_N \iota_N',$
$t=1,\dots,T$, where $\boldsymbol{\alpha} \in \mathbb{R}^N$, $\mathbf{f} \in \mathbb{R}^T$. Assume the following: A1) $J_N \boldsymbol{\alpha}=\boldsymbol{\alpha}$ and 
$J_N \tilde{\boldsymbol{\alpha}}=\tilde{\boldsymbol{\alpha}}$, $\forall\boldsymbol{\alpha},\tilde{\boldsymbol{\alpha}}
\in \mathbb{R}^N$,
where $J_N = I_N - \frac{1}{N}\iota_N \iota_N'$. Then the model $\mathcal{M}_1$ is identifiable, that is,
$$
\log(M(\boldsymbol{\alpha},\mathbf{f}))
=
\log(M(\tilde{\boldsymbol{\alpha}},\tilde{\mathbf{f}}))
\,\Longrightarrow\,
(\boldsymbol{\alpha},\mathbf{f})
=
(\tilde{\boldsymbol{\alpha}},\tilde{\mathbf{f}}), \mathbf{f}\in\mathbb{R}^T, \boldsymbol{\alpha},\tilde{\boldsymbol{\alpha}}
\in \mathbb{R}^N.
$$
\end{proposition}

Regarding $\mathcal{M}_2$, we notice that a further assumption is required on the design matrix $W =(\boldsymbol{w}_1',\ldots,\boldsymbol{w}_T')' \in \mathbb{R}^{T\times k}$ with $\boldsymbol{w}_t = (\tilde{y}_{t-1}, \ldots, \tilde{y}_{t-k})$. One needs $W$ to be pre-determined and full-ranked.

\begin{proposition}[Identifiability of $\mathcal{M}_2$] \label{lem:idm2}
Define the tensor 
$\log(M(\boldsymbol{\alpha},\boldsymbol{\delta}))\in\mathbb{R}^{N\times N\times T}$ with frontal slices $\log(M(\boldsymbol{\alpha},\boldsymbol{\delta}))_{::t}
=
\boldsymbol{\alpha}\iota_N'
+
\iota_N\boldsymbol{\alpha}'
+
(\boldsymbol{w}_t' \boldsymbol{\delta}) \, \iota_N \iota_N',
$, $t=1,\dots,T$, with $\boldsymbol{\alpha} \in \mathbb{R}^N$, $\boldsymbol{\delta} \in \mathbb{R}^k$,  with $W =(\boldsymbol{w}_1',\ldots,\boldsymbol{w}_T')' \in \mathbb{R}^{T\times k}$. Assume: A1) $J_N \boldsymbol{\alpha}=\boldsymbol{\alpha}$ and 
$J_N \tilde{\boldsymbol{\alpha}}=\tilde{\boldsymbol{\alpha}}$, 
where $J_N = I_N - \frac{1}{N}\iota_N \iota_N'$, $\forall \boldsymbol{\alpha},\tilde{\boldsymbol{\alpha}}\in\mathbb{R}^N$; A2) The matrix 
$W$
is predetermined, $T\ge k$, and $\operatorname{rank}(W)=k$. Then the model $\mathcal{M}_2$ is identifiable, that is,
$$
\log(M(\boldsymbol{\alpha},\boldsymbol{\delta}))
=
\log(M(\tilde{\boldsymbol{\alpha}},\tilde{\boldsymbol{\delta}}))
\,\,\Longrightarrow\,\,
(\boldsymbol{\alpha},\boldsymbol{\delta})
=
(\tilde{\boldsymbol{\alpha}},\tilde{\boldsymbol{\delta}}), \boldsymbol{\alpha},\tilde{\boldsymbol{\alpha}}\in\mathbb{R}^{N}, \boldsymbol{\delta},\tilde{\boldsymbol{\delta}}\in\mathbb{R}^k.
$$

\end{proposition}

In addition to the aforementioned issues, $\mathcal{M}_3$ suffers from well-known indeterminacy problems related to the likelihood invariance of the latent coordinates in terms of translation, reflection and rotation (see \citealp{hoff2002latent}), which can be solved via a combination of dimension-specific zero-sum restrictions and with the use of Procrustes transformation as stated in the following.

\begin{proposition}[Identifiability of $\mathcal{M}_3$] \label{lem:idm3}

Define the tensor $\log(M(\boldsymbol{\alpha},\mathbf{f},\{X_t\}_{t =1}^T)) \in \mathbb{R}^{N \times N \times T}$ with frontal slices 
$
\log(M(\boldsymbol{\alpha},\mathbf{f},X_t ))_{::t}
=
\boldsymbol{\alpha}\iota_N'
+ \iota_N\boldsymbol{\alpha}'
+ f_t\iota_N\iota_N'
- \big(\mathbf{g}_t\iota_N' + \iota_N\mathbf{g}_t' - 2X_tX_t'\big),
$ $t = 1, \ldots, T$, with $\boldsymbol{\alpha} \in \mathbb{R}^N$, $\mathbf{f} \in \mathbb{R}^T$, and $\mathbf{g}_t=\operatorname{diag}(X_tX_t')$. Assume the following: A1) $J_T\mathbf{f}=\mathbf{f}$,and $J_T\tilde{\mathbf{f}}=\tilde{\mathbf{f}}$, where $J_T=I_T-\frac{1}{T}\iota_T\iota_T'$; A2) $J_N X_t = X_t$, and $J_N \tilde X_t=\tilde X_t$, $\forall X_t,\tilde{X}_t\in\mathbb{R}^{N\times d}$ $t=1,\dots,T$, where $J_N=I_N-\frac{1}{N}\iota_N\iota_N'$; A3) $\operatorname{rank}(X_t)=\operatorname{rank}(\tilde X_t)=d$ for all $t=1,\dots,T$. Then the model $\mathcal M_3$ is identifiable on the restricted parameter space, that is
$$
\log(M(\boldsymbol{\alpha},\mathbf{f},\{X_t\}))
=
\log(M(\tilde{\boldsymbol{\alpha}},\tilde{\mathbf{f}},\{\tilde X_t\}))
\,\,\Longrightarrow\,\,
X_tX_t'=\tilde X_t\tilde X_t' \, \forall t,
\,\,
(\boldsymbol{\alpha},\mathbf{f})=(\tilde{\boldsymbol{\alpha}},\tilde{\mathbf{f}}),
$$
with $\boldsymbol{\alpha},\tilde{\boldsymbol{\alpha}}\in\mathbb{R}^{N}, \mathbf{f},\tilde{\mathbf{f}}\in\mathbb{R}^T, X_t,\tilde{X}_t\in\mathbb{R}^{N\times d}$. Moreover, under A3, for each $t = 1,\ldots, T$ there exists an orthogonal matrix 
$P_t\in\mathbb{R}^{d\times d}$ with $P_t'P_t=I_d$ such that $\tilde X_t = X_t P_t$.
\end{proposition}

\section{Posterior Approximation}\label{sec:sim}
We derive the full conditional distributions used in the Gibbs sampling posterior approximation and conduct simulation exercises to assess the efficiency and effectiveness of our algorithm. Moreover, we demonstrate numerically that substantial misspecification bias and prediction errors arise when a Poisson model is applied to data with unequal dispersion. 

\subsection{Gibbs sampler}
Since the posterior distribution is not tractable, Bayesian inference is performed through Markov-Chain-Monte Carlo (MCMC) sampling, in particular via Metropolis-within-Gibbs steps \citep[e.g., see][]{chib1995understanding}. Posterior inference is conducted via a Metropolis-within-Gibbs algorithm, which alternates between draws from the full conditional distributions of the model parameters. To simply the notation, we define the following collections of variables and parameters:  $\mathbf{Y} = \{Y_1, \ldots, Y_T \}$ $\mathbf{X} = \{X_1, \ldots, X_T \}$ with $X_t = (\mathbf{x}'_{1t}, \ldots, \mathbf{x}'_{Nt})' \in \mathbb{R}^{N \times d}$, $\boldsymbol{\alpha} = (\alpha_1,\ldots, \alpha_N)$, $\boldsymbol{\delta} = (\delta_1,\ldots, \delta_k)$, and $\mathbf{f} = (f_1, \ldots, f_T)$. At each iteration $h=1,\dots,H$, the following updates are performed.

In the first block of updates, sample the individual effects in $\boldsymbol{\alpha}$ by drawing each $\alpha_i$ from its full conditional distribution:
\begin{equation}
p(\alpha_i \mid \mathbf{Y}, \boldsymbol{\alpha}_{-i}, \boldsymbol{\delta}, \mathbf{f}, \mathbf{X}, \zeta)
\propto 
\mathcal{N}(\alpha_i;0,\sigma_\alpha^2)\prod_{t=1}^T \prod_{j \neq i} \mathcal{GP}(Y_{ijt}\mid\mu_{ijt},\theta(\zeta)),
\end{equation}
using a random-walk Metropolis-Hastings update. 

In the second block of the Gibbs sampler, the parameters and the latent variables of the dynamic components are sampled. Depending on the chosen specification $\mathcal{M}_j$.

\par\noindent 1) In the autoregressive specification $\mathcal{M}_1$, update the autoregressive coefficients $\boldsymbol{\delta}=(\delta_1,\ldots,\delta_K)$ from
    \begin{equation}
    p(\boldsymbol{\delta}\mid \mathbf{Y},\boldsymbol{\alpha},\zeta)
    \propto 
     \mathcal{N}(\boldsymbol{\delta};0,\sigma_\delta^2 I_K)\prod_{t=1}^T\prod_{j > i} \mathcal{GP}(Y_{ijt}\mid\mu_{ijt},\theta(\zeta)),
    \end{equation}
    typically via a multivariate random-walk Metropolis-Hastings step.

\par\noindent 2) In the latent factor specification $\mathcal{M}_2$, sample the latent factor sequence $\mathbf{f}$ from its state-space posterior
    \begin{equation}
    p(\mathbf{f}\mid \mathbf{Y},\boldsymbol{\alpha},\zeta,\sigma_\varepsilon^2)
    \propto \mathcal{N}(f_1;f_{0},\sigma^2_{f_0}) \prod_{t=2}^T  \mathcal{N}(f_t;f_{t-1},\sigma_\varepsilon^2)
    \prod_{j > i} \mathcal{GP}(Y_{ijt}\mid \mu_{ijt}, \theta(\zeta)),
    \end{equation}
    assuming $f_0 = 0$ and $\sigma_{f_0}^2$ large, while using a random-walk Metropolis-Hastings step. Sample $\sigma^2_\epsilon$ exactly from its inverse gamma posterior:   $\sigma^2_{\epsilon}\mid \mathbf{f} \sim \mathcal{IG}(\overline{a}, \overline{b})$, where $\overline{a} = \underline{a} + (T-1)/2$ and $\overline{b} =\underline{b} + \sum_{t = 2}^T(f_{t}-f_{t-1})/2$.

\par\noindent 3) In the dynamic latent position specification $\mathcal{M}_3$, first, draw the latent factor sequence $\mathbf{f}$ from its state-space posterior $p(\mathbf{f}|\mathbf{Y},\boldsymbol{\alpha}, \boldsymbol{X},\zeta,\sigma_\varepsilon^2)$ similarly to what was done in $\mathcal{M}_2$, and secondly draw the latent coordinates $\mathbf{x}_{it}$ $i=1,\ldots,N$ $t = 1,\ldots, T$ from the joint full conditional distribution
\begin{align}
    p(\mathbf{x}_{it}\mid \mathbf{Y},\boldsymbol{\alpha},\zeta, \mathbf{f}, \boldsymbol{X}_{-it})
    &\propto 
      \mathcal{N}(\mathbf{x}_{it};\mathbf{x}_{i,t-1},\Sigma_x)\mathcal{N}(\mathbf{x}_{i,t+1};\mathbf{x}_{i,t},\Sigma_x)^{\mathbb{I}(t<T)} \prod_{j > i}\mathcal{GP}(Y_{ijt}\mid f_{ijt},\theta(\zeta)),\label{fullx}
\end{align}
assuming $\mathbf{x}_{i,0}$ equal to the zero vector. We use a Metropolis-Hastings algorithm with a proposal derived from a log-Taylor expansion of the likelihood.
\begin{proposition}\label{propTaylor}
The log-likelihood function, as a function of $\mathbf{x}_{it}$ can be approximated with a first-order Taylor approximation
$\sum_{i=1}^{N}\sum_{j\neq i}\log p\!\left(y_{ijt}\mid \lambda_{ijt},\theta\right)=\sum_{i=1}^{N}\sum_{j\neq i}\log p\!\left(y_{ijt}\mid \tilde{\mu}_{ijt}\rho^{-1/2},\theta\right)+\sum_{j\neq i}C_{ijt}(\tilde{\mathbf{x}})'(\mathbf{x}_{it}-\tilde{\mathbf{x}})+o\!\left(\|\mathbf{x}_{it}-\tilde{\mathbf{x}}\|\right)$,
where the gradient is
\begin{align}
&C_{ijt}(\tilde{\mathbf{x}}) =2\tilde{\mu}_{ijt} \left(\frac{1}{\tilde{\mu}_{ijt}}
+\frac{(y_{ijt}-1)}{\tilde{\mu}_{ijt}+\rho^{1/2}\theta y_{ijt}}
-\rho^{1/2}\right)(\tilde{\mathbf{x}}- \mathbf{x}_{jt})
\end{align}
with $\tilde{\mu}_{ijt}=\exp(\alpha_i+\alpha_j+f_t-(\tilde{\mathbf{x}}- \mathbf{x}_{jt})'(\tilde{\mathbf{x}}- \mathbf{x}_{jt}))$.
\end{proposition}
\begin{corollary}\label{propTaylorProposal}
The full conditional distribution in Eq. \ref{fullx} can be approximated through the normal distribution $\mathcal{N}_{d}((\mathbf{x}_{i,t-1}+\mathbf{x}_{i,t+1}+\Sigma_x C_{it}(\tilde{\mathbf{x}}))/2,\Sigma_x/2)$ where $C_{it}(\tilde{\mathbf{x}})=\sum_{j\neq i} C_{ijt}(\tilde{\mathbf{x}})$.
\end{corollary}
At a given Gibbs iteration $j$, when sampling $\mathbf{x}_{it}^{(j)}$, the log-Taylor approximation is centered at the previous iteration value of the coordinate, that is $\tilde{\mathbf{x}}=\mathbf{x}_{it}^{(j-1)}$.

In the last block of the Gibbs sampler, the dispersion parameters $(\zeta, \rho, \theta)$ are sampled conditionally given the first and second block, and the other observations. First $\zeta$ is updated (and hence $\rho = \tfrac{1}{4}+\exp(\zeta)$ and $\theta = 1 - 1/\sqrt{\rho}$) by sampling from
\[
p(\zeta\mid \mathbf{Y},\boldsymbol{\alpha},\boldsymbol{\delta},\mathbf{f},\mathbf{X})
\propto 
\mathcal{N}(\zeta;\mu_\zeta,\sigma_\zeta^2)\prod_{t=1}^T\prod_{j > i}\mathcal{GP}(Y_{ijt}\mid \mu_{ijt},\theta(\zeta)).
\]

A further step can be added to provide intensity posterior point forecasts and network distribution forecasts via an MCMC sampling scheme, following \cite{CARALLO20241359} and \cite{CheLee2016}. 

We implement our MCMC algorithm in \texttt{C++}, leveraging the \texttt{Rcpp} package \citep{eddelbuettel2013seamless} for easy integration with the \texttt{R} environment.
For the random number generator and likelihood function of the GP distribution, we adapted the code developed in the \texttt{VGAM} package \citep{yee2015vector}.

\subsection{Approximation efficiency and misspecification bias}

We assess the performance of the proposed estimation procedure through Monte Carlo simulations. Data are generated from three alternative specifications ($\mathcal{M}_1$--$\mathcal{M}_3$) under a Generalized Poisson framework, incorporating dynamic latent factors, autoregressive dynamics, and time-varying latent coordinates. Across all designs, the MCMC algorithm (5,000 iterations, 2,000 burn-in, thinning every 5 draws) shows good mixing and convergence, as confirmed by graphical diagnostics, effective sample sizes, and Geweke tests. Posterior distributions accurately recover the true structural and dispersion parameters, with only minor bias in the most complex latent-coordinate specification. Overall, the results indicate that the proposed methodology reliably recovers the underlying data-generating parameters. See Section \ref{app:sim_res} in the Supplement for full details on the simulation exercise.

Omitting underdispersion or overdispersion features in the model can lead to significant estimation bias in the parameters and imprecise uncertainty quantification. Figure \ref{fig:bias} reports the posterior distribution (light blue) and the true value of some characteristic parameters for $\mathcal{M}_1$, $\mathcal{M}_2$ and $\mathcal{M}_3$ under Poisson (upper panel) and GP specification (bottom panel) when the data generating process assumes a GP likelihood. The bias in the estimates induced by the misspecified Poisson model clearly emerges (top vis-a-vis bottom panels). 

We compare the Deviance Information Criterion \citep{gelman2014understanding} of the correctly specified Generalized Poisson models and the misspecified Poisson models under the three specifications $\mathcal{M}_1$, $\mathcal{M}_2$, and $\mathcal{M}_3$. Correctly specified models exhibit lower DIC values $DIC(\mathcal{M}_1)=120851.9$, $DIC(\mathcal{M}_2)=23774.6$, $DIC(\mathcal{M}_3)=25953.6$ than those of the misspecified models $DIC(\mathcal{M}_1)=336530.6$, $DIC(\mathcal{M}_2)=345649.9$, $DIC(\mathcal{M}_3)=52027.5$. Failing to account for overdispersion leads to higher DIC values.

\begin{figure}[t]
    \centering
    \resizebox{1.00\textwidth}{!}{
    \begin{tabular}{ccc}
     Poisson $\mathcal{M}_1$&  Poisson $\mathcal{M}_2$&  Poisson $\mathcal{M}_3$\\
     \includegraphics[trim={1cm 1.5cm 1cm 1.5cm}, clip, width=0.41\linewidth]{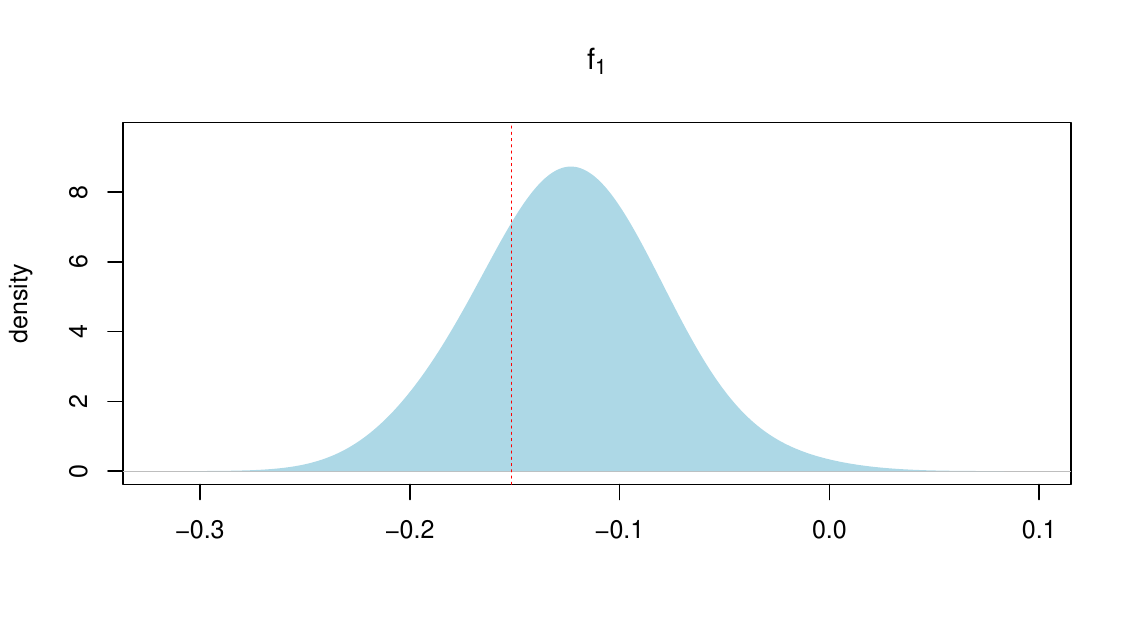} & \includegraphics[trim={1cm 1.5cm 1cm 1.5cm}, clip, width=0.41\linewidth]{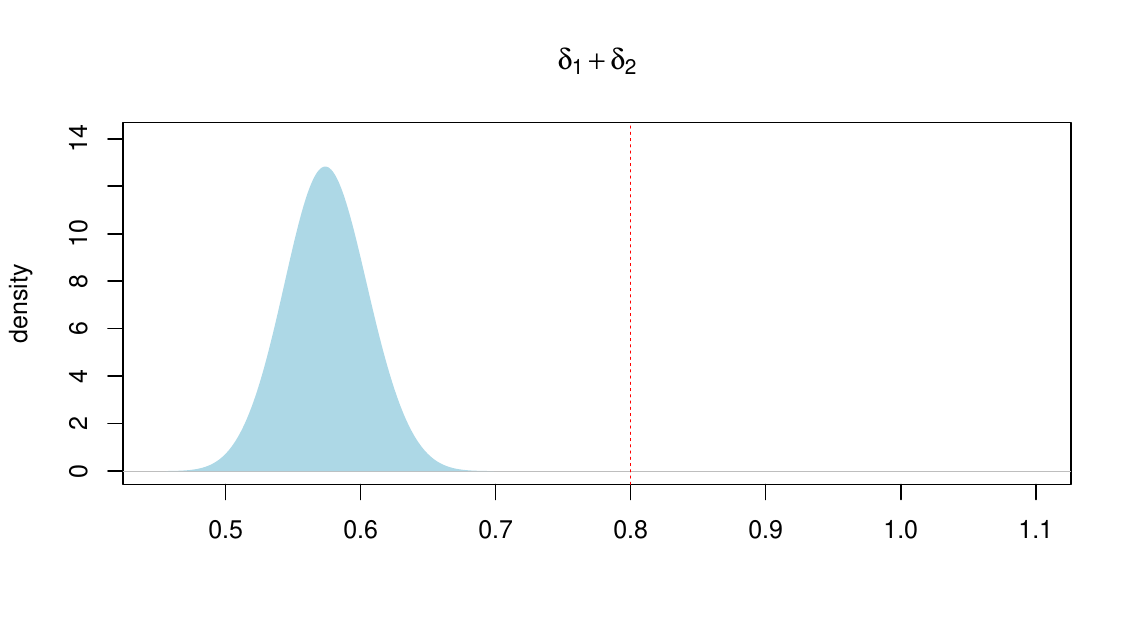}  & 
    \includegraphics[trim={1cm 1.5cm 1cm 1.5cm},clip, width=0.41\linewidth]{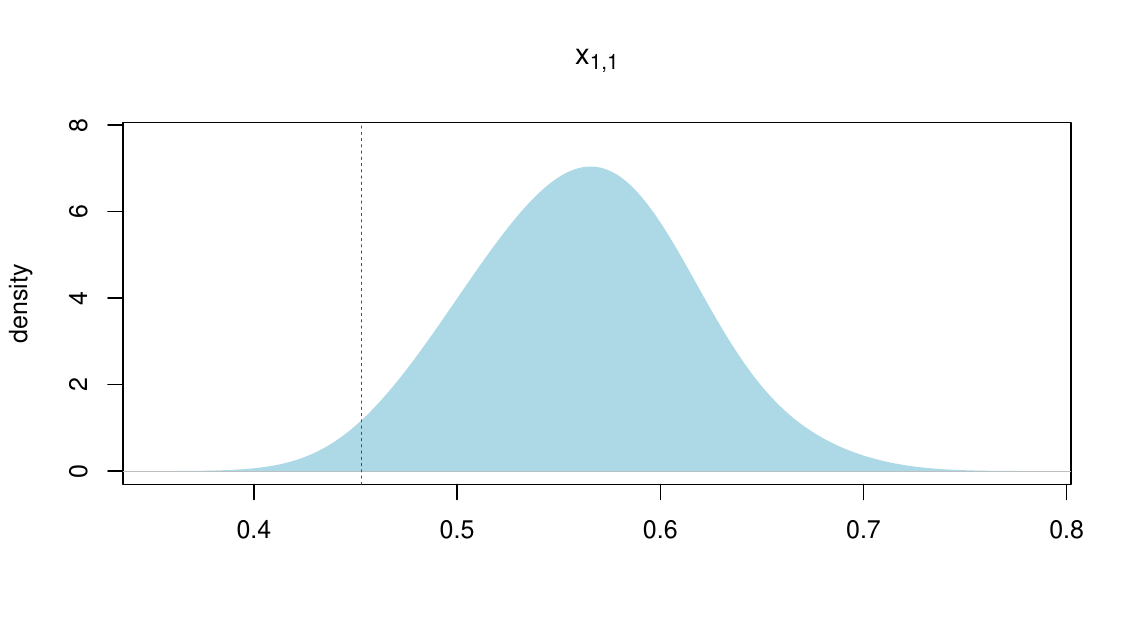}\\
     GP $\mathcal{M}_1$&  GP $\mathcal{M}_2$&  GP $\mathcal{M}_3$\\
     \includegraphics[trim={1cm 1.5cm 1cm 1.5cm}, clip, width=0.41\linewidth]{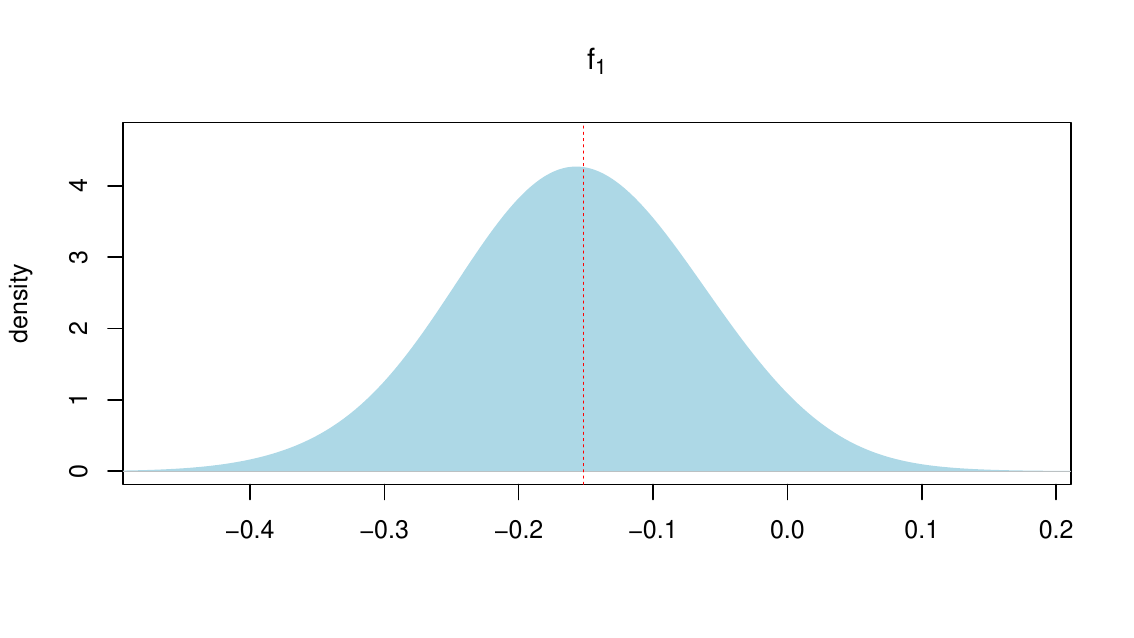}  & \includegraphics[trim={1cm 1.5cm 1cm 1.5cm}, clip, width=0.41\linewidth]  {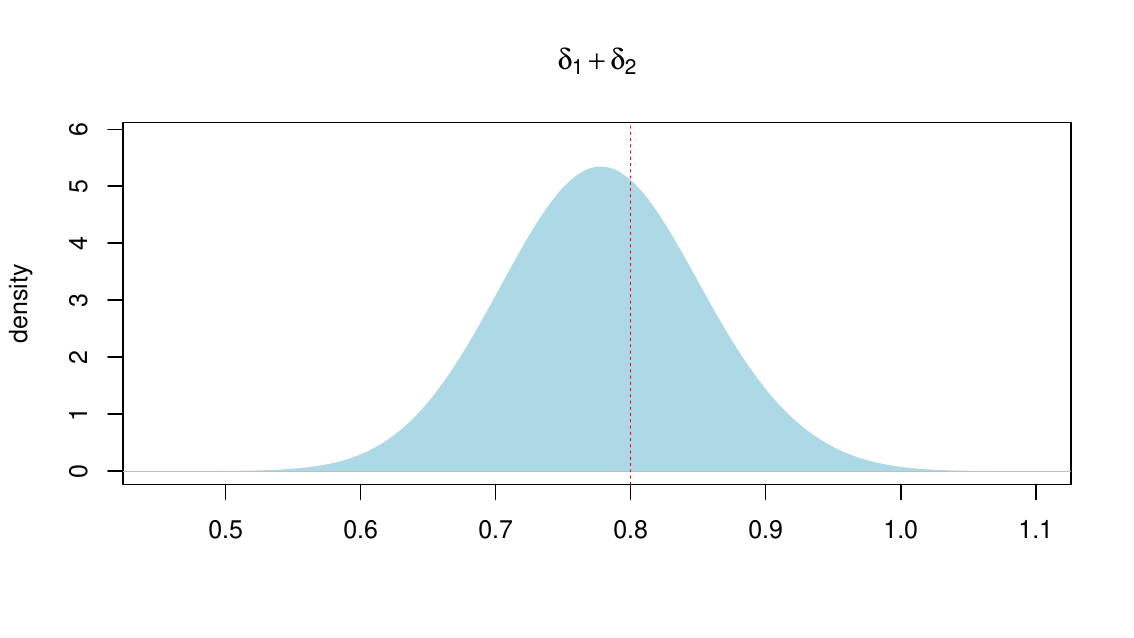}   & \includegraphics[trim={1cm 1.5cm 1cm 1.5cm}, clip, width=0.41\linewidth]{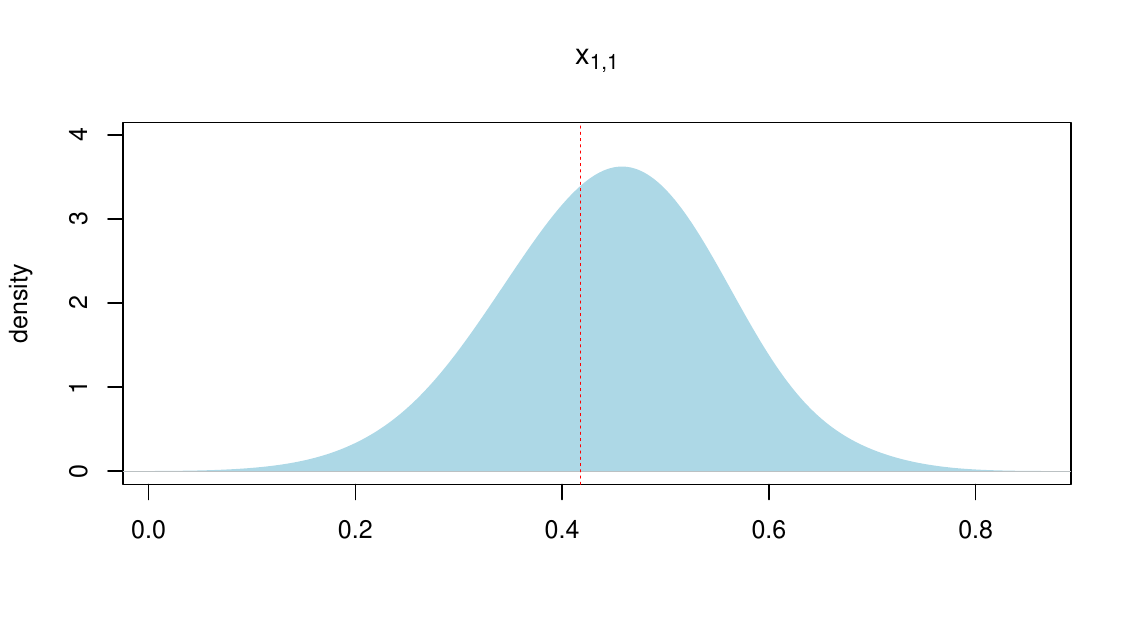}\\
    \end{tabular}}
    \caption{\textbf{Estimation Bias: }   Posterior Density (\usebox{\bluerectanglebox}) and true value (\usebox{\lineRed}) of $f_1$ for $\mathcal{M}_1$, $\delta_1+\delta_2$ for $\mathcal{M}_2$, and $\mathbf{x}_{1,1,1}$ for $\mathcal{M}_3$, for the Poisson specification (top) and the  Generalized Poisson specification (bottom) when data have been generated assuming a GP with overdispersion.
    }
    \label{fig:bias}
\end{figure}

\section{Applications}\label{sec:ill}
We illustrate the relevance of the proposed models through two empirical applications exhibiting overdispersion. In both cases, we estimate and compare the three model specifications, $\mathcal{M}_1$, $\mathcal{M}_2$, and $\mathcal{M}_3$, under both the Poisson and the Generalized Poisson (GP) likelihoods. Model $\mathcal{M}_2$ is estimated by fixing $k=2$, while Model $\mathcal{M}_3$ adopts a latent space representation with dimension $d=2$.

\subsection{Citibike dataset}
In the first application, we consider the Citibike dataset \cite{citibike2019}, which contains information on rides between any two stations of the Citibike bike-sharing service in New York City.  Part of this dataset -- one weekday in 2019 --  has been previously studied in \cite{he2025semiparametric} where the authors proposed a semi-parametric dynamic Poisson LS network model.
In this application, we consider bike-sharing activity for the entire year 2019 at a monthly frequency.  We aggregate ride counts between any two stalls at the Neighborhood--Tabulation--Area (NTA) level. We thus obtain a count-weighted symmetric network with 61 nodes representing the 61 NTA areas and edges representing the bike-sharing service. 

The top panel of Table \ref{tab:diccity} reports the overdispersion parameter $\hat\rho$ for for model $\mathcal{M}_1$, $\mathcal{M}_2$, and $\mathcal{M}_3$ under the GP likelihood assumption. In all three model specifications, we observe strong overdispersion. Model $\mathcal{M}_3$ exhibits the lowest $\hat\rho$ as overdispersion is also captured by the latent coordinates. The bottom panel of Table \ref{tab:diccity} reports the model-sections results in terms of DIC  for $\mathcal{M}_1$, $\mathcal{M}_2$, and $\mathcal{M}_3$ both under Poisson and GP likelihood. Overall, the GP specification is preferred to its Poisson counterpart. As $\mathcal{M}_3$ is the preferred model under both likelihood assumptions, we present a description and comparative analysis of the results under this model. Further estimates are reported in Section \ref{app:emp} of the Supplement.
\begin{table}[h!]
    \centering

\resizebox{0.85\textwidth}{!}{\renewcommand{\arraystretch}{0.6}
\begin{tabular}{lccc}
\hline
 & $\mathcal{M}_1$ & $\mathcal{M}_2$ & $\mathcal{M}_3$ \\
\hline
\multicolumn{4}{c}{\textit{(a) Dispersion Features under GP likelihood}} \\
\hline
$\hat\rho$ 
& 12'344.17
& 13'769.64
& 1'397.01 \\
& $[11'219.90, 13'616.03]$
&  $[12'117.37, 15'689.06]$
& $[1'331.62, 1'458.43]$ \\
\hline
\multicolumn{4}{c}{\textit{(b) Model Selection (DIC)}} \\
\hline
Generalized Poisson 
& \textbf{219'497.9}
& \textbf{251'559.1}
& \textbf{192'668.4} \\
Poisson
& 27'201'265.0
& 3'073'572.0
& 2'833'043.0 \\
\hline
\end{tabular}}
    \caption{\textbf{Citibike dataset:} (a) posterior mean of the overdispersion parameter $\hat\rho$ for $\mathcal{M}_1$, $\mathcal{M}_2$, and $\mathcal{M}_3$ in the Citibike application with GP likelihood. 95\% credible interval in brackets. (b) DIC comparison between the Generalized Poisson models and the Poisson models with $\mathcal{M}_1$, $\mathcal{M}_2$, and $\mathcal{M}_3$ specification. }
    \label{tab:diccity}
\end{table}

\begin{figure}[h!]
    \centering
\resizebox{\textwidth}{!}{
    \begin{tabular}{cc}
    \footnotesize a) Node Centrality, GP Specification &\footnotesize b) PP Average Expected Strength, GP Specification\\
    \includegraphics[width=0.25\linewidth]{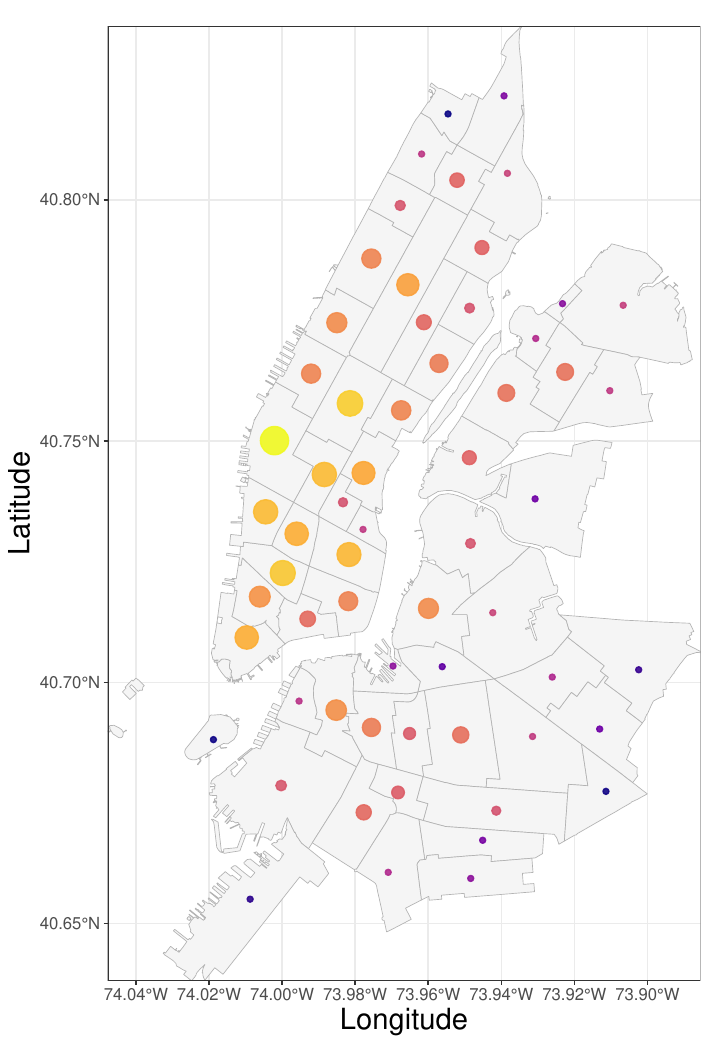}  &  \includegraphics[width=0.45\linewidth]{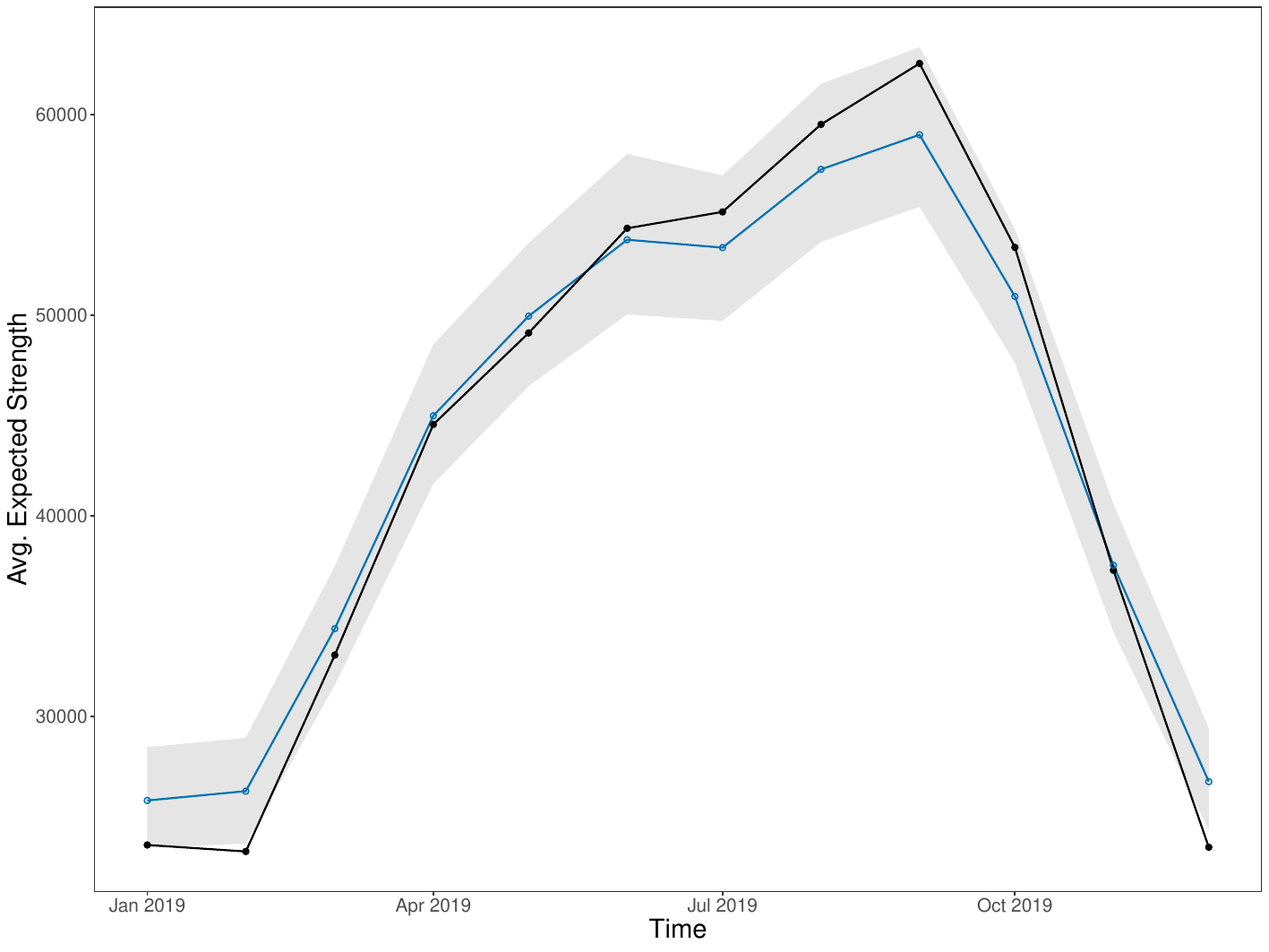}\\    \footnotesize c)  Latent Space, Poisson Specification  (Apr 2019)& \footnotesize d) Latent Space, GP Specification (Apr 2019)\\ [0.5em] \includegraphics[trim={0cm 0cm 0cm 1.1cm}, clip, width= 0.45\linewidth]{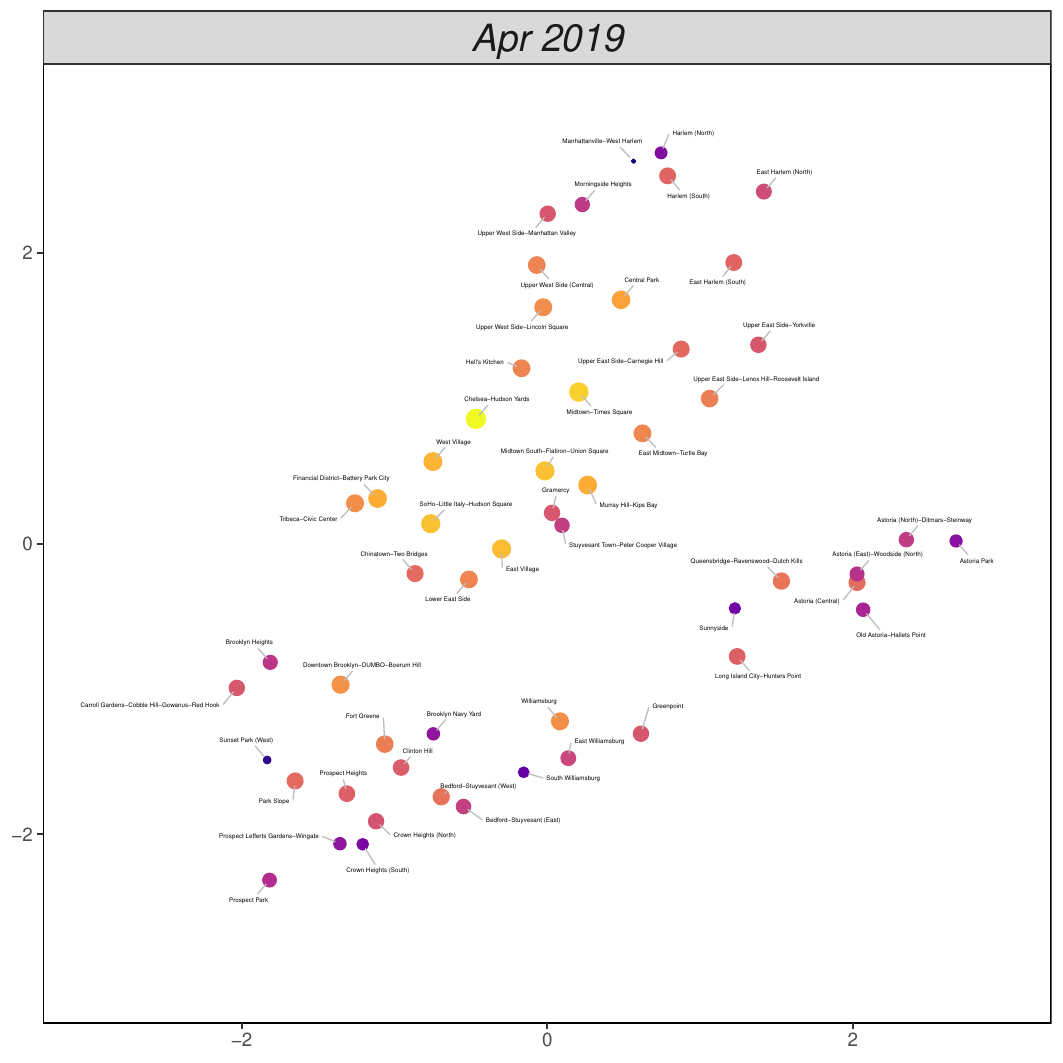} &
    \includegraphics[trim={0cm 0cm 0cm 1.1cm}, clip, width= 0.45\linewidth]{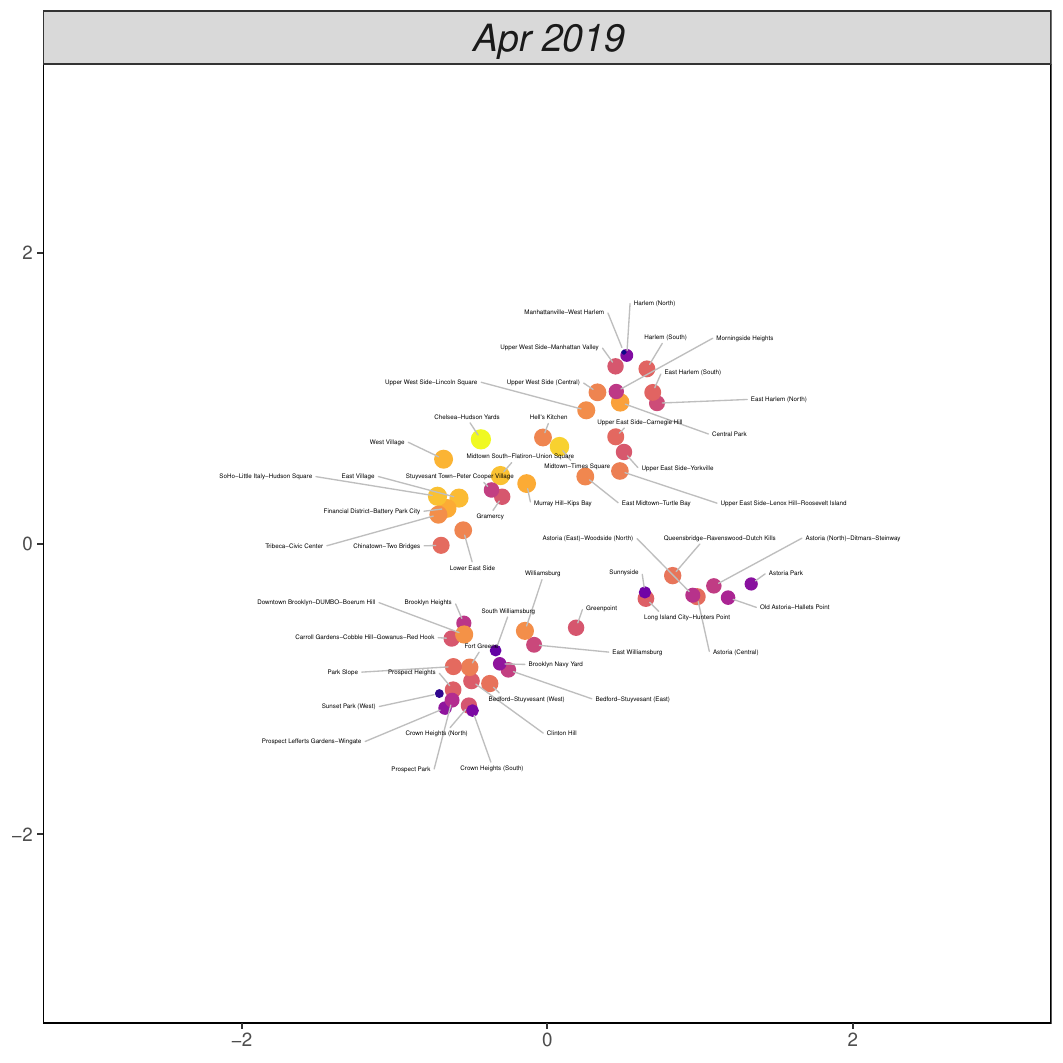} 
    \end{tabular}
     }
    \caption{\textbf{Citi Bike Network, Model  $\mathcal{M}_3$ Results:} (a) posterior mean of the $\alpha_i$ parameter associated to each of the 61 neighborhoods. The larger the NTA dot size, the more prominent the neighborhood in the Citi Bike network (\usebox{\purplesmalldot} - \usebox{\yellowbigdot}). (b) time series of the observed average strength (\usebox{\linewithcirclesbox}) against the posterior mean (\usebox{\linewithcirclesboxBlue}) with its 95\% credible interval (\usebox{\grayrectanglebox}). (c)-(d) latent space representation of the neighborhoods with $d = 2$ in April 2019 for the Poisson specification (c) and Generalized Poisson specification (d). The size of the nodes is proportional to $\alpha_i$. The larger the NTA dot size, the more prominent the neighborhood in the Citi Bike network (\usebox{\purplesmalldot} - \usebox{\yellowbigdot}). }
    \label{fig:application1_m3}
\end{figure}

The top-left panel in Figure \ref{fig:application1_m3} shows the posterior mean of the $\alpha_i$ parameter associated with each of the 61 New York City neighborhoods. We observe that prominent and touristic NTAs, e.g., those located on the island of Manhattan (Murray Hill-Kips Bay, East Village, and Union Square), exhibit a higher $\alpha_i$ coherently with the interpretation of these parameters in terms of node centrality. The top-right panel shows the posterior predictive (PP) of the average expected strength over time. The model correctly captures the network's yearly seasonality. Higher riding activity is observed in the spring and summer, while a reduction in the activity is observed in the fall and winter. 

In the bottom panels, we present a comparison of latent space estimation with $d=2$ under the Poisson and GP specifications for April 2019. Further estimates are reported in the Supplement (Section \ref{app:emp}). We observe that the latent-space representation enables the retrieval of the relative geographical positions of the nodes (compare with the left panel). Stalls located in the Manhattan Island cluster together, and so do stalls located in the Bronx and Queens. Latent coordinates are more scattered under the Poisson specification since the Poisson model struggles to capture overdispersion, and a higher posterior variance of the latent coordinates partially mitigates this misspecification. This result is consistent with the misspecification bias for the latent coordinates discussed in Section \ref{sec:sim} (see also Figure \ref{fig:bias}). 

\subsection{Media network dataset}

As a second application, we consider the time-varying media network dataset (\citealp{schmidt2018polarization},\citealp{casarin2025media}). The dataset enables the construction of timely networks in which nodes are news outlets, and the weight of each edge is the count of unique commenters who interact with a pair of news outlets within the time window under consideration. In this application, we consider media networks aggregated monthly, totaling 24 networks spanning 2015 and 2016.

The top panel of Table \ref{tab:dic_media} reports the overdispersion parameter $\hat\rho$ for for model $\mathcal{M}_1$, $\mathcal{M}_2$, and $\mathcal{M}_3$ under the GP likelihood assumption. In all three model specifications and for all countries, we observe the presence of overdispersion. Model $\mathcal{M}_3$ exhibits the lowest $\hat\rho$ as overdispersion is also partially captured by the latent coordinates. The bottom panel reports the comparison between Poisson and GP assumptions in terms of DIC. The improvement in DIC when accounting for overdipersion is evident in all the specifications considered.

\begin{table}[h!]
\centering 
\setcellgapes{0pt}
\makegapedcells
\resizebox{0.98\textwidth}{!}{
{\renewcommand{\arraystretch}{0.6}
\begin{tabular}{lcccccc}
\hline
& \multicolumn{3}{c}{France} & \multicolumn{3}{c}{Germany} \\
\cline{2-4}\cline{5-7}
& $\hat\rho_{\mathcal{M}_1}$ & $\hat\rho_{\mathcal{M}_2}$ & $\hat\rho_{\mathcal{M}_3}$
& $\hat\rho_{\mathcal{M}_1}$ & $\hat\rho_{\mathcal{M}_2}$ & $\hat\rho_{\mathcal{M}_3}$ \\
GP
&  136.81 & 148.85     &     15.43   &   73.94&83.98 &   14.73   \\
&$[133.26, 140.56]$&$[143.76,155.07]$&$[15.11, 15.76]$&$[71.85, 76.41]$&$[81.32, 87.65]$&$[14.35, 15.19]$\\
\hline
& \multicolumn{3}{c}{Italy} & \multicolumn{3}{c}{Spain} \\
\cline{2-4}\cline{5-7} & $\hat\rho_{\mathcal{M}_1}$ & $\hat\rho_{\mathcal{M}_2}$ & $\hat\rho_{\mathcal{M}_3}$
& $\hat\rho_{\mathcal{M}_1}$ & $\hat\rho_{\mathcal{M}_2}$ & $\hat\rho_{\mathcal{M}_3}$ \\
     GP
&   196.03    &  185.32     &     31.82&  266.21 & 263.83 &80.57     \\     
&   $[189.92, 203.84]$    &  $[176.99, 194.43]$     &      $[30.66, 32.84]$  &   $[248.49, 277.94]$  &  $[251.04, 275.94]$  &    $[77.01, 84.64]$     \\     
\hline
& \multicolumn{3}{c}{France} & \multicolumn{3}{c}{Germany} \\
\cline{2-4}\cline{5-7}
& $DIC(\mathcal{M}_1)$ & $DIC(\mathcal{M}_2)$ & $DIC(\mathcal{M}_3)$
& $DIC(\mathcal{M}_1)$ & $DIC(\mathcal{M}_2)$ & $DIC(\mathcal{M}_3)$ \\
\hline
GP
& \textbf{367'424.9} & \textbf{398'815.8} & \textbf{326'543.3}
& \textbf{206'671.0} & \textbf{222'767} & \textbf{184'544.2} \\
Poisson
& 2'687'285.0 & 3'082'142.0 & 668'439.5
& 937'230.1 & 10'93'692.0 & 383'815.8 \\
\hline
& \multicolumn{3}{c}{Italy} & \multicolumn{3}{c}{Spain} \\
\cline{2-4}\cline{5-7}
& $DIC(\mathcal{M}_1)$ & $DIC(\mathcal{M}_2)$ & $DIC(\mathcal{M}_3)$
& $DIC(\mathcal{M}_1)$ & $DIC(\mathcal{M}_2)$ & $DIC(\mathcal{M}_3)$ \\
\hline
GP
& \textbf{203'143.8} & \textbf{220'956.5} & \textbf{190'296.7}
& \textbf{185'350.3} & \textbf{199'642.4} & \textbf{173'968.4} \\
Poisson
& 1'899'189.0 & 2'244'081.0 & 561'232.7.0
& 2'396'896.0 & 2'733'931.0 & 1'004'719.0 \\
\hline
\end{tabular}}}
\caption{Top panel,  posterior mean of the overdispersion parameter $\hat\rho$ for $\mathcal{M}_1$, $\mathcal{M}_2$, and $\mathcal{M}_3$ in the Media network application with GP likelihood. 95\% credible interval in squared brackets. Bottom panel, DIC comparison between the Generalized Poisson model and the Poisson model for $\mathcal{M}_1$, $\mathcal{M}_2$, and $\mathcal{M}_3$ in the Media network application.}
\label{tab:dic_media}
\end{table}

As the GP LS model ($\mathcal{M}_3$) is the best performing in terms of DIC, we will report and comment on the estimation results under this specification.  For illustration purposes, we show in the main text the results for France in Figure \ref{fig:FRANCE_application1_m3}. The left panel shows a spatial representation of the news outlets (located at their headquarters), with point size and color proportional to $\alpha_i$. As expected, we find that national news outlets, often located in the most prominent cities of the four countries, are more central within the networks, whereas local news outlets tend to be more peripheral at the geographical level. We also observe differences in the geographical distributions of the most central news outlets. From a cross-country comparison (see Figures \ref{fig:media_bias_de_frMAP}-\ref{fig:media_bias_it_spMAP} in the Supplement), France and Spain are more centralized, with prominent news outlets located in their respective capital cities, whereas Germany and Italy exhibit a more heterogeneous geographical distribution of central news outlets. This may be due to very well-known historical reasons related to a monocentric \emph{vs} policentric development of the four countries. See \citet{tezis2008european, newman2015reuters} for a complete description of the different four media environments.

The top-right panel in Figure \ref{fig:FRANCE_application1_m3} shows the latent-space representation of news outlets in January 2016. There are two forms of clustering. First, national news outlets with similar editorial lines are close in the latent space.  Second, news outlets that are geographically closer tend to cluster. This effect is stronger for local news outlets. Similar evidence is found for the other countries in the dataset (see top-plots of Figures \ref{fig:media_bias_de_fr}-\ref{fig:media_bias_it_sp} in Supplement). 
\begin{figure}[t]
    \centering
    \begin{tabular}{cc}
    \footnotesize a) Node Centrality Measures, GP Specification &\footnotesize b) Latent Space, GP Specification (Jan 2016)\\
    \includegraphics[width=0.48\linewidth]{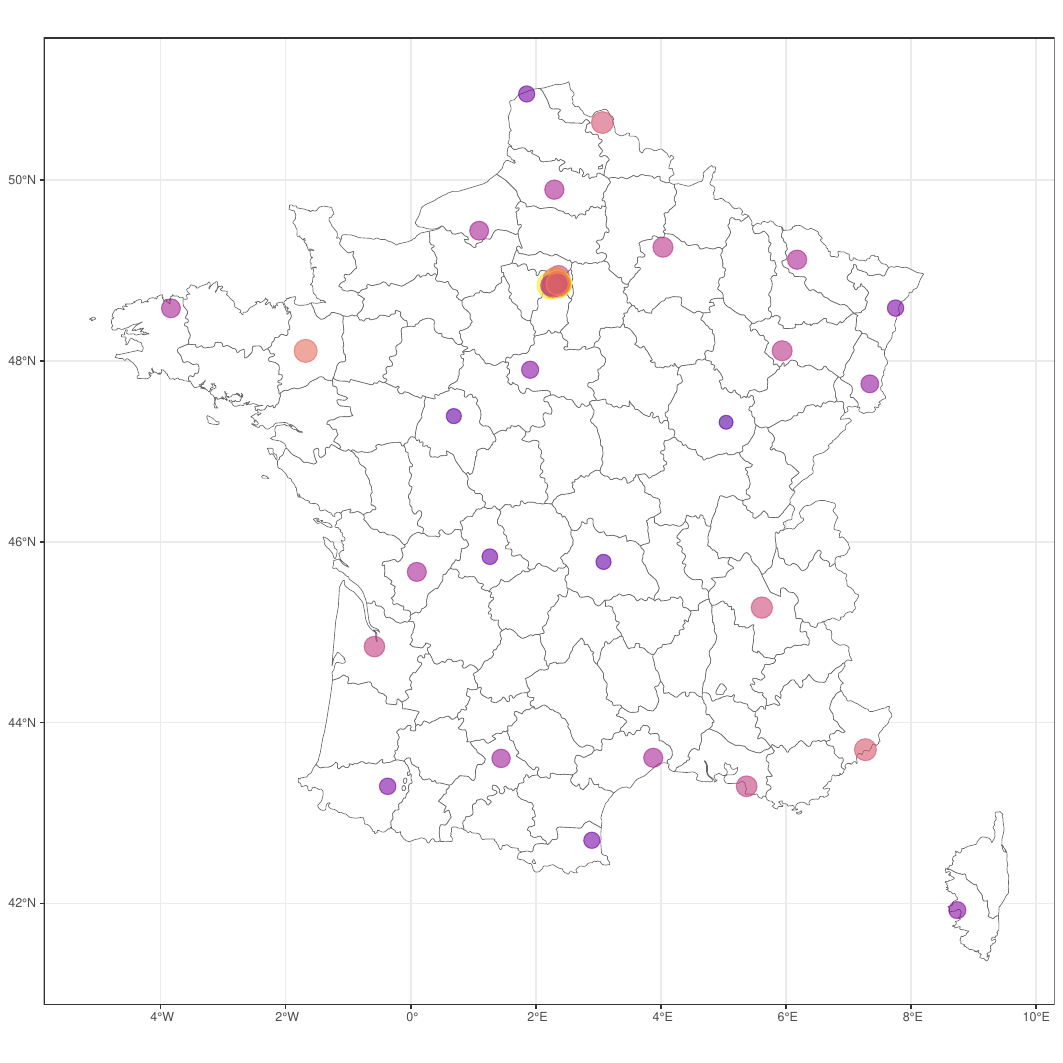}  &  \includegraphics[width=0.48\linewidth]{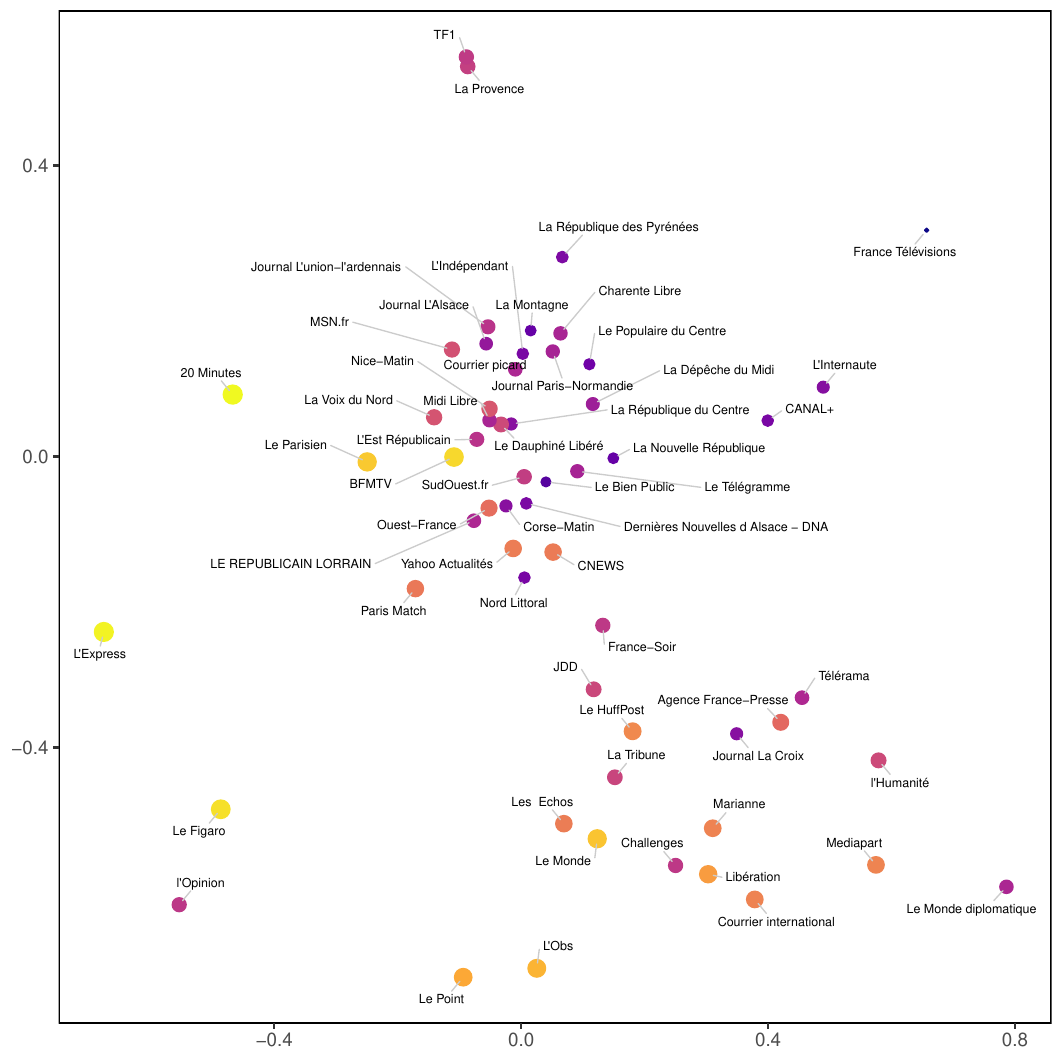}\\
    \end{tabular}
    \caption{\textbf{Media Network, Model  $\mathcal{M}_3$ Results:} (a) Geographical distribution of the news outlets for France with node color and size proportional to the posterior mean of $\alpha_i$ (\usebox{\purplesmalldot} - \usebox{\yellowbigdot}). (b) Posterior mean of the latent coordinates with node color proportional to the posterior mean of $\alpha_i$ (January 2016). }
    \label{fig:FRANCE_application1_m3}
\end{figure}

Regarding the coordinate dynamics, there is evidence of heterogeneous trending behavior across countries (see the latent factor estimates in the bottom panels of Figures \ref{fig:media_bias_de_fr}-\ref{fig:media_bias_it_sp} in Supplement). For both France and Germany, there is evidence of an upward trend. In Italy, the factor exhibits a pronounced upward trend at the beginning of the sample, peaking in July 2015, followed by a gradual decline that reaches its minimum around March 2016. Subsequently, the latent factor increases again, suggesting a sharpening in the overall connectivity. The time-varying latent factor for Spain is relatively stable with peaks in July 2015, December 2015, February and March 2016, and October 2016. Overall, the trend behavior confirms the results obtained in \cite{casarin2025media}, which fitted a Poisson Markov switching LS model with a large number of hidden states.

We conclude our analysis with an out-of-sample prediction exercise. As LS models reveal themselves to be powerful for edge weight imputation, we inject 3 \emph{NA}s at random in the last adjacency matrix of each country. We then test the ability of the estimated LS model to impute the missing values under the Poisson and GP specifications. 
\begin{table}[t]
\centering
\resizebox{0.99\textwidth}{!}{{\renewcommand{\arraystretch}{0.6}
\begin{tabular}[t]{cccccccccc}
\hline
Model & Country & MAE & MSE & RMSE & C & AWI & MTP & VTP & PE\\
\hline
France & GP & \textbf{72.235} & \textbf{41150.805} & \textbf{202.857} & \textbf{0.967} & 258.642 & \textbf{0.540} & \textbf{0.069} & \textbf{0.100}\\
France & Pois & 81.319 & 58706.617 & 242.294 & 0.667 & \textbf{64.419} & 0.587 & 0.140 & 0.367\\
\hline
Germany & GP & \textbf{12.904} & \textbf{372.171} & \textbf{19.292} & \textbf{1.000} & 94.686 &\textbf{0.586} & \textbf{0.053} & \textbf{0.100}\\
Germany & Pois & 14.586 & 496.301 & 22.278 & 0.500 & \textbf{24.339} & 0.599 & 0.188 & 0.700\\
\hline
Italy & GP & 182.308 & 577707.031 & 760.070 & \textbf{0.933} & 279.208 & \textbf{0.520} & \textbf{0.084} & \textbf{0.167}\\
Italy & Pois & \textbf{151.173} & \textbf{471287.814} & \textbf{686.504} & 0.633 & \textbf{50.207} & 0.531 & 0.188 & 0.567\\
\hline
Spain & GP & 51.549 & 8339.663 & 91.322 & \textbf{0.967} & 331.824 & \textbf{0.649} & \textbf{0.072} & \textbf{0.233}\\
Spain & Pois & \textbf{49.932} & \textbf{8307.403} & \textbf{91.145} & 0.400 & \textbf{37.308} & 0.665 & 0.180 & 0.733\\
\hline
\end{tabular}}}
\caption{Out-of-sample predictive diagnostics for the Media Network application by country. The table reports the mean absolute error (MAE), mean squared error (MSE), root mean squared error (RMSE), the coverage as the percentage of true values falling within the 95\% predictive intervals (C), the average width of these intervals (AWI), the mean and variance of the predictive tail probabilities (MTP and VTP, respectively), and the proportion of extreme observations.}
\label{tab:prediction_diag}
\end{table}
Table \ref{tab:prediction_diag} reports the prediction metrics by country under the two likelihood assumptions averaged over 10 simulations. Dealing with point-wise prediction MAE, MSE, RMSE), we notice that the results are mixed (GP outperforms Poisson for France and Germany, vice versa for Italy and Spain).
Dealing with the distributional forecast, we notice that the GP model consistently shows a more satisfactory coverage ($>$ 90.0\%). As concerns tail probabilities, the GP model appears much better calibrated (mean and variance of the tail probability, respectively, close to 1/2 and 1/12) at the cost of a higher interval width in all cases. To sum up, the Poisson model, despite providing competitive point-wise predictions, is also overconfident, while the GP clearly wins in terms of uncertainty quantification. 

\section{Conclusion}
New model classes for count--weighted temporal networks are proposed. The Generalized Poisson (GP) distribution captures edge--weight underdispersion and overdispersion, while lagged network features or latent variables drive the dynamics of edge weights. Sufficient conditions for latent variable identification are derived, and a Bayesian inference framework is developed together with an efficient posterior sampling procedure. Numerical illustrations provide evidence of misspecification bias and significant posterior predictive errors when weight overdispersion is neglected. Bike-sharing and media temporal network applications demonstrated the GP latent space model's ability to accommodate unequal dispersion, as well as other features such as seasonality, trend, and spatial effects. 




\bibliography{biblio}

\clearpage

\begin{center}
{\Large Supplementary Materials to the paper\\\vspace{12pt}
Generalized Poisson Dynamic Network
Models
}
\vspace{12pt}
\end{center}

\renewcommand{\theequation}{A.\arabic{equation}}
\renewcommand{\theproposition}{A.\arabic{proposition}}
\setcounter{equation}{0}
\setcounter{proposition}{0}

\begin{appendix}
\section{Proof of the results in the paper}
\label{app:proofs}
\par\noindent\begin{proof}[\textbf{Proof of Proposition \ref{prop1}}]
For simplicity, we drop the time index. Let us assume $Y_{ij} \sim \mathcal{GP}(\lambda_{ij}, \theta)$, $i,j=1,\ldots,n$ $i\neq j$ is a sequence of GP variables, with pmf
\begin{equation*}
\Pr(\{Y_{ij} = y_{ij}\}) = \frac{\lambda_{ij} (\lambda_{ij} + \theta y_{ij})^{y_{ij}-1} \exp(-(\lambda_{ij} + \theta y_{ij}))}{y_{ij}!}, 
\quad y_{ij} = 0, 1, 2, \ldots,
\end{equation*}
Following \cite{Ambagaspitiya_Balakrishnan_1994}, Eq. 2.5, the probability generating function of $Y_{ij}$ is
\begin{equation*}
G_{Y_{ij}}(z) = \mathbb{E}[z^{Y_{ij}}] = \exp\!( -\frac{\lambda_{ij}}{\theta}(W(-\theta z\exp(-\theta))+\theta)),
\end{equation*}
where $W(x)$ denotes the Lambert's function defined as $W(x)\exp(W(x))=x$ \citep{mezo2022lambert}. Let $S=\sum_{i=1}^{n} Y_i$, since $Y_1, \ldots, Y_n$ are independent with the same $\theta$ but possibly different $\lambda_i$,
\begin{equation*}
G_{S}(z) 
= \prod_{i=1}^n\prod_{j\neq i} G_{Y_{ij}}(z)
= \exp\!( -\frac{1}{\theta}\sum_{i=1}^{n}\sum_{j\neq i}\lambda_i(W(-\theta z\exp(-\theta))+\theta)),
\end{equation*}
which is the pgf of a generalized Poisson distribution with parameters $\sum_{i=1}^n\sum_{j\neq i}\lambda_{ij}$ and $\theta$. Hence,
$S \sim \mathcal{GP}\!( \sum_{i=1}^n\sum_{j\neq i} \lambda_{ij}, \, \theta )$. The results easily extend to double sequences of independent GP variables.
\end{proof}
\par\noindent\begin{proof}[\textbf{Proof of Lemma \ref{propsubexp}}] 
\textit{i)} The proof follows from showing the mgf in Eq. \ref{mgf} is bounded in a neighborhood of 0 and then by applying Theorem 2.7.1 iv) in \cite{vershynin2018high}. Since $W(z)$ in the mgf is bounded for $z\geq -1/e$ then the mgf is bounded for $-\theta \exp(u-\theta)\geq -1/e$ that is for $u\in(-\infty,\theta-\log(\theta))$.
\bigskip
\par\textit{ii)} We expand the logarithm of the moment generating function (mgf) and bound the residual term. Let \( z(u) = -\theta \exp(u-\theta)\), the mgf of $Y$ is 
\[
M_Y(u) = \exp(-\frac{\lambda}{\theta}\left(W(z(u))+\theta\right))
\]
$u < u_0(\theta)$, where $W(z)$ is the Lambert's W function \citep{Ambagaspitiya_Balakrishnan_1994}. 

Let us define \( g(u) = \log M_Y(u) \) then its derivatives are
\[
\begin{aligned}
g^{(1)}(u) &= -\frac{\lambda}{\theta} \frac{W}{1+W},\,\, g^{(2)}(u) = -\frac{\lambda}{\theta} \frac{W}{(1+W)^3}, \,\, g^{(3)}(u) = \frac{\lambda}{\theta}  \frac{(2W - 1)W}{(1+W)^5}, \\[4pt]
g^{(4)}(u) &= \frac{\lambda}{\theta}
\frac{(1+W)\big(4W^3 - 6W^2 + W + 1\big) + 6W^2(1 - W)}{(1+W)^7},
\end{aligned}
\]
where we used \( z'(u) = z(u) \) and \( W'(u) = W(u)/(z(1+W(z)))\) \citep{mezo2022lambert}. Evaluating the derivative $g^{(j)}(u)$ at $u=0$  and using the definition of Lambert's function \( W(-\theta \exp(-\theta)) = -\theta \), we obtain the cumulants of the GP:
\[
\begin{aligned}
\kappa_1 &= g^{(1)}(0) = \frac{\lambda}{1-\theta},\,\, \kappa_2 = g^{(2)}(0) = \frac{\lambda}{(1-\theta)^3},\\
\kappa_3 &= g^{(3)}(0) = \frac{\lambda(2\theta+1)}{(1-\theta)^5},\,\,\kappa_4 = g^{(4)}(0) = \frac{\lambda(6\theta^2+8\theta+1)}{(1-\theta)^7}.
\end{aligned}
\]
For an alternative derivation method, see \cite{Con1989}. The Taylor expansion of the log–mgf, for \(|u| < u_0(\theta)\), is \(
g(u)
= \kappa_1 u + \frac{\kappa_2}{2} u^2 + \frac{\kappa_3}{6} u^3 + R_4(u)\) with remainder \(R_4(u) = \frac{g^{(4)}(\xi)}{24}\,u^4\) for some \(\xi \in (0,u)\).  Since \(\mathbb{E}[Y]=\kappa_1\), then
\[
\log \mathbb{E}\!\left[\exp(u(Y-\kappa_1))\right]
= \frac{\kappa_2}{2} u^2 + \frac{\kappa_3}{6} u^3 + R_4(u)\leq \frac{\kappa_2}{2} u^2
+ \frac{|\kappa_3|}{6}|u|^3
+ \frac{B_r(\theta,\lambda)}{24}|u|^4,
\]
where we defined $B_r(\theta,\lambda) = \sup\{|g^{(4)}(u)|: |u|\le r\}$ for any \( 0 < r < u_0(\theta) \).  For \( |u|\le r \), the power series inequality \(
(1 - b|u|)^{-1}= 1 + b|t| + b^2 t^2 + \dots \) implies that
\(v t^2(1 - b|u|)^{-1}\ge v u^2 +v b|u|^3 + v b^2 |u|^4\) \citep[][]{vershynin2018high}. Choosing \(b\) large enough to dominate the cubic and quartic terms guarantees
\[
\frac{\kappa_2}{2} u^2
+ \frac{|\kappa_3|}{6}|u|^3
+ \frac{B_r(\theta,\lambda)}{24}|u|^4
\;\le\;
\frac{v u^2}{2(1 - b|u|)},
\quad |u|\le r.
\]
We conclude $\log \mathbb{E}[\exp(u(Y-\kappa_1))]
\;\le\;
\frac{v u^2}{2(1 - b|u|)}$, where $v = \kappa_2 \frac{\lambda}{(1-\theta)^3},
\quad
b = 
\frac{|2\theta+1|}{3(1-\theta)^2}
+ \frac{B_r(\theta,\lambda)}{12\lambda}(1-\theta)^3 r$.
\end{proof}
\begin{proof}[\textbf{Proof of Proposition \ref{propB}}]
Since the definition of Orlicz norm  \(||Z||_{\psi_1}=\inf\left\{K>0: \mathbb{E}(\exp(|Z|/K))\le 2\right\}\), requires $\mathbb{E}(\exp(|Z|/K))\le 2$, the log-mgf bound in Lemma \ref{propsubexp} can be used to derive, through Chernoff bound and tail-integral representation, a bound for the Orlicz norm. The result in \cite[][Lemma 2.7.6, ii)]{vershynin2018high} and the log-mgf bound in Lemma \ref{propsubexp} gives
\[
||Y||_{\psi_1}\leq
C\big(\sqrt{v}+b\big)
=
C(
\frac{\sqrt{\lambda}}{(1-\theta)^{3/2}}
\;+\;
\frac{|2\theta+1|}{3(1-\theta)^2}
+\frac{(1-\theta)^3\,r}{12\lambda}\,B_r(\theta,\lambda)).
\]
for an absolute constant \(C>0\) and \( r\geq u_0(\theta)\) where $v$ and $b$ have been defined in Proposition \ref{propsubexp}.
\end{proof}

\begin{proof}[\textbf{Proof of Proposition \ref{prop2}}]
Define the symmetric $(N\times N)$ matrix $X_t$ with entries $X_{ijt}=Y_{ijt}-\mathbb{E}(Y_{ijt}|\mathcal{F}_{t-1})$ $i=1,\ldots,N$, $j=1,\ldots,N$, $j\neq i$. Since the mgf of the GP is well defined on a neighborhood of 0, and the GP has finite exponential moments, then $Y_{ijt}$ is sub-exponential, in the light-tail sense, as stated in Proposition \ref{propsubexp}. Moreover $||Y_{ijt}||_{\psi_1}<\infty$ following the result in Proposition \ref{propB}. In conclusion, the matrix $X_t$ is self-adjoint (symmetric), has zero mean, and sub-exponential tails. Thus, the conditions of  \citep[][Sec. 6]{tropp2015introduction} are satisfied, and a Bernstein matrix inequality can be obtained as shown in the following. See also \cite[][th. 2.8.1 and th. 5.4.1]{vershynin2018high}. 

We follow the same procedure used in the proof of Th. 6.5.1 of \cite{vershynin2018high} for matrices with non-i.i.d. entries and
decompose $X_t$ as:
\[X_t=\sum_{i=1}^N \sum_{j\neq i}^N X_{ijt}E^{ij},\]
where $E^{ij}=\mathbf{e}_i\mathbf{e}_j'$ with $\mathbf{e}_i$ $i=1,\ldots,N$ the standard basis vector of $\mathbb{R}^N$. Define $v_t^2=||\sum_{i\neq j} \mathbb{E}(X_{ijt}^2E^{ij}E^{ij}|\mathcal{F}_{t-1}) ||_{op}^2$ and $K_t=\underset{i\neq j}{\max}||X_{ijt}||_{\psi_1}$, which is bounded for sub-exponential variables \citep[][Proposition 2.7.1, iv)]{vershynin2018high}. Since $E_{ij}E_{ij}=E_{ij}$ it follows that
\begin{align}
&v^2_t=||\sum_{i\neq j} \mathbb{E}(X_{ijt}^2E^{ij}E^{ij}|\mathcal{F}_{t-1}) ||_{op}=||\sum_{i\neq j} \mathbb{E}(X_{ijt}^2E^{ij}|\mathcal{F}_{t-1}) ||_{op}\\
&=||\sum_{i=1}^{N}\sum_{j\neq i}^{N}\mathbb{E}(X_{ijt}^2)\mathbf{e}_i\mathbf{e}_j'||_{op}=\underset{i}{\max}\sum_{j\neq i}\mathbb{E}(X_{ijt}^2)=\underset{i}{\max}\sum_{j\neq i}\frac{\lambda_{ij}}{(1-\theta)^3},
\end{align}
where the last line follows from the fact that the matrix is diagonal and from the operator norm of $\mathbf{e}_i$. From \citep[][Sec. 6]{tropp2015introduction}, there exist universal constants $C$ and $c$ such that for al $\tilde{u}\geq 0$
\begin{equation}
\mathbb{P}r(\{||X_t||_{op}>\tilde{u}\}|\mathcal{F}_{t-1})\leq 2 \exp(-c \min (\tilde{u}^2/v^2_t,\tilde{u}/K_t)).
\end{equation}
Assume $\tilde{u}=C(\sqrt{v_t^2 \log N}+K_t \log N)+u)$ then $\tilde{u}^2/v^2_t\approx \log N $ and $\tilde{u}/K_t\approx \log N$, thus the exponential bound becomes $\exp(-c\log N)=N^{-c}$. We obtain
\begin{equation}
\mathbb{P}r(\{||X_t||_{op}>C(\sqrt{v_t^2 \log N}+K_t \log N)+u\}|\mathcal{F}_{t-1})\leq 2 \exp(-c \min (u^2/v^2_t,u/K_t))
\end{equation}
and we write $\|X_t\|_{\mathrm{op}}=\mathcal{O}_{\mathbb{P}}\!\big(\sqrt{v_t^2\log N}+K_t\log N\big)$. From Weyl's inequality $|\varrho(Y_t)-\varrho(\Lambda_t)|\leq||Y_t - \Lambda_t||_{op}=||X_t||_{op}$ and by applying Matrix Bernstein inequality, a probabilistic bound on this deviation is obtained. We conclude that with high probability, larger than $1-2 N^{-c}$
\begin{equation}
|\varrho(Y_t)-\varrho(\Lambda_t)|\leq C((\underset{i}{\max} \sum_{j=1,\,j\neq i}^N\lambda_{ijt}/(1-\theta)^3 \log N)^{1/2}+K_t \log N).
\end{equation}
\end{proof}

\begin{proof}[\textbf{Proof of Proposition \ref{prop3K}}]
Define, for parameters $\lambda>0$ and $\theta>0$, $M_t(u)=\mathbb{E}(\bar{S}_t|\boldsymbol{\alpha}\vee\mathcal{F}^{*}_{t-1})$ is
\begin{equation*}
    M_t(u) = \exp\!\bigl(-\lambda_t\,\varphi(g(u))\bigr),\quad
    \varphi(g) = \frac{W(g)+\theta}{\theta},\quad
    g(u) = -\theta\,\exp(-\theta + u),
\end{equation*}
where $W(x)$ denotes the Lambert's $W$-function (principal branch). The cumulant function is defined as $K_t(u)=\log M_t(u)=-\lambda_t \varphi(g(u))$. In order to find the find the expected value of the cumulant we proceed as follows: i) apply the Fa\`a di Bruno's formula to find the derivatives of $K_t(u)$ and ii) compute the expected value $\mathbb{E}(K_t(u)|\mathcal{F}^{*}_{t-1})$ with respect the distribution of $\boldsymbol{\alpha}$.

By using the Fa\`a di Bruno's formula, for functions $f$ and $h$ sufficiently smooth, one has
\begin{equation}\label{eq:FDB}
  \frac{d^j}{du^j} f\!\big(h(u)\big)
  \;=\; \sum_{\ell=1}^{j} f^{(\ell)}\!\big(h(u)\big)\; B_{j,\ell}\!\big(h^{(1)}(u),h^{(2)}(u),\ldots,h^{(j-\ell+1)}(u)\big),
\end{equation}
where $B_{j,k}$ denote the (partial/exponential) Bell polynomials. Moreover, if $h^{(j)}(u)=h(u)$ for all $j\ge 1$ and $u$, then the Bell polynomials reduce to $B_{j,\ell}(h(u),h(u),\ldots,h(u)) \;=\; h(u)^{\,\ell}\,S(j,\ell)$, where $S(j,\ell)$ are the Stirling numbers of the second kind. 

Define $\zeta^{(j)}(u)=\dfrac{d^j}{du^j}\varphi\!\big(g(u)\big)$ for $j\ge 1$ and apply \eqref{eq:FDB} with $f=\varphi$ and $h=g$. Since $g^{(r)}=\dfrac{d^r}{du^r}g(u)=g(u)$ for all $r\ge 1$ and all $u$, we obtain $B_{j,\ell}(g,\ldots,g)=g(u)^\ell\,S(j,\ell)$, yielding
\begin{align}\label{eq:zeta}
  &\zeta^{(j)}(u) \;=\; \sum_{\ell=1}^{j} \varphi^{(\ell)}\!\big(g(u)\big)\;
  B_{j,\ell}\!\big(g^{(1)}(u),\ldots,g^{(j-\ell+1)}(u)\big)=\sum_{\ell=1}^{j} \varphi^{(\ell)}\!\big(g(u)\big)\; g(u)^{\ell}\,S(j,\ell).
\end{align}
Since $\varphi(g)=\theta^{-1}W(g)+1$, we have for $\ell\ge 1$,
\begin{equation}\label{eq:phi-derivs}
  \varphi^{(\ell)}(x) \;=\; \frac{1}{\theta}\, W^{(\ell)}(x),\, W^{(\ell)}(x) \;=\; \frac{W(x)^\ell\, p_{\ell}(W(x))}{x^\ell \bigl(1+W(x)\bigr)^{2\ell-1}},
\end{equation}
where the polynomials $p_\ell$ obey the recurrence $p_{\ell+1}(x) = (1+x)\,p_\ell'(x) - (\ell x + 3\ell - 1)\,p_\ell(x)$, for $\ell\geq 1$, with $p_1(x)=1$ \citep[][Eq. 1.33-1.34]{mezo2022lambert}. For $j\ge 1$, using \eqref{eq:phi-derivs} one obtains
\begin{equation}\label{eq:zeta_concrete}
  \frac{d^j K_t(u)}{du^j} \;=\; -\lambda_t\sum_{\ell=1}^{j} \frac{1}{\theta}\,\frac{W\!\big(g(u)\big)^{\ell}\,
  p_{\ell}\!\big(W\!\big(g(u)\big)\big)}{\,\bigl(1+W\!\big(g(u)\big)\bigr)^{2\ell-1}}\,S(j,\ell).
\end{equation}

Recall that $g(0) = -\theta \exp(-\theta)$, $W\!\big(g(0)\big) = -\theta$, $\varphi\!\big(g(0)\big) = 0$. We have
\[
  \zeta^{(j)}(0)
  = \frac{1}{\theta} \sum_{\ell=1}^{j} S(j,\ell)g(0)^{\ell}\, W^{(\ell)}\!\big(g(0)\big),\,\, W^{(\ell)}\!\big(g(0)\big)
  = \frac{(-\theta)^{\ell}\,p_{\ell}(-\theta)}{(-\theta \exp(-\theta))^{\ell}\,(1-\theta)^{2\ell-1}}
  = \frac{\exp(\theta \ell)p_{\ell}(-\theta)}{(1-\theta)^{2\ell-1}}.
\]
Substituting this back, we get the explicit closed form
\begin{equation}\label{eq:zeta0}
  \zeta^{(j)}(0)
  = \frac{1}{\theta}
    \sum_{\ell=1}^{j}
      S(j,\ell)\,(-\theta)^{\ell}
      \frac{\,p_{\ell}(-\theta)}{(1-\theta)^{2\ell-1}}.
\end{equation}
This gives the value of $\zeta^{(j)}(g(u))$ at $u=0$ entirely in terms of $\theta$, the polynomials $p_{\ell}$, and Stirling numbers of the second kind. We conclude $\mathbb{E}((\dfrac{d^{j}\log K_t(u)}{du^j}\Big|_{u=0})^m|\mathcal{F}_{t-1}^{*})$ equals to
\begin{equation}
\mathbb{E}(\lambda_t^m|\mathcal{F}^{\ast}_{t-1})(\zeta^{(j)}(0))^m=
\left\{
\begin{array}{ll}
    \exp(m^2\frac{\sigma^2_\alpha}{2}+m f_t)(\zeta^{(j)}(0))^m,\, &f_t\in\mathcal{F}^{\ast}_{t-1}\\
    \exp(m^2\frac{\sigma^2_\alpha}{2}+m f_{t-1}+m^2\frac{\sigma^2_\epsilon}{2})(\zeta^{(j)}(0))^m,\,&f_t\notin\mathcal{F}^{\ast}_{t-1}
\end{array}
\right.
\end{equation}
\end{proof}

Before proving Proposition \ref{prop3}, we show a preliminary result.
\begin{proposition}\label{prop5}
Let $V_j\sim \mathcal{N}(0,\sigma^2)$ i.i.d. and define $X_j=\exp(V_j)$ then $\mathbb{V}(\sum_{i=1}^{N}\sum_{j=1,j\neq i}^{N}X_i X_j)=2N(N-1)\,\exp(2\sigma^2)\,\big(\exp(\sigma^2)-1\big)\,\big(2N+\exp((\sigma^2)-3\big)$
\end{proposition}

\begin{proof}
Let $Z_i=\sum_{j=1,j\neq i}^{N}X_i X_j$, following the standard decomposition $\mathbb{V}(\sum_{i=1}^{N}\sum_{j=1,j\neq i}^{N}X_i X_j)=$ $\sum_{i=1}^{N}\mathbb{V}(Z_i)$ $+\sum_{\ell=1}^{N}\sum_{k=1,k\neq \ell}^{N}\mathbb{C}ov\left(Z_\ell,Z_k\right)$. Regarding the elements $\mathbb{V}(Z_i)$ in first term, let $S_i=\sum_{j=1,j\neq i}^{N} X_j$, then
\begin{align}    
&\mathbb{V}(Z_i)=\mathbb{V}(X_i S_i)=\mathbb{E}(X_i^2)\mathbb{E}((S_i)^2)-\mathbb{E}(X_i)^2(\mathbb{E}(S_i))^2\\
&=\exp(2\sigma^2)((N-1)\exp(2\sigma^2)+(N-2)(N-1)\exp(\sigma^2))-\exp(2\sigma^2)(N-1)^2
\end{align}
Regarding the elements $\mathbb{C}ov\left(Z_\ell,Z_k\right)$ in the covariance term
\begin{align}    
&\mathbb{C}ov\left(X_\ell S_{\ell},X_k S_k\right)=\mathbb{E}(X_\ell X_k \sum_{j=1,j\neq \ell}^{N} \sum_{j'=1,j'\neq k}^{N} X_j X_j')-\mathbb{E}(X_\ell\sum_{j=1,j\neq \ell}^{N} X_j)\mathbb{E}(X_k\sum_{j=1,j\neq k}^{N} X_j)\\
&=3(N-2)\,\exp(2\sigma^2)\big(\exp(\sigma^2)-1\big)+ \exp(2\sigma^2)\big(\exp(2\sigma^2)-1\big)\\
&=\exp(2\sigma^2)\big(\exp(\sigma^2)-1\big)\big(\exp(\sigma^2)+3N-5\big),
\end{align}
where the last line follows from
\[\mathbb{C}ov(X_\ell X_j, X_k X_j')=
\left\{
\begin{array}{ll}
0 & j\neq k, j'\neq \ell\\
\exp(2\sigma^2)(\exp(\sigma^2)-1) & j= k, j'\neq \ell, j\neq j'\\
\exp(2\sigma^2)(\exp(\sigma^2)-1) & j\neq k, j'= \ell, j\neq j'\\
\exp(2\sigma^2)(\exp(2\sigma^2)-1) & j= k, j'= \ell, j\neq j'
\end{array}
\right.
\]
In conclusion $\mathbb{V}(\sum_{i=1}^{N}\sum_{j=1,j\neq i}^{N}X_i X_j)=N\exp(2\sigma^2)((N-1)\exp(2\sigma^2)+(N-2)(N-1)\exp(\sigma^2))-N\exp(2\sigma^2)(N-1)^2+N(N-1)(\exp(2\sigma^2)\big(\exp(\sigma^2)-1\big)\big(\exp(\sigma^2)+3N-5\big))$ $2N(N-1)\,\exp(2\sigma^2)\,\big(\exp(\sigma^2)-1\big)\,\big(2N+\exp(\sigma^2)-3\big)$.
\end{proof}
\begin{proof}[\textbf{Proof of Proposition \ref{prop3}}]
The two cases: i) $f_t\in \mathcal{F}_{t-1}$ and ii) $f_t\notin \mathcal{F}_{t-1}$ are discussed separately.
\par\noindent\textit{i)} If $f_t\in \mathcal{F}_{t-1}$, by the law of iterated expectations $\mathbb{E}(\bar{Y}_t|\mathcal{F}_{t-1})=\mathbb{E}(\mathbb{E}(\bar{Y}_t|\boldsymbol{\alpha}\vee \mathcal{F}_{t-1}^{\ast})|\mathcal{F}_{t-1}^{\ast})$ that is equal to
\begin{align}
&\sum_{i=1}^{N}\sum_{j \neq i}^{N}\mathbb{E}(\mu_{ijt}|\mathcal{F}_{t-1}^{\ast})=\sum_{i=1}^{N}\sum_{j \neq i }^{N}\mathbb{E}(\exp(\alpha_i)\exp(\alpha_j)|\mathcal{F}_{t-1}^{\ast})\exp(f_t)= N(N-1)\eta_t,
\end{align}
where $\eta_t=\exp(\sigma^2_\alpha+f_t)$ and $\mathbb{V}(\bar{Y}_t|\mathcal{F}_{t-1}^{\ast})=\mathbb{E}(\mathbb{V}(\bar{Y}_t|\boldsymbol{\alpha}\vee \mathcal{F}_{t-1}^{\ast})|\mathcal{F}_{t-1}^{\ast})+\mathbb{V}(\mathbb{E}(\bar{Y}_t|\boldsymbol{\alpha}\vee \mathcal{F}_{t-1}^{\ast})|\mathcal{F}_{t-1}^{\ast})$ is equal to
\begin{align}
&\mathbb{E}(\sum_{i=1}^{N}\sum_{j \neq i}^{N}\mathbb{V}(Y_{ijt}|\boldsymbol{\alpha}\vee \mathcal{F}_{t-1}^{\ast})|\mathcal{F}_{t-1}^{\ast})+\mathbb{V}(\sum_{i=1}^{N}\sum_{j\neq i}^{N}\mathbb{E}(Y_{ijt}|\boldsymbol{\alpha}\vee \mathcal{F}_{t-1}^{\ast})|\mathcal{F}_{t-1}^{\ast})\\
&=\mathbb{E}(\sum_{i=1}^{N}\sum_{j \neq i}^{N}\frac{\mu_{ijt}}{(1-\theta)^2}|\mathcal{F}_{t-1}^{\ast}) + \mathbb{V}(\sum_{i=1}^{N}\sum_{j \neq i}^{N}\mu_{ijt} |\mathcal{F}_{t-1}^{\ast})\\
&=\frac{N(N-1)}{(1-\theta)^2}\eta_t\left(1+2(1-\theta)^2(\exp(\sigma^2_\alpha)-1)(2N+\exp(\sigma^2_\alpha)-3)\eta_t\right),
\end{align}
where the last line follows from Proposition \ref{prop5}. The dispersion index $D_t=\mathbb{V}(\bar{Y}_t|\mathcal{F}_{t-1}^{\ast})/\mathbb{E}(\bar{Y}_t|\mathcal{F}_{t-1}^{\ast})$ is
$D_t=(1-\theta)^{-2}+2(\exp(\sigma^2_\alpha)-1)(2N+\exp(\sigma^2_\alpha)-3)\eta_t$.

\par\noindent\textit{ii)} If $f_t=f_{t-1}+\epsilon_t$, $\epsilon_t\sim\mathcal{N}(0,\sigma^2_{\varepsilon})$, then by the law of iterated expectations $\mathbb{E}(\bar{Y}_t|\mathcal{F}_{t-1})$ writes as
\begin{align}
&\mathbb{E}(\mathbb{E}(\bar{Y}_t|\boldsymbol{\alpha}\vee \mathcal{F}_{t-1})|\mathcal{F}_{t-1})=\frac{1}{1-\theta}\sum_{i=1}^{N}\sum_{j\neq i}^{N}\mathbb{E}(\exp(\alpha_i)\exp(\alpha_j)|\mathcal{F}_{t-1})\mathbb{E}(\exp(f_t)|\mathcal{F}_{t-1})=N(N-1)\eta_t,
\end{align}
where $\eta_t=\exp(\sigma^2_\alpha+f_{t-1}+\frac{\sigma^2_\epsilon}{2})$ and
$\mathbb{V}(\bar{Y}_t|\mathcal{F}_{t-1}^{\ast})=\mathbb{E}(\mathbb{V}(\bar{Y}_t|\boldsymbol{\alpha}\vee \mathcal{F}_{t-1}^{\ast})|\mathcal{F}_{t-1}^{\ast})+\mathbb{V}(\mathbb{E}(\bar{Y}_t|\boldsymbol{\alpha}\vee \mathcal{F}_{t-1}^{\ast})|\mathcal{F}_{t-1}^{\ast})$ is
\begin{align}
&\frac{1}{(1-\theta)^2}\sum_{i=1}^{N}\sum_{j=i+1}^{N}\mathbb{E}(\exp(\alpha_i+\alpha_j+f_t)|\mathcal{F}_{t-1}^{\ast})+ \sum_{i=1}^{N}\sum_{j=i+1}^{N}\mathbb{V}(\exp(\alpha_i+\alpha_j+f_t)|\mathcal{F}_{t-1}^{\ast})\\
&=\frac{N(N-1)}{(1-\theta)^2}\eta_{t}\left(1+2(1-\theta)^2(\exp(\sigma^2_\alpha)-1)(2N+\exp(\sigma^2_\alpha)-3)\eta_t\right).
\end{align}
The dispersion index is: $D_t=(1-\theta)^{-2}+2(\exp(\sigma^2_\alpha)-1)(2N+\exp(\sigma^2_\alpha)-3)\eta_t$
\end{proof}

\begin{proof}[Proof of Proposition \ref{propTaylor}]
Recall that $\lambda_{ijt}=\mu_{ijt}\rho^{-1/2}$ and 
$\mu_{ijt}=\exp\!\big(\alpha_i+\alpha_j+f_t-\|\mathbf{x}_{it}-\mathbf{x}_{jt}\|^2\big)$, and define $\Delta_{ijt}=\mathbf{x}_{it}-\mathbf{x}_{jt}$, 
$\ell_{it}(\mathbf{x}_{it})=\sum_{j\neq i}\ell_{ijt}(\mathbf{x}_{it})$, $\ell_{ijt}(\mathbf{x}_{it})
=\log p\!\left(y_{ijt}\mid \lambda_{ijt}(\mathbf{x}_{it}),\theta\right)$. By the chain rule, for each $j\neq i$, 
\begin{align}
C_{ijt}(\mathbf{x}_{it})=\frac{\partial \ell_{ijt}(\mathbf{x}_{it})}{\partial \mathbf{x}_{it}}
=
\frac{\partial \ell_{ijt}}{\partial \lambda_{ijt}}
\frac{\partial \lambda_{ijt}}{\partial \mu_{ijt}}
\frac{\partial \mu_{ijt}}{\partial \mathbf{x}_{it}},
\end{align}
with
$\frac{\partial \ell_{ijt}}{\partial \lambda_{ijt}}
=
1/\lambda_{ijt}
+(y_{ijt}-1)/(\lambda_{ijt}+\theta y_{ijt})-1$
$=s_{ijt}/\lambda_{ijt}$, $s_{ijt}(\lambda_{ijt})=1+(y_{ijt}-1)\lambda_{ijt}/(\lambda_{ijt}+\theta y_{ijt})-\lambda_{ijt}$, $\frac{\partial \lambda_{ijt}}{\partial \mu_{ijt}}=\rho^{-1/2}$, $\frac{\partial \mu_{ijt}}{\partial \mathbf{x}_{it}}
=
-2\mu_{ijt}\Delta_{ijt}$. Hence $C_{ijt}(\mathbf{x}_{it})=
-2\,s_{ijt}(\lambda_{ijt})\,\Delta_{ijt}$. Differentiating again yields the Hessian
\begin{equation}\label{hessian}
H_{ijt}(\mathbf{x}_{it})=\frac{\partial^2 \ell_{ijt}(\mathbf{x}_{it})}{\partial\mathbf{x}_{it}\partial\mathbf{x}_{it}'}
=
-2\,s_{ijt}(\lambda_{ijt})\,I_d
+
4\,((s_{ijt}(\lambda_{ijt})-1+\lambda_{ijt})\frac{\theta y_{ijt}}{\lambda_{ijt}+\theta y_{ijt}}-\lambda_{ijt})\,\Delta_{ijt}\Delta_{ijt}'.
\end{equation}
By Taylor's theorem, for each $j\neq i$ there exists
$\tau\in(0,1)$ and
$\xi_{ijt}=\tilde{\mathbf{x}}+\tau(\mathbf{x}_{it}-\tilde{\mathbf{x}})$
such that
\[
\ell_{ijt}(\mathbf{x}_{it})
=
\ell_{ijt}(\tilde{\mathbf{x}}_{it})
+
C_{ijt}(\tilde{\mathbf{x}})'(\mathbf{x}_{it}-\tilde{\mathbf{x}})
+\frac12(\mathbf{x}_{it}-\tilde{\mathbf{x}})'
H_{ijt}(\xi_{ijt})
(\mathbf{x}_{it}-\tilde{\mathbf{x}}),
\]
where $\tilde{\lambda}_{ijt}=\tilde{\mu}_{ijt}\rho^{-1/2}$ with
$\tilde{\mu}_{ijt}=\exp\!\big(\alpha_i+\alpha_j+f_t-\|\tilde{\mathbf{x}}-\mathbf{x}_{jt}\|^2\big)$. Summing over $j\neq i$ gives the stated expansion for $\ell_{it}(\mathbf{x}_{it})$. If the Hessian in Eq.~\ref{hessian} is bounded in a neighborhood of $\tilde{\mathbf{x}}_{it}$, then the remainder in the Taylors' expansion is $o\!\left(\|\mathbf{x}_{it}-\tilde{\mathbf{x}}_{it}\|\right)$ for $\mathbf{x}_{it}\to \tilde{\mathbf{x}}_{it}$. If $\tilde{\mathbf{x}}_{it}\in C_{\varepsilon}=\{\tilde{\mathbf{x}}:\tilde{\lambda}_{ijt}+\theta y_{ijt}\ge \varepsilon>0\}$,
then, since $\lambda_{ijt}(\mathbf{x}_{it})$ is continuous in $\mathbf{x}_{it}$, there exists
$\delta>0$ such that $\lambda_{ijt}(\mathbf{x}_{it})+\theta y_{ijt}\ge \varepsilon/2$ for all $\mathbf{x}_{it}\in B_{\delta}(\tilde{\mathbf{x}})=\{\mathbf{x}:\|\mathbf{x}-\tilde{\mathbf{x}}\|\le \delta\}$, which prevents the denominator from going to zero. Since $\lambda_{ijt}<\infty$ then $s_{ijt}(\lambda_{ijt})<\infty$ for all $\mathbf{x}_{it}\in B_{\delta}(\tilde{\mathbf{x}})\cap C_{\varepsilon/2}$,  and $||H_{ijt}(\xi_{ijt})||_{op}<M$ with 
\begin{equation}
M=\begin{cases}
2\sup|s_{ijt}(\lambda_{ijt})|+4(\sup |s_{ijt}(\lambda_{ijt})|+2\sup \lambda_{ijt}+1)\,\sup||\Delta_{ijt}||^2,&\theta\geq0\\
2\sup|s_{ijt}(\lambda_{ijt})|+4(|\theta|y(\sup |s_{ijt}(\lambda_{ijt})|+1+\sup\lambda_{ijt})/\varepsilon+\sup\lambda_{ijt})\,\sup||\Delta_{ijt}||^2,&\theta<0,
\end{cases}
\end{equation}
where the supremum is taken wrt to $\mathbf{x}_{it}$ in $B_{\delta}(\tilde{\mathbf{x}})\cap C_{\varepsilon}$. It follows that $|(\mathbf{x}_{it}-\tilde{\mathbf{x}})'
H_{ijt}(\xi_{ijt})
(\mathbf{x}_{it}-\tilde{\mathbf{x}})|\leq ||H_{ijt}(\xi_{ijt})||_{op} ||(\mathbf{x}_{it}-\tilde{\mathbf{x}})||^2\leq M ||(\mathbf{x}_{it}-\tilde{\mathbf{x}})||^2=o(||\mathbf{x}_{it}-\tilde{\mathbf{x}}||)$.
\end{proof}
\begin{proof}[\textbf{Proof of Corollary \ref{propTaylorProposal}}]
From the first-order log-Taylor expansion of the GP likelihood given in Proposition \ref{propTaylor}, the full conditional in \ref{fullx} can be approximated by
\begin{align}
    &q(\mathbf{x}_{it})\propto \exp(-\frac{1}{2}(\mathbf{x}_{it}'\Sigma_x^{-1} \mathbf{x}_{it}-2\mathbf{x}_{it}'\Sigma_x^{-1} \mathbf{x}_{i,t-1}+\mathbf{x}_{it}'\Sigma_x^{-1} \mathbf{x}_{it}-2\mathbf{x}_{it}'\Sigma_x^{-1} \mathbf{x}_{i,t+1}))\exp((\mathbf{x}_{it}-\tilde{\mathbf{x}})'C_{it}(\tilde{\mathbf{x}}))\\
    &\propto \exp(-\frac{1}{2}(2\mathbf{x}_{it}'\Sigma_x^{-1} \mathbf{x}_{it}-2\mathbf{x}_{it}'\Sigma_x^{-1}(\mathbf{x}_{i,t-1}+\mathbf{x}_{i,t+1}+\Sigma_x C_{it}(\tilde{\mathbf{x}}))\propto \mathcal{N}_{d}((\mathbf{x}_{i,t-1}+\mathbf{x}_{i,t+1}+\Sigma_x C_{it}(\tilde{\mathbf{x}}))/2,\Sigma_x/2)\nonumber 
\end{align}
\end{proof}


Before proving the identification results, let us introduce some notation and preliminaries. Denote with $\boldsymbol{\iota}_a \in \mathbb{R}^a$  the unit vector and with
$I_a$ the $a\times a$ identity matrix. Denote with $J_a = I_a - \frac{1}{a}\boldsymbol{\iota}_a \boldsymbol{\iota}_a'$
the centering matrix. Denote by $\mathcal{Z}\in\mathbb{R}^{n_1\times \cdots \times n_K}$
a $K$-th order tensor. For $k\in\{1,\dots,K\}$ and a fixed index $i_k\in\{1,\dots,n_k\}$,
the mode-$k$ slice of $\mathcal{Z}$ at position $i_k$
is defined as the $(K-1)$-th order tensor
$\mathcal{Z}_{: \, \cdots \, : \, i_k \, : \, \cdots \, :} \in \mathbb{R}^{n_1\times\cdots\times n_{k-1}\times n_{k+1}\times\cdots\times n_K}$,
whose entries are given by
$(\mathcal{Z}_{: \, \cdots \, : \, i_k \, : \, \cdots \, :})_{i_1\ldots i_{k-1} i_{k+1}\ldots i_K}
=\mathcal{Z}_{i_1\ldots i_{k-1} i_k i_{k+1}\ldots i_K}$. For $K=3$, the mode-1, mode-2 and mode-3 slices are commonly referred to as
horizontal, lateral, and frontal slices, respectively. 
For a matrix $A\in\mathbb{R}^{m_k\times n_k}$,
the mode-$k$ product $\mathcal{Z}\times_k A$
is defined as multiplication along the $k$-th mode:
$$(\mathcal{Z}\times_k A)_{i_1\ldots i_{k-1}\, j\, i_{k+1}\ldots i_K}
=
\sum_{i_k=1}^{n_k}
\mathcal{Z}_{i_1\ldots i_{k-1}i_k i_{k+1} \ldots i_K}\, A_{j i_k}.$$
The resulting object is a tensor in 
$\mathbb{R}^{n_1\times\cdots\times n_{k-1}\times m_k \times n_{k+1}\times\cdots\times n_K}$. Please refer to \citet{kolda2009tensor, hackbusch2012tensor} for a complete treatment of tensor algebra. We restate here Lemma 4 from \citet{loyal2023eigenmodel}, which will turn out to be useful later.
\begin{lemma}\label{lem:loyal}
For any $\mathbf{v}=\left(v_1, \ldots, v_a\right)' \in \mathbb{R}^a$, if $\mathbf{v} \boldsymbol{\iota}'_a \boldsymbol{\iota}_a+\boldsymbol{\iota}_a \mathbf{v}' \boldsymbol{\iota}_a=0$, then $\mathbf{v}=0$.
\end{lemma}

\begin{proof}[\textbf{Proof of Proposition \ref{lem:idm1}}]
Let $\log(M)\in\mathbb{R}^{N\times N\times T}$ satisfy $\log(M)
= \mathcal{Z}_1 +\mathcal{Z}_2+\mathcal{Z}_3$ where for each $t\in\{1,\dots,T\}$ we define the frontal slices of the tensors $\mathcal{Z}_1$, $\mathcal{Z}_2$, and $\mathcal{Z}_3$ as $
Z_{1,::t} = \boldsymbol{\alpha}\iota'_N, \quad Z_{2,::t} =  \boldsymbol{\iota}_N \boldsymbol{\alpha}', \quad Z_{3,::t} = f_t\boldsymbol{\iota}_N\iota'_N$, with $\mathbf{f} = (f_1, \ldots, f_t, \ldots, f_T)\in\mathbb{R}^{T}$ and $\boldsymbol{\alpha} = (\alpha_1, \ldots, \alpha_i, \ldots, \alpha_N)\in\mathbb{R}^{N}$. Let $\boldsymbol{\theta}_1=(\boldsymbol{\alpha},\mathbf{f})$ and $\tilde{\boldsymbol{\theta}}_1=(\tilde{\boldsymbol{\alpha}},\tilde{\mathbf{f}})$.
We say that $(\boldsymbol{\alpha},\mathbf{f})$ is identifiable if $\log(M(\boldsymbol{\theta}_1))=\log(M(\tilde{\boldsymbol{\theta}}_1))
\Longrightarrow
\boldsymbol{\alpha}=\tilde{\boldsymbol{\alpha}},\;\;\mathbf{f}=\tilde{\mathbf{f}} \quad \forall t$. Let $\mathcal{Z}_1+\mathcal{Z}_2+\mathcal{Z}_3 =  \tilde{\mathcal{Z}}_1+\tilde{\mathcal{Z}}_2+\tilde{\mathcal{Z}}_3$ with $\tilde{Z}_{1,::t} = \tilde{\boldsymbol{\alpha}}\iota'_N$, $\tilde{Z}_{2,::t} =  \boldsymbol{\iota}_N \tilde{\boldsymbol{\alpha}}'$, $Z_{3,::t} = \tilde{f}_t\boldsymbol{\iota}_N\iota'_N$, $\tilde{\mathbf{f}}\in\mathbb{R}^{T}$ and $\tilde{\boldsymbol{\alpha}}\in\mathbb{R}^{N}$. We first show that the dynamic effects are uniquely identified under A1).
Following the definition of the mode product \begin{align}
&\left((\mathcal{Z}_1+\mathcal{Z}_2+\mathcal{Z}_3)
      \times_1 \boldsymbol{\iota}_N \right) \times_2 \boldsymbol{\iota}_N =
\left((\tilde{\mathcal{Z}}_1+\tilde{\mathcal{Z}}_2+\tilde{\mathcal{Z}}_3)
    \times_2 \boldsymbol{\iota}_N \right) \times_2 \boldsymbol{\iota}_N\\
     &\Leftrightarrow N\boldsymbol{\alpha}'\boldsymbol{\iota}_N+N\boldsymbol{\alpha}'\boldsymbol{\iota}_N+N^2\mathbf{f}=N\tilde{\boldsymbol{\alpha}}'\boldsymbol{\iota}_N+N\tilde{\boldsymbol{\alpha}}'\boldsymbol{\iota}_N+N^2\tilde{\mathbf{f}}
\Leftrightarrow \mathbf{f} = \tilde{\mathbf{f}}.
\end{align}
where the last equation follows from A1) 
which implies that $\boldsymbol{\alpha}'\boldsymbol{\iota}_N = 0 $ and that $\tilde{\boldsymbol{\alpha}}'\boldsymbol{\iota}_N= 0 $. Finally, we show that the individual effects are uniquely identified, i.e. $\boldsymbol{\alpha}= \tilde{\boldsymbol{\alpha}}$. Given that $ \mathbf{f}_t = \tilde{\mathbf{f}}_t$, it follows $\mathcal{Z}_1+\mathcal{Z}_2+\mathcal{Z}_3 =  \tilde{\mathcal{Z}}_1+\tilde{\mathcal{Z}}_2+\tilde{\mathcal{Z}}_3\Leftrightarrow$, $      \mathcal{Z}_1+\mathcal{Z}_2 =  \tilde{\mathcal{Z}}_1+\tilde{\mathcal{Z}}_2$ and 
\begin{align}
\left((\mathcal{Z}_1+\mathcal{Z}_2
      \times_2 \boldsymbol{\iota}_N \right) \times_3 \iota_T &=
\left((\tilde{\mathcal{Z}}_1+\tilde{\mathcal{Z}}_2)
    \times_2 \boldsymbol{\iota}_N \right) \times_3 \iota_T\,\Leftrightarrow
TN\boldsymbol{\alpha}
+T\boldsymbol{\iota}_N(\boldsymbol{\alpha}'\boldsymbol{\iota}_N) &= TN \tilde{\boldsymbol{\alpha}}
+T\,\boldsymbol{\iota}_N(\tilde{\boldsymbol{\alpha}}'\boldsymbol{\iota}_N).
\end{align}

Let $\boldsymbol{v}= \boldsymbol{\alpha}-\tilde{\boldsymbol{\alpha}}$ and obtain
$N\boldsymbol{v}+\boldsymbol{\iota}_N(\boldsymbol{v}'\boldsymbol{\iota}_N)=0
\quad\Longleftrightarrow\quad
\boldsymbol{v}\,\boldsymbol{\iota}_N'\boldsymbol{\iota}_N+\boldsymbol{\iota}_N \boldsymbol{v}'\boldsymbol{\iota}_N=0$. Then, by Lemma \ref{lem:loyal} it must follow that $\boldsymbol{v} = \boldsymbol{0}_N$, thus $\boldsymbol{\alpha}=\tilde{\boldsymbol{\alpha}}$. Under the stated assumptions $\boldsymbol{\alpha}=\tilde{\boldsymbol{\alpha}}
\qquad
\mathbf{f}=\tilde{\mathbf{f}}.$
\end{proof}

\begin{proof}[\textbf{Proof of Proposition \ref{lem:idm2}}]
Let $\log(M)\in\mathbb{R}^{N\times N\times T}$ satisfy $\log(M)
= \mathcal{Z}_1 +\mathcal{Z}_2+\mathcal{Z}_3$ where for each $t\in\{1,\dots,T\}$ we define the frontal slices of the tensors $\mathcal{Z}_1$, $\mathcal{Z}_2$, and $\mathcal{Z}_3$ as $Z_{1,::t} = \boldsymbol{\alpha}\iota'_N$, $Z_{2,::t} =  \boldsymbol{\iota}_N \boldsymbol{\alpha}'$, $Z_{3,::t} = \boldsymbol{w}_t'\boldsymbol{\delta} \boldsymbol{\iota}_N \boldsymbol{\iota}'_N$
with $\boldsymbol{\alpha} = (\alpha_1, \ldots, \alpha_i, \ldots, \alpha_N)\in\mathbb{R}^{N}$, $\boldsymbol{\delta} \in \mathbb{R}^k$ and the set of known regressors $W = (\boldsymbol{w}_1', \ldots, \boldsymbol{w}_t', \ldots, \boldsymbol{w}_T')' \in \mathbb{R}^{T \times k}$.

Let $\boldsymbol{\theta}_2=(\boldsymbol{\alpha},\boldsymbol{\delta})$ and $\tilde{\boldsymbol{\theta}}_2=(\tilde{\boldsymbol{\alpha}},\tilde{\boldsymbol{\delta}})$.
We say that $(\boldsymbol{\alpha},\boldsymbol{\delta})$ is identifiable if
$\log(M(\boldsymbol{\theta}_2))=\log(M(\tilde{\boldsymbol{\theta}}_2))
\Longrightarrow
\boldsymbol{\alpha}=\tilde{\boldsymbol{\alpha}},\;\;\boldsymbol{\delta}=\tilde{\boldsymbol{\delta}} \quad  \forall t$. Let $\mathcal{Z}_1+\mathcal{Z}_2+\mathcal{Z}_3  =  \tilde{\mathcal{Z}}_1+\tilde{\mathcal{Z}}_2+\tilde{\mathcal{Z}}_3$ with $\tilde{Z}_{1,::t} = \tilde{\boldsymbol{\alpha}}\iota'_N$, $\tilde{Z}_{2,::t} =  \boldsymbol{\iota}_N \tilde{\boldsymbol{\alpha}}'$, $Z_{3,::t} = \boldsymbol{w}_t'\tilde{\boldsymbol{\delta}} \boldsymbol{\iota}_N \boldsymbol{\iota}'_N$,  $\tilde{\boldsymbol{\alpha}}\in\mathbb{R}^{N}$ and $\tilde{\boldsymbol{\delta}} \in\mathbb{R}^{k}$. We first show that the dynamic effects are uniquely identified under A1) and A2). We have $\left((\mathcal{Z}_1+\mathcal{Z}_2+\mathcal{Z}_3)
      \times_1 \boldsymbol{\iota}_N \right) \times_2 \boldsymbol{\iota}_N$ $=
\left((\tilde{\mathcal{Z}}_1+\tilde{\mathcal{Z}}_2+\tilde{\mathcal{Z}}_3)
    \times_2 \boldsymbol{\iota}_N \right) \times_2 \boldsymbol{\iota}_N$ is equivalent to
\begin{align}
     N\sum_{i = 1}^N\alpha_i +   N\sum_{i = 1}^N\alpha_i + N^2 W\boldsymbol{\delta} &=  N\sum_{i = 1}\tilde{\alpha}_i +   N\sum_{i = 1}^N\tilde{\alpha}_i + N^2 W\tilde{\boldsymbol{\delta}}\, \Leftrightarrow\,
      W\boldsymbol{\delta} = W\tilde{\boldsymbol{\delta}} \Leftrightarrow  \boldsymbol{\delta} = \tilde{\boldsymbol{\delta}},
\end{align}

where the last passage follows from A2).
Finally, we show that under A1) and A2), the individual effects vector is uniquely identified, i.e. $\boldsymbol{\alpha}= \tilde{\boldsymbol{\alpha}}$. Given that $  \boldsymbol{\delta}  =  \tilde{\boldsymbol{\delta}}$, it follows that
$  \mathcal{Z}_1+\mathcal{Z}_2 +\mathcal{Z}_3 =  \tilde{\mathcal{Z}}_1+\tilde{\mathcal{Z}}_2 +\tilde{\mathcal{Z}}_3 \Leftrightarrow 
\mathcal{Z}_1+\mathcal{Z}_2 =  \tilde{\mathcal{Z}}_1+\tilde{\mathcal{Z}}_2$
Exploiting the mode product notation and considering the following contraction
\begin{align}
\left((\mathcal{Z}_1+\mathcal{Z}_2
      \times_2 \boldsymbol{\iota}_N \right) \times_3 \iota_T &=
\left((\tilde{\mathcal{Z}}_1+\tilde{\mathcal{Z}}_2)
    \times_2 \boldsymbol{\iota}_N \right) \times_3 \iota_T\Leftrightarrow
TN\boldsymbol{\alpha}
+T\boldsymbol{\iota}_N(\boldsymbol{\alpha}'\boldsymbol{\iota}_N) &= TN \tilde{\boldsymbol{\alpha}}
+T\,\boldsymbol{\iota}_N(\tilde{\boldsymbol{\alpha}}'\boldsymbol{\iota}_N).
\end{align}
Let $\boldsymbol{v}= \boldsymbol{\alpha}-\tilde{\boldsymbol{\alpha}}$ and obtain
$N\boldsymbol{v}+\boldsymbol{\iota}_N(\boldsymbol{v}'\boldsymbol{\iota}_N)=0
\Longleftrightarrow
\boldsymbol{v}\,\boldsymbol{\iota}_N'\boldsymbol{\iota}_N+\boldsymbol{\iota}_N \boldsymbol{v}'\boldsymbol{\iota}_N=0$. By Lemma \ref{lem:loyal} it must follow that $\boldsymbol{v} = \boldsymbol{0}_N$, thus $\boldsymbol{\alpha}=\tilde{\boldsymbol{\alpha}}$.
Under the stated assumptions,
$\boldsymbol{\alpha}=\tilde{\boldsymbol{\alpha}}$ and $\boldsymbol{\delta} =\tilde{\boldsymbol{\delta}}$.
\end{proof}

\begin{proof}[\textbf{Proof of Proposition \ref{lem:idm3}}]
Let $\log(M)\in\mathbb{R}^{N\times N\times T}$ satisfy $
\log(M)
= \mathcal{Z}_1 +\mathcal{Z}_2+\mathcal{Z}_3
- \mathcal{D} = \mathcal{Z}_1 +\mathcal{Z}_2+\mathcal{Z}_3
- \mathcal{G}_1 -\mathcal{G}_2+\mathcal{G}_3$, with frontal slices $\mathcal{Z}_1$, $\mathcal{Z}_2$, $\mathcal{Z}_3$, $\mathcal{G}_1$, $\mathcal{G}_2$, and  $\mathcal{G}_3$ as: $Z_{1,::t} = \boldsymbol{\alpha}\iota'_N, \quad Z_{2,::t} =  \boldsymbol{\iota}_N \boldsymbol{\alpha}', \quad Z_{3,::t} = f_t\boldsymbol{\iota}_N\iota'_N$, with $\mathbf{f} = (f_1, \ldots, f_t, \ldots, f_T)\in\mathbb{R}^{T}$ and $\boldsymbol{\alpha} = (\alpha_1, \ldots, \alpha_i, \ldots, \alpha_N)\in\mathbb{R}^{N}$, while $\mathcal{D}$ denotes the tensor of squared Euclidean slices for which $ D_{::t} =  G_{1,::t} +  G_{2,::t} - G_{3,::t}$,$G_{1,::t} =  \mathbf{g}_t\iota'_N, $, $G_{2,::t} =   \boldsymbol{\iota}_N \mathbf{g}_t'$, and $ G_{3,::t} =   2X_{t}X_t'$.
with $ \mathbf{g}_t =(\mathbf{x}_{1t}'\mathbf{x}_{1t}, \ldots , \mathbf{x}_{Nt}'\mathbf{x}_{Nt})
=\operatorname{diag}(X_tX_t')$, $X_t\in\mathbb{R}^{N\times D}$ for $t = 1, \ldots, T$.

Let $\boldsymbol{\theta}_3=(\boldsymbol{\alpha},\mathbf{f},\{X_t\}_{t=1}^T)$ and $\tilde{\boldsymbol{\theta}}_3=(\tilde{\boldsymbol{\alpha}},\tilde{\mathbf{f}},\{\tilde X_t\}_{t=1}^T)$.
We say that $(\boldsymbol{\alpha},\mathbf{f},\{X_tX_t'\}_{t=1}^T)$ is identifiable (and $X_t$ identifiable up to rotation) if $
\log(M(\boldsymbol{\theta}_3))=\log(M(\tilde{\boldsymbol{\theta}}_3))
\quad\Longrightarrow\quad
\boldsymbol{\alpha}=\tilde{\boldsymbol{\alpha}},\;\;\mathbf{f}=\tilde{\mathbf{f}},\;\; X_tX_t'=\tilde X_t\tilde X_t' \ \forall t$.
Moreover, if $\operatorname{rank}(X_t)=\operatorname{rank}(\tilde X_t)=D$, then for each $t$ there exists an orthogonal matrix $P_t$ such that $\tilde X_t=X_tP_t$,  with $P_t'P_t=I_D$. Let $\mathcal{Z}_1+\mathcal{Z}_2+\mathcal{Z}_3 - \mathcal{D}  =  \tilde{\mathcal{Z}}_1+\tilde{\mathcal{Z}}_2+\tilde{\mathcal{Z}}_3 - \tilde{\mathcal{D}} \Leftrightarrow \mathcal{Z}_1+\mathcal{Z}_2+\mathcal{Z}_3 - \mathcal{G}_1 - \mathcal{G}_2 + 2\mathcal{G}_3=  \tilde{\mathcal{Z}}_1+\tilde{\mathcal{Z}}_2+\tilde{\mathcal{Z}}_3 - \tilde{\mathcal{G}}_1 - \tilde{\mathcal{G}}_2 + 2\tilde{\mathcal{G}}_3$, with $\tilde{Z}_{1,::t} = \tilde{\boldsymbol{\alpha}}\iota'_N$, $\tilde{Z}_{2,::t} =  \boldsymbol{\iota}_N \tilde{\boldsymbol{\alpha}}'$, $Z_{3,::t} = \tilde{\mathbf{f}}_t\boldsymbol{\iota}_N\iota'_N$, $\tilde{G}_{1,::t} =  \tilde{\mathbf{g}}_t\iota'_N$,$ \tilde{G}_{2,::t} =   \boldsymbol{\iota}_N \tilde{\mathbf{g}}_t'$, $\tilde{G}_{3,::t} =   2X_{t}X_t'$ and with $\tilde{\mathbf{g}}_t =(\tilde{\mathbf{x}}_{1t}'\tilde{\mathbf{x}}_{1t}, \ldots , \tilde{\mathbf{x}}_{Nt}'\tilde{\mathbf{x}}_{Nt})
=\operatorname{diag}(\tilde X_t\tilde X_t')$, $\tilde X_t\in\mathbb{R}^{N\times D}$, $\tilde{\mathbf{f}}\in\mathbb{R}^{T}$ and $\tilde{\boldsymbol{\alpha}}\in\mathbb{R}^{N}$.

We first show that the tensor of squared Euclidean distances $\mathcal{D} = \tilde{\mathcal{D}}$ is uniquely identified under A2) and that the latent coordinates are identified up to rotation, reflection and translation under A2) and A3). Multiply the first and second modes by $J_N$: 
$\left(\left(\mathcal{Z}_1+\mathcal{Z}_2+\mathcal{Z}_3 - \mathcal{G}_1 - \mathcal{G}_2 + 2\mathcal{G}_3\right) \times_{1} J_N \right) \times_{2} J_N $ $=   \left(\left(\tilde{\mathcal{Z}}_1+\tilde{\mathcal{Z}}_2+\tilde{\mathcal{Z}}_3 - \tilde{\mathcal{G}}_1 - \tilde{\mathcal{G}}_2 + 2\tilde{\mathcal{G}}_3\right) \times_{1} J_N \right) \times_{2} J_N$. This implies that for each slice 
{\footnotesize
\begin{align}
J_N\left(\boldsymbol{\alpha}\boldsymbol{\iota}_N' +  \boldsymbol{\iota}_N\boldsymbol{\alpha}' + f_t \boldsymbol{\iota}_N\boldsymbol{\iota}_N' -  
\boldsymbol{g}_t\boldsymbol{\iota}_N'  - \boldsymbol{\iota}_N\boldsymbol{g}'_t + 2X_tX'_t\right)J_N
=
J_N\left(\tilde{\boldsymbol{\alpha}}\boldsymbol{\iota}_N' +  \boldsymbol{\iota}_N\tilde{\boldsymbol{\alpha}}' + \tilde{f}_t \boldsymbol{\iota}_N\boldsymbol{\iota}_N' -    
\tilde{\boldsymbol{g}}_t\boldsymbol{\iota}_N'  - \boldsymbol{\iota}_N\tilde{\boldsymbol{g}}'_t + 2\tilde{X}_t\tilde{X}'_t\right)J_N
\end{align}}
Since $J_N\boldsymbol{\iota}_N = \boldsymbol{0}_N$ we have that $J_N(     
\boldsymbol{g}_t\boldsymbol{\iota}_N' + \boldsymbol{\iota}_N\boldsymbol{g}'_t - 2X_tX'_t  - \tilde{\boldsymbol{g}}_t\boldsymbol{\iota}_N'$ $
 - \boldsymbol{\iota}_N\tilde{\boldsymbol{g}}'_t  + 2\tilde X_t\tilde X'_t)J_N$ $= J_N( 
 \boldsymbol{\alpha}\boldsymbol{\iota}_N' +  \boldsymbol{\iota}_N\boldsymbol{\alpha}' + f_t \boldsymbol{\iota}_N\boldsymbol{\iota}_N' -\tilde{\boldsymbol{\alpha}}\boldsymbol{\iota}_N'
 -\boldsymbol{\iota}_N\tilde{\boldsymbol{\alpha}}'
 - \tilde f_t \boldsymbol{\iota}_N\boldsymbol{\iota}_N')J_N$ is equivalent to
{\footnotesize
\begin{align}
 J_N\left(     
\boldsymbol{g}_t\boldsymbol{\iota}_N' + \boldsymbol{\iota}_N\boldsymbol{g}'_t - 2X_tX'_t  - \tilde{\boldsymbol{g}}_t\boldsymbol{\iota}_N'
 - \boldsymbol{\iota}_N\tilde{\boldsymbol{g}}'_t  + 2\tilde X_t\tilde X'_t\right )J_N &=  \mathbf{O}_{N} \Leftrightarrow
J_N\left( - 2X_tX'_t  + 2\tilde X_t\tilde X'_t\right )J_N =  \mathbf{O}_{N}\\
-2J_N\left(X_tX'_t\right )J_N &=  -2J_N\left(\tilde X_t\tilde X'_t\right)J_N \Leftrightarrow
\left(X_tX'_t\right )=  \left(\tilde X_t\tilde X'_t\right),
\end{align}}
where the last step follows from A2) and exploiting the fact that $J_N$ is idempotent.

Note that  $X_tX_t'=\tilde X_t\tilde X_t'$  implies that $ \mathbf{g}_t 
=\operatorname{diag}(X_tX_t') = \operatorname{diag}(\tilde X_t\tilde X_t') =  \tilde{\mathbf{g}}_t$ for all $t = 1, \ldots, T$. Thus, the tensor of squared euclidean distances is uniquely identified as $\mathcal{D} = \tilde{\mathcal{D}}$. Finally, since $\operatorname{rank}(X_t)=\operatorname{rank}(\tilde X_t)= d$ under A3), for each $t$ there exists an orthogonal matrix $P_t$ such that $\tilde X_t=X_tP_t$ with $P_t'P_t=I_d$ which makes the latent coordinates identified up to rotation and translation. We now show that the individual effects are uniquely identified under A1) and A2).
Consider the contraction
{\footnotesize
\begin{align}
&\left((\mathcal{Z}_1+\mathcal{Z}_2+\mathcal{Z}_3
     - \mathcal{G}_1 - \mathcal{G}_2 + 2\mathcal{G}_3) \times_{2}
     \boldsymbol{\iota}_N\right)
     \times_{3}\boldsymbol{\iota}_T=
\left((\tilde{\mathcal{Z}}_1+\tilde{\mathcal{Z}}_2+\tilde{\mathcal{Z}}_3
     - \tilde{\mathcal{G}}_1 - \tilde{\mathcal{G}}_2 + 2\tilde{\mathcal{G}}_3)
     \times_{2}
     \boldsymbol{\iota}_N\right)
     \times_{3}\boldsymbol{\iota}_T \,\Longleftrightarrow\\
&\quad TN\boldsymbol{\alpha}
+T\boldsymbol{\iota}_N(\boldsymbol{\alpha}'\boldsymbol{\iota}_N)
    +N\Big(\sum_{t=1}^T f_t\Big)\boldsymbol{\iota}_N
    -N\sum_{t=1}^T g_t
    -\boldsymbol{\iota}_N\sum_{t=1}^T (g_t'\boldsymbol{\iota}_N)
     +2\sum_{t=1}^T X_t(X_t'\boldsymbol{\iota}_N)\\
&\qquad= TN \tilde{\boldsymbol{\alpha}}
    +T\,\boldsymbol{\iota}_N(\tilde{\boldsymbol{\alpha}}'\boldsymbol{\iota}_N)
    +N\Big(\sum_{t=1}^T\tilde{f}_t\Big)\boldsymbol{\iota}_N
    -N\sum_{t=1}^T\tilde{g}_t
    -\boldsymbol{\iota}_N\sum_{t=1}^T (\tilde{g}_t'\boldsymbol{\iota}_N)
     +2\sum_{t=1}^T \tilde{X}_t(\tilde{X}_t'\boldsymbol{\iota}_N).
\end{align}}
Given that $\mathbf{f}'\boldsymbol{\iota}_T = \sum_{t}^Tf_t = 0$ and $\tilde{\mathbf{f}}'\boldsymbol{\iota}_T = \sum_{t}^T\tilde{f}_t = 0$ and that $\mathcal{D} = \tilde{\mathcal{D}}$, we have

{\footnotesize
\begin{align}
TN\boldsymbol{\alpha}
+T\boldsymbol{\iota}_N(\boldsymbol{\alpha}'\boldsymbol{\iota}_N)  &= TN \tilde{\boldsymbol{\alpha}}
+T\,\boldsymbol{\iota}_N(\tilde{\boldsymbol{\alpha}}'\boldsymbol{\iota}_N) \Leftrightarrow
    TN(\boldsymbol{\alpha}-\tilde{\boldsymbol{\alpha}})+T\boldsymbol{\iota}_N\big((\boldsymbol{\alpha} -\tilde{\boldsymbol{\alpha}}')'\boldsymbol{\iota}_N\big)=0
\end{align}}

Let $\boldsymbol{v}= \boldsymbol{\alpha}-\tilde{\boldsymbol{\alpha}}$ and obtain $N\boldsymbol{v}+\boldsymbol{\iota}_N(\boldsymbol{v}'\boldsymbol{\iota}_N)=0
\Longleftrightarrow
\boldsymbol{v}\,\boldsymbol{\iota}_N'\boldsymbol{\iota}_N+\boldsymbol{\iota}_N \boldsymbol{v}'\boldsymbol{\iota}_N=0$. By Lemma \ref{lem:loyal} it must follow that $\boldsymbol{v} = \boldsymbol{0}_N$, thus $\boldsymbol{\alpha}=\tilde{\boldsymbol{\alpha}}$. Finally, we show that the dynamic latent factor is identified up to translation, i.e. $\mathbf{f} = \tilde{\mathbf{f}}$. From $\mathcal{Z}_1+\mathcal{Z}_2+\mathcal{Z}_3 - \mathcal{D}  =  \tilde{\mathcal{Z}}_1+\tilde{\mathcal{Z}}_2+\tilde{\mathcal{Z}}_3 - \tilde{\mathcal{D}}$ and given that $\mathcal{D} = \tilde{\mathcal{D}}$, and  $\boldsymbol{\alpha} = \tilde{\boldsymbol{\alpha}}$ we have that $ \mathcal{Z}_3   =  \tilde{\mathcal{Z}}_3$. Considering the contraction $\left(\mathcal{Z}_3 \times_{1}\boldsymbol{\iota}_N\right)\times_{2}\boldsymbol{\iota}_N  =  \left(\tilde{\mathcal{Z}}_3 \times_1 \boldsymbol{\iota}_N\right)\times_{2}\boldsymbol{\iota}_N$, it immediately follows that $ \mathbf{f} = \tilde{\mathbf{f}}$. Under the stated assumptions, $
X_tX_t'=\tilde X_t\tilde X_t' \ \forall t$,$\boldsymbol{\alpha}=\tilde{\boldsymbol{\alpha}}$, $\mathbf{f}=\tilde{\mathbf{f}}$, and if $\operatorname{rank}(X_t)=d$, then for each $t$ there exists an orthogonal matrix $P_t$ such that $\tilde X_t=X_tP_t$.
\end{proof}

\clearpage

\renewcommand{\theequation}{B.\arabic{equation}}
\renewcommand{\thefigure}{B.\arabic{figure}}
\renewcommand{\thetable}{B.\arabic{table}}
\setcounter{equation}{0}
\setcounter{figure}{0}
\setcounter{table}{0}

\section{Simulation Results}\label{app:sim_res}

We generate randomly synthetic datasets of observable interactions among $N = 40$ nodes over a small number of time instances, $ T = 8$, or a large number of time instances, $ T = 80$. We set the dispersion parameter $\zeta = 3$, implying $\rho \approx 20.33$ and $\theta  \approx 0.78$, and we draw the individual effects $\alpha_i \overset{iid}{\sim} \mathcal{N}(\mu_\alpha, \sigma^2_\alpha)$. 

For specification $\mathcal{M}_1$, we assume $f_t = f_{t-1} + \epsilon_t$, $\epsilon_t \sim \mathcal{N}(0, 0.01)$ iid for $t = 1, \ldots, T$  with $f_0 \sim \mathcal{N}(0, 5^2)$, $\mu_\alpha = 0$, $\sigma^2_\alpha = 0.01$. For specification $\mathcal{M}_2$, we  assume an AR(2) dynamics with $\delta_1 = 0.7$ and $\delta_2 = 0.1$ so that $\log \mu_{ijt} 
= \alpha_i + \alpha_j +
 0.7\log(\tilde{y}_{t-k}) + 0.1\log(\tilde{y}_{t-k})$ for $t = 1, \ldots, T$ with $\log(\tilde{y}_{0}) = \log(\tilde{y}_{-1}) = 1.2$, and $\mu_\alpha = 0$,  $\sigma^2_\alpha = 0.01$. For specification $\mathcal{M}_3$, we assume $f_t = f_{t-1} + \epsilon_t$ for $t = 1, \ldots, T$ with $\epsilon_t \sim \mathcal{N}(0, 0.01)$ and with $f_0 \sim \mathcal{N}(0, 1)$,
$\mathbf{x}_{it-1}+\boldsymbol{\nu}_{it},\, \boldsymbol{\nu}_{it}\overset{iid}{\sim}\mathcal{N}_d(\mathbf{0},\Sigma_x)$, $\boldsymbol{x}_{i0}\overset{iid}{\sim}\mathcal{N}_d(\mathbf{0},\Sigma_x)$,  with  $\Sigma_x = 0.25I_d$, $d = 2$,  $\mu_\alpha = 2$, $\sigma^2_\alpha = 0.025$. We generate a series of $T$ count adjacency matrices by simulating from the Generalized Poisson data-generating processes described in \ref{sec:model}. We run the algorithm for 5,000 iterations, using the first 2,000 as burn-in samples and thinning every 5 iterations. 

Graphical inspection of Figures \ref{fig:mcmc}-\ref{fig:acf} reveals convergence to the true parameter values and a good mixing of the MCMC chains. We also report in Tab. \ref{tab:diagnostics} some standard convergence diagnostics (ESS, Geweke CD), which confirm the findings of the graphical inspection. In all models, the posterior distributions accurately recover the true parameter values used in the data-generating process (Figure \ref{fig:sim_ex_m1}). Specifically, for $\mathcal{M}_1$ (top) we recover $\sigma_\epsilon^2$, $\zeta$, and $\theta$; for $\mathcal{M}_2$ (middle) the dynamic coefficients $\delta_1$ and $\delta_2$ as well as the dispersion parameter $\theta$; and for $\mathcal{M}_3$ (bottom) $\zeta$ and $\theta$. As shown in Figure \ref{fig:underdispersion}, the true value of the dispersion parameter $\rho$ is estimated accurately for $\mathcal{M}_1$ and $\mathcal{M}_2$, whereas some bias appears for $\mathcal{M}_3$, likely because the latent coordinates introduce additional dispersion.

\begin{table}[h!]
\resizebox{\textwidth}{!}{
\begin{tabular}{lccccccccccccc}
\hline
& \multicolumn{13}{c}{2000 iteration Burn-in}\\
 \hline

          & \multicolumn{4}{c}{$\mathcal{M}_1$}                                              & \multicolumn{4}{c}{$\mathcal{M}_2$}                 & \multicolumn{5}{c}{$\mathcal{M}_3$}                                     \\ \cline{2-14} 
          & $\theta$ & $\overline{\alpha}$ & $\overline{f}$ & $\sigma^2_\epsilon$ & $\theta$ & $\overline{\alpha}$ & $\delta_1$ & $\delta_2$ & $\theta$ & $\overline{\alpha}$ & $\overline{\mathbf{x}}_1$ & $\overline{\mathbf{x}}_2$ 
          & $\overline{f}$
          \\ \cline{2-14} 
Geweke CD p-val &      0.39    &        0.24        &    0.21       &     0.06               &  0.072        &     0.25           &     0.43       &     0.012       &     0.39       &         0.17      &    0.15                &               0.12     &   0.26      \\
ESS \%    &     19.62     &  15.85              &   8.76        &       6.52                                   &    20.42      &      15.89          &     6.49       &     6.99       &     7.2     &            7.8   &    7.49                &     8.04    &  4.4            \\ 
\hline
& \multicolumn{13}{c}{2000 iteration Burn-in and thinning every 5 iterations}
 \\
 \hline

          & \multicolumn{4}{c}{$\mathcal{M}_1$}                                              & \multicolumn{4}{c}{$\mathcal{M}_2$}                 & \multicolumn{5}{c}{$\mathcal{M}_3$}                                     \\ \cline{2-14} 
          & $\theta$ & $\overline{\alpha}$ & $\overline{f}$ & $\sigma^2_\epsilon$ & $\theta$ & $\overline{\alpha}$ & $\delta_1$ & $\delta_2$ & $\theta$ & $\overline{\alpha}$ & $\overline{\mathbf{x}}_1$ & $\overline{\mathbf{x}}_2$ &  $\overline{f}$ \\ \cline{2-14} 
Geweke CD p-val &   0.31 & 0.28 & 0.23 & 0.15&   0.18    &        0.24        &    0.32       &     0.03       &   0.33 & 0.22 &  0.16 &  0.14   & 0.27\\
ESS \%  & 71.30 & 66.96& 39.37& 25.42 &     73.21     &  66.52              &   35.07        &       26.38   &   36 &    38.84   &   32.90   &     34.12   &    22.13                           \\

\hline
\end{tabular}}
\caption{Effective Sample Size over the number of draws (ESS) and Convergence Diagnostic p-value (CD) as defined in \citet{geweke1991evaluating} for the parameters of the three models, $\mathcal{M}_1$, $\mathcal{M}_2$, and  $\mathcal{M}_3$. Sampling carried out over 5'000 iterations with a burn-in of 2'000 iterations (top panel) and thinning (bottom panel).}
\label{tab:diagnostics}
\end{table}

Regarding the latent variables, the results indicate that the inference procedure performs well. Panel (a) of Figure \ref{fig:sim_ex_m2} reports boxplots of the node-specific parameters $\alpha_i$, $i=1,\ldots,N$ (top), and of the posterior draws for the latent factor $f_t$, $t=1,\ldots,T$ (bottom). Panel (b) reports boxplots of the $\alpha_i$’s under $\mathcal{M}_2$. Panel (c) shows boxplots of $\alpha_i$, $i=1,\ldots,N$ (top), and posterior credible ellipses for the latent coordinates $\mathbf{x}_{it}$, $i=1,\ldots,N$, $t=1,\ldots,T$, at eight time points (bottom). Red crosses denote the true latent positions for model $\mathcal{M}_3$. Across all panels, the posterior summaries (blue) cover the true parameter values (red).


\clearpage
\begin{figure}[h!]
    \centering
    \resizebox{0.8\textwidth}{!}{
    \begin{tabular}{ccc}
\multicolumn{3}{c}{\small Model $\mathcal{M}_1$ (Factor Model)}\\
    \includegraphics[trim=0.2cm 0.2cm 0 0, clip, width=0.32\linewidth]{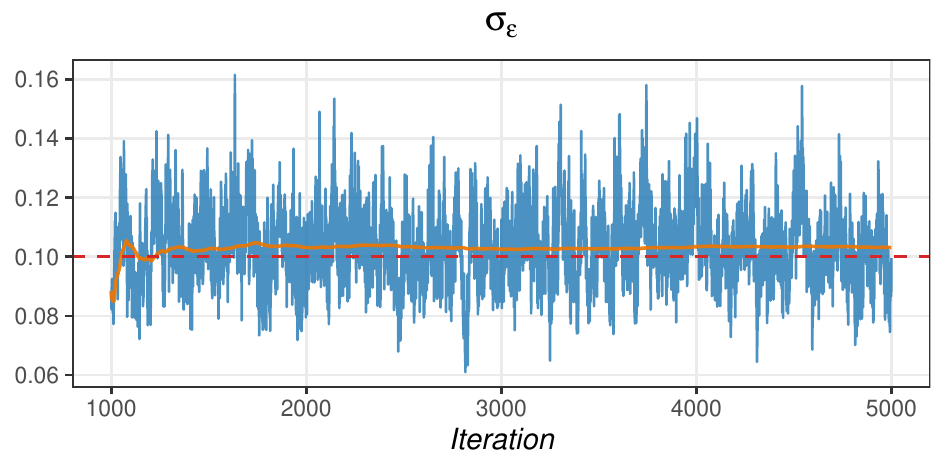} & 
    \includegraphics[trim=0.2cm 0.2cm 0 0, clip, width=0.32\linewidth]{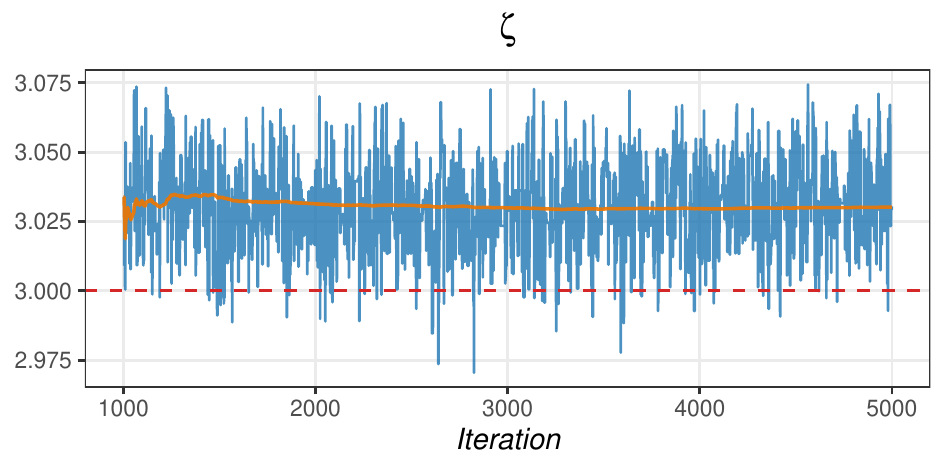} & 
    \includegraphics[trim=0.2cm 0.2cm 0 0, clip, width=0.32\linewidth]{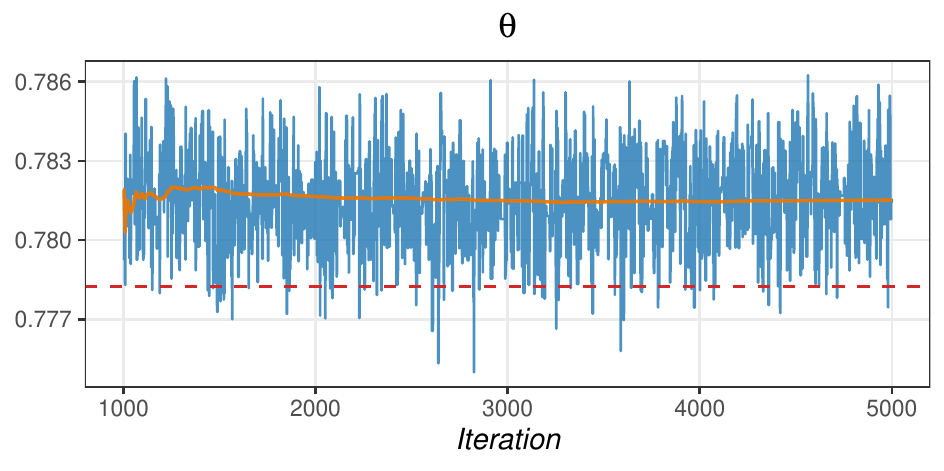} \\
    
     \includegraphics[trim=0.2cm 0.2cm 0 0, clip, width=0.32\linewidth]{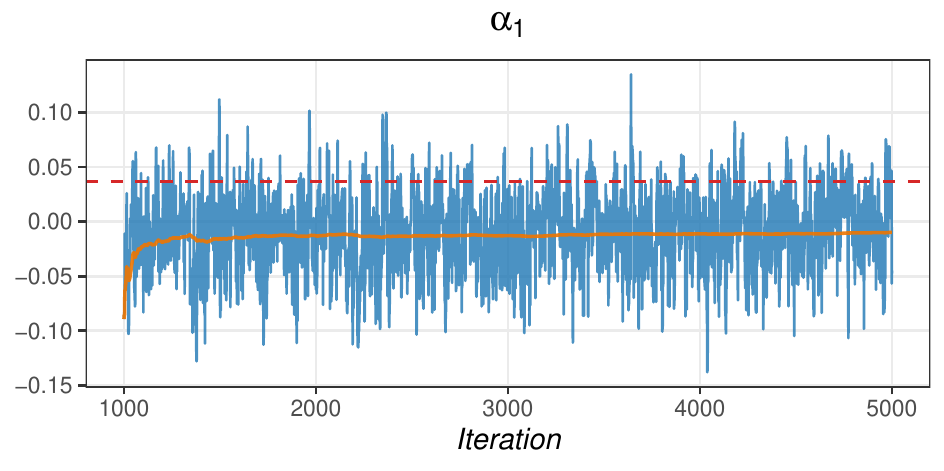} & 
    \includegraphics[trim=0.2cm 0.2cm 0 0, clip, width=0.32\linewidth]{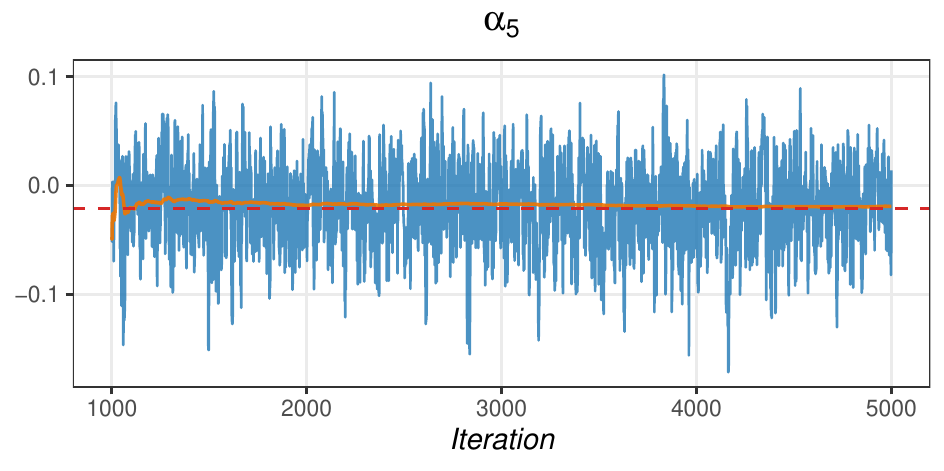} & 
    \includegraphics[trim=0.2cm 0.2cm 0 0, clip, width=0.32\linewidth]{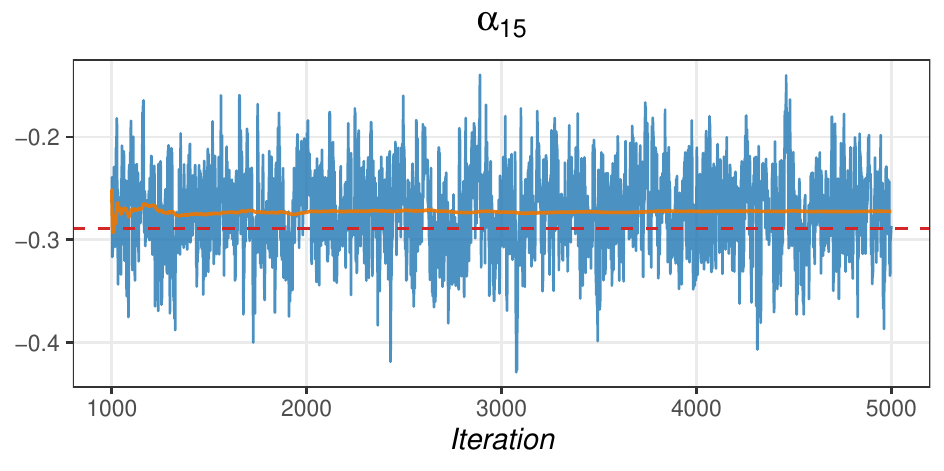} \\
        
    \includegraphics[trim=0.2cm 0.2cm 0 0, clip, width=0.32\linewidth]{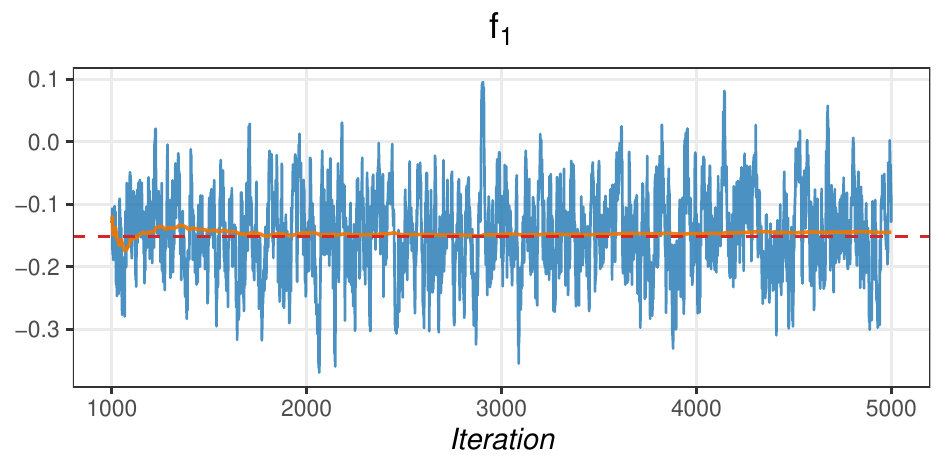} & 
    \includegraphics[trim=0.2cm 0.2cm 0 0, clip, width=0.32\linewidth]{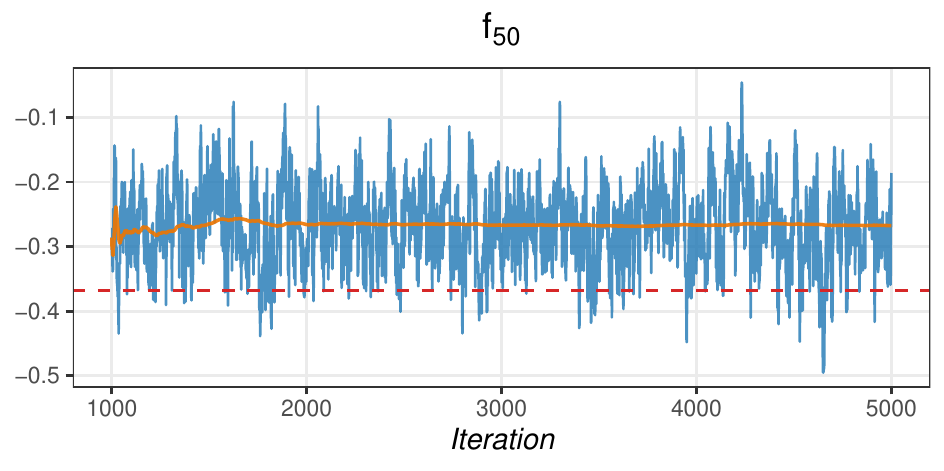} & 
    \includegraphics[trim=0.2cm 0.2cm 0 0, clip, width=0.32\linewidth]{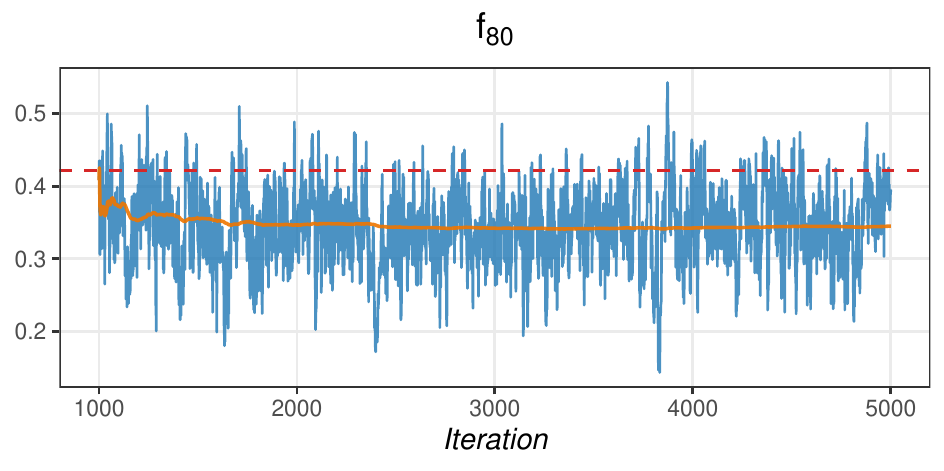}
\\

    \multicolumn{3}{c}{\small Model $\mathcal{M}_2$ (Autoregressive Model)}\\
    \includegraphics[trim=0.2cm 0.2cm 0 0, clip, width=0.32\linewidth]{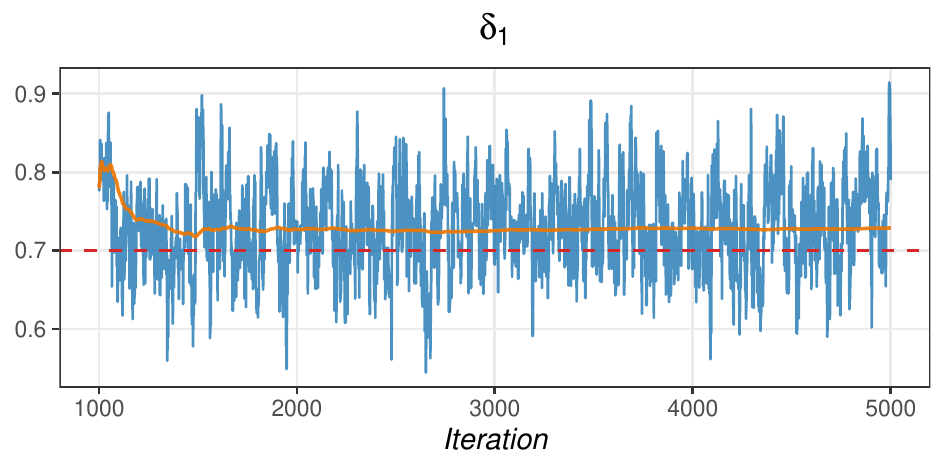} & 
    \includegraphics[trim=0.2cm 0.2cm 0 0, clip, width=0.32\linewidth]{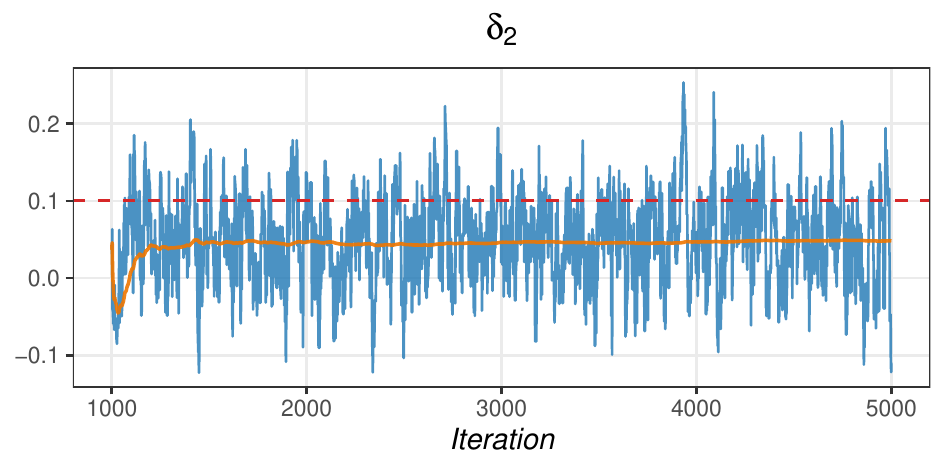} & 
    \includegraphics[trim=0.2cm 0.2cm 0 0, clip, width=0.32\linewidth]{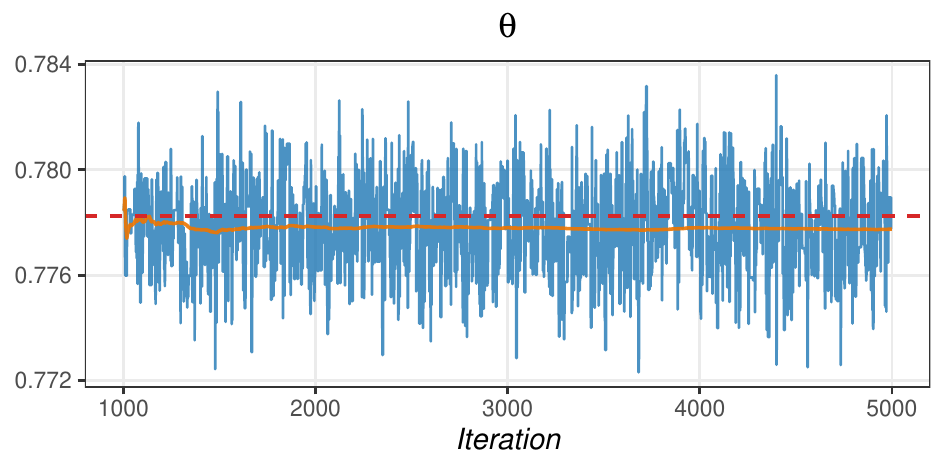} \\
    
     \includegraphics[trim=0.2cm 0.2cm 0 0, clip, width=0.32\linewidth]{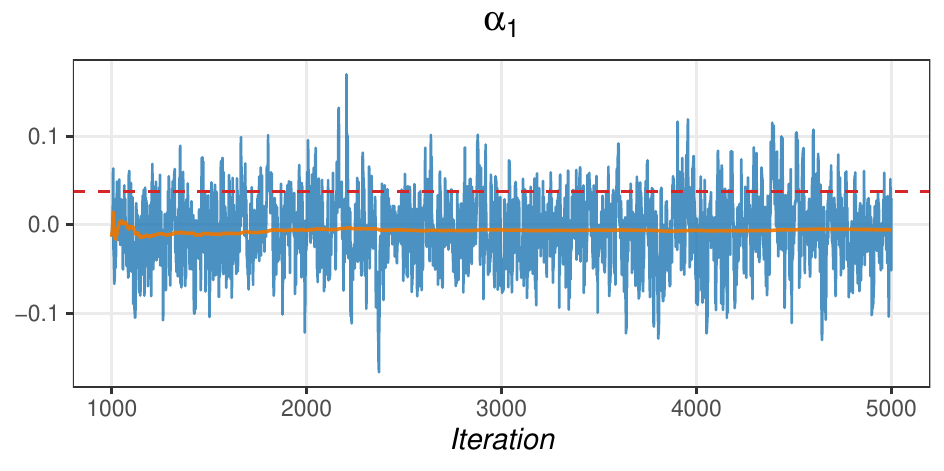} & 
    \includegraphics[trim=0.2cm 0.2cm 0 0, clip, width=0.32\linewidth]{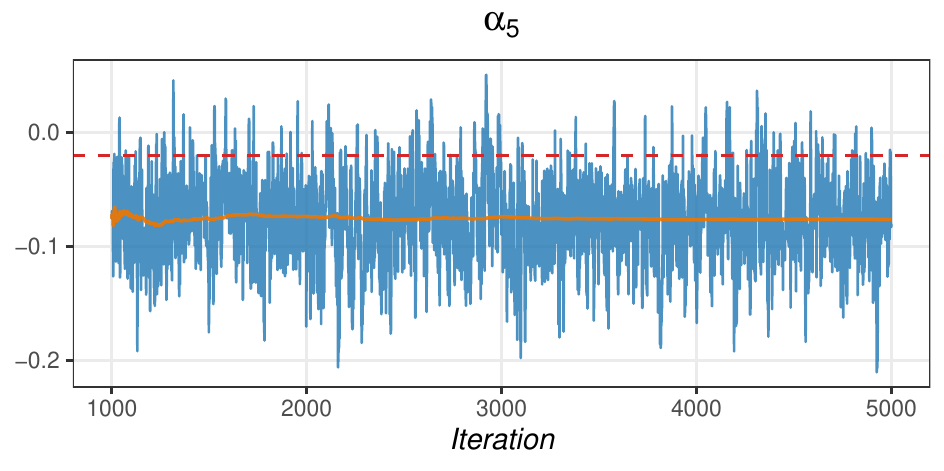} & 
    \includegraphics[trim=0.2cm 0.2cm 0 0, clip, width=0.32\linewidth]{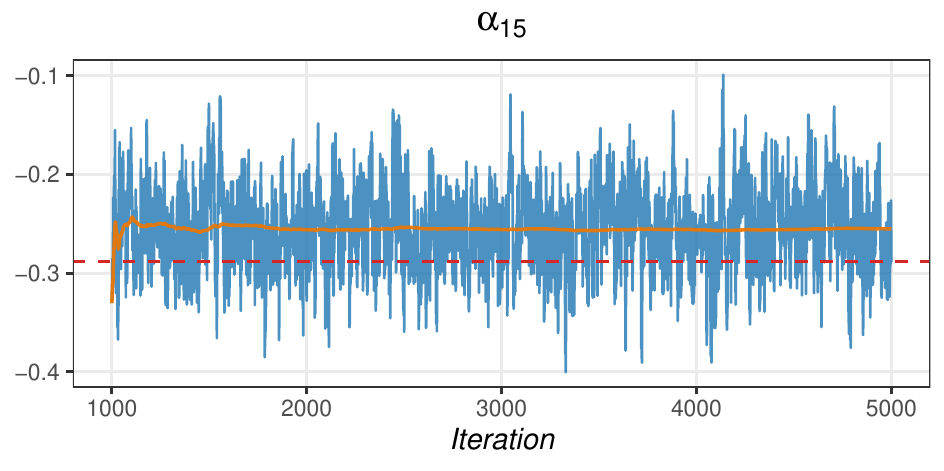} \\

\multicolumn{3}{c}{\small Model $\mathcal{M}_3$ (Latent Space Model)}\\
\includegraphics[trim=0.2cm 0.2cm 0 0, clip, width=0.32\linewidth]{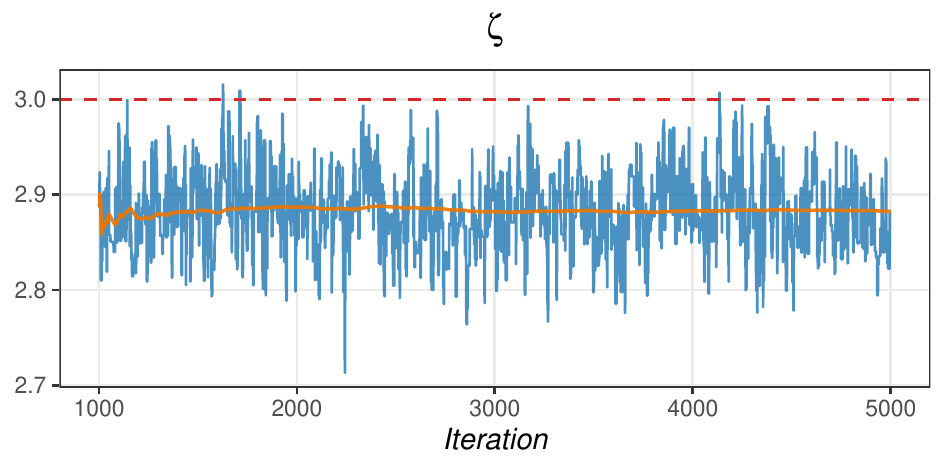} & 
    \includegraphics[trim=0.2cm 0.2cm 0 0, clip, width=0.32\linewidth]{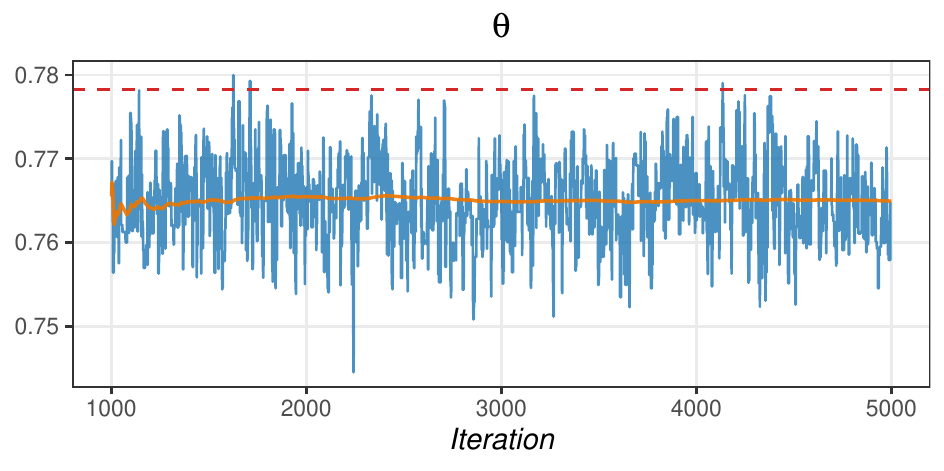} & 
    \includegraphics[trim=0.2cm 0.2cm 0 0, clip, width=0.32\linewidth]{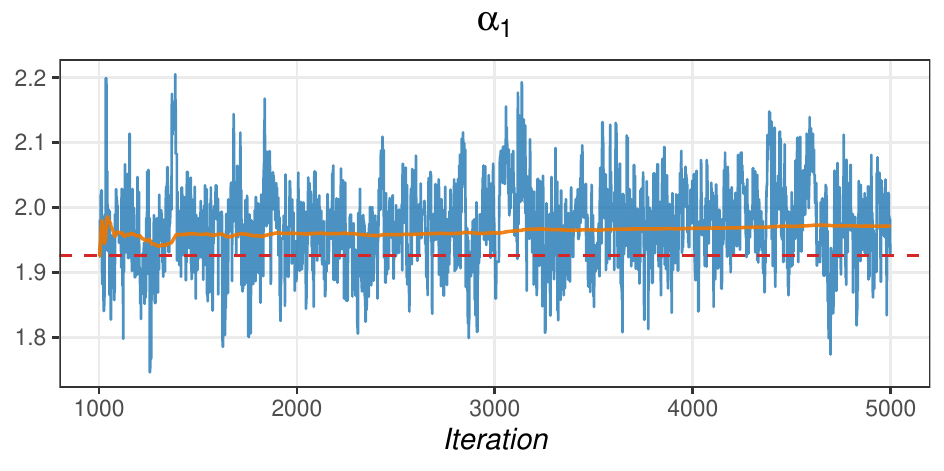} \\
    
     \includegraphics[trim=0.2cm 0.2cm 0 0, clip, width=0.32\linewidth]{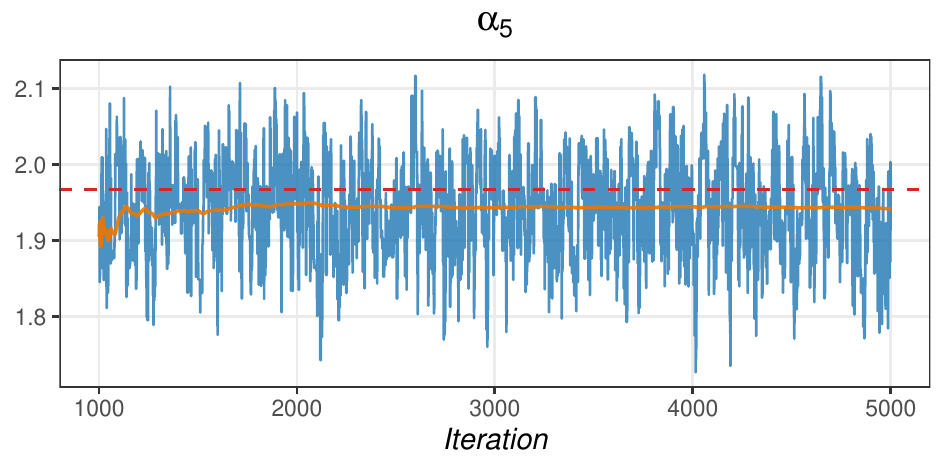} & 
    \includegraphics[trim=0.2cm 0.2cm 0 0, clip, width=0.32\linewidth]{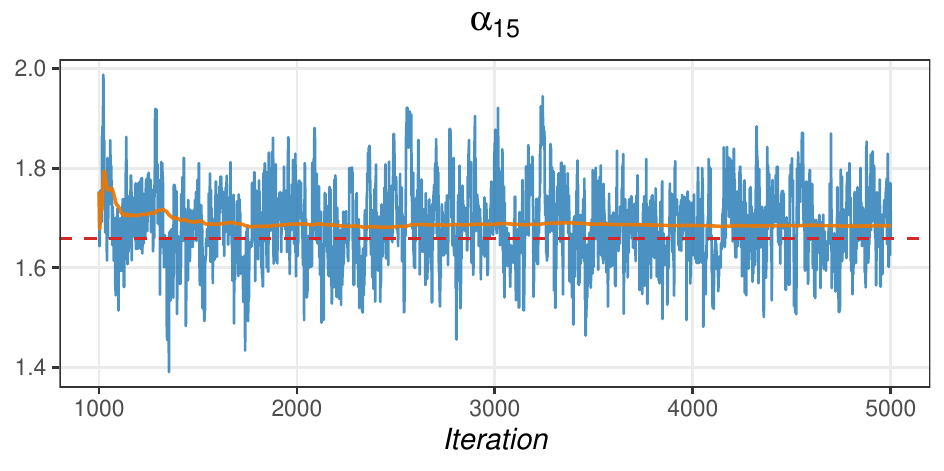} & 
    \includegraphics[trim=0.2cm 0.2cm 0 0, clip, width=0.32\linewidth]{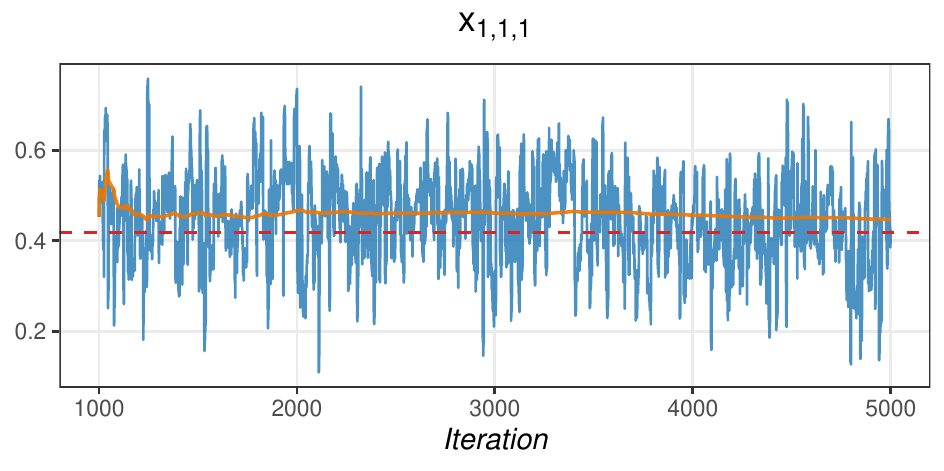} \\
    
    \includegraphics[trim=0.2cm 0.2cm 0 0, clip, width=0.32\linewidth]{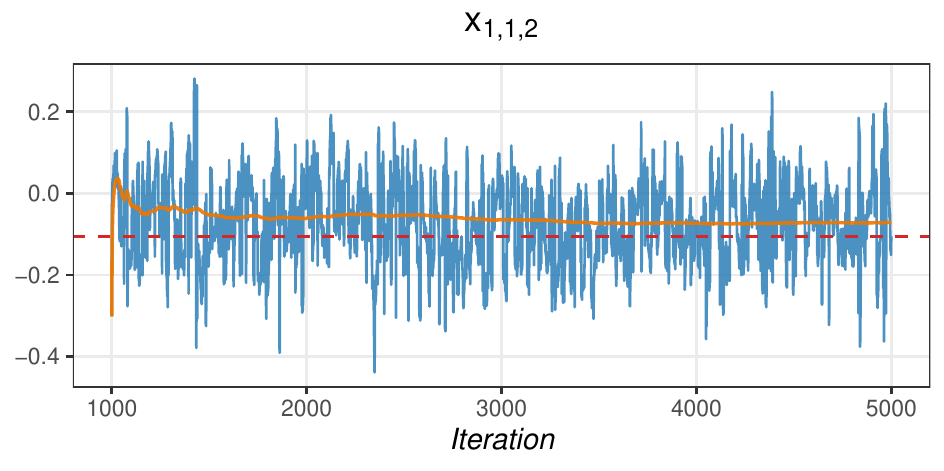} & 
    \includegraphics[trim=0.2cm 0.2cm 0 0, clip, width=0.32\linewidth]{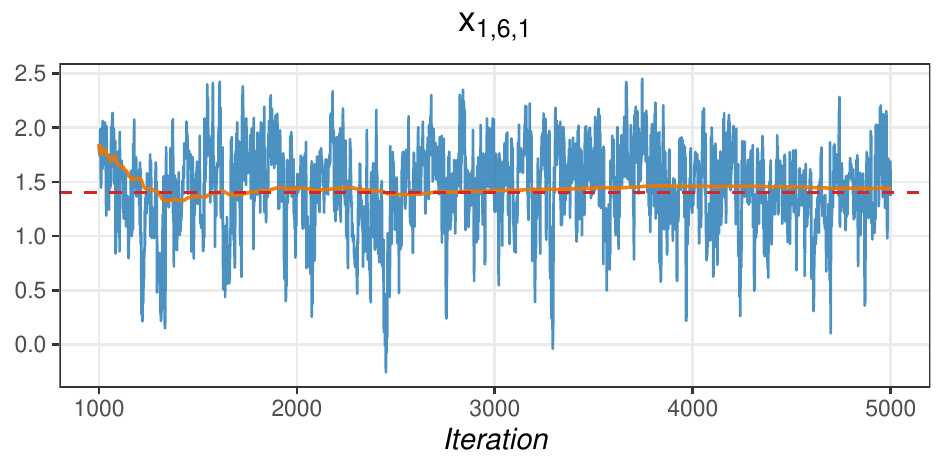} & 
    \includegraphics[trim=0.2cm 0.2cm 0 0, clip, width=0.32\linewidth]{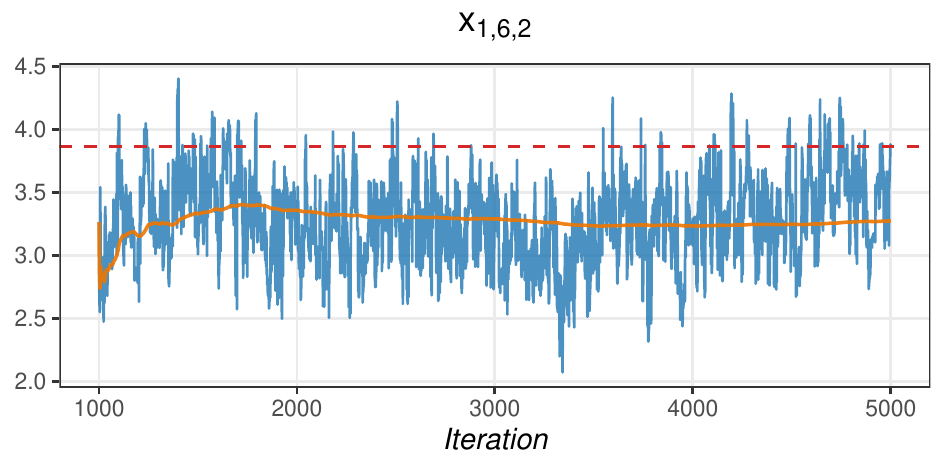}

    \end{tabular}}
    \caption{\textbf{Trace plots for $\mathcal{M}_1$ (factor model),  $\mathcal{M}_2$ (autoregressive model) and  $\mathcal{M}_3$ (latent space):} Posterior samples are reported in blue, progressive averages in solid orange and true values in black dashed. Sampling with 5'000 iterations and 1'000 iterations as burn-in.}  
    \label{fig:mcmc}
\end{figure}

\clearpage

\begin{figure}[h!]
    \centering
\resizebox{0.8\textwidth}{!}{
    \begin{tabular}{ccc}
\multicolumn{3}{c}{\small Model $\mathcal{M}_1$ (Factor Model)}\\
    \includegraphics[trim=0.8cm 0.8cm 0 0, clip, width=0.28\linewidth]{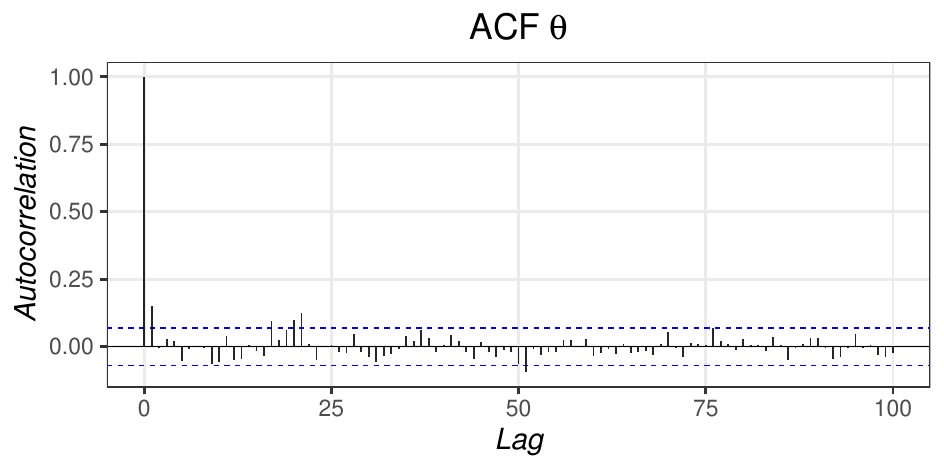} & 
    \includegraphics[trim=0.8cm 0.8cm 0 0, clip, width=0.28\linewidth]{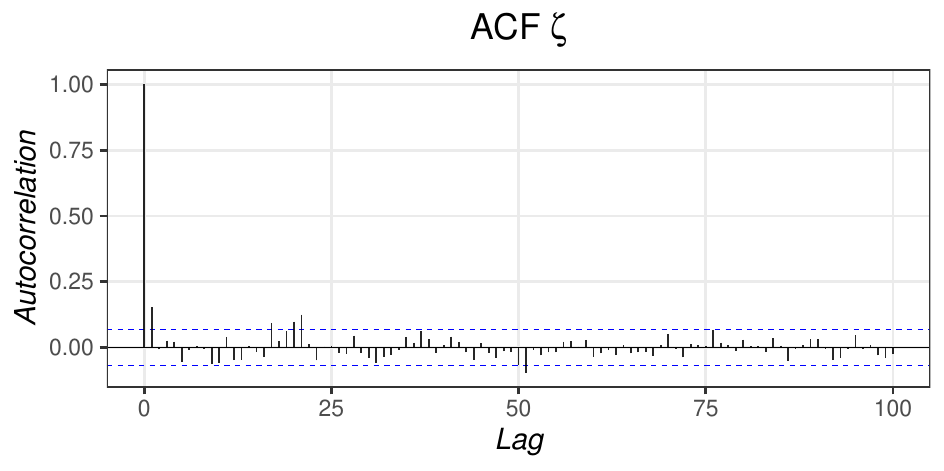} & 
    \includegraphics[trim=0.8cm 0.8cm 0 0, clip, width=0.28\linewidth]{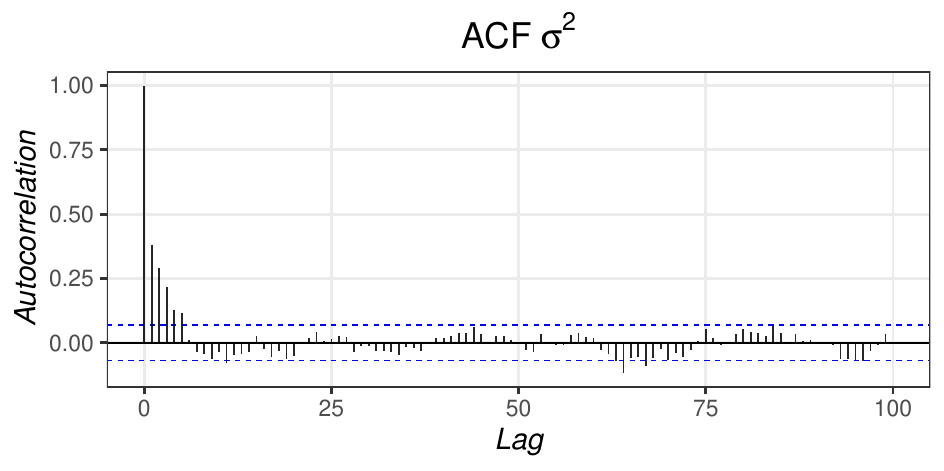} \\
    
     \includegraphics[trim=0.8cm 0.8cm 0 0, clip, width=0.28\linewidth]{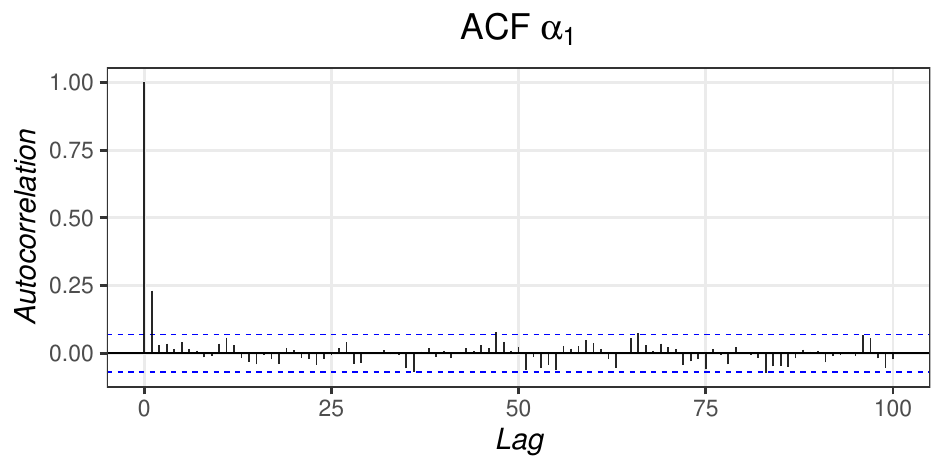} & 
    \includegraphics[trim=0.8cm 0.8cm 0 0, clip, width=0.28\linewidth]{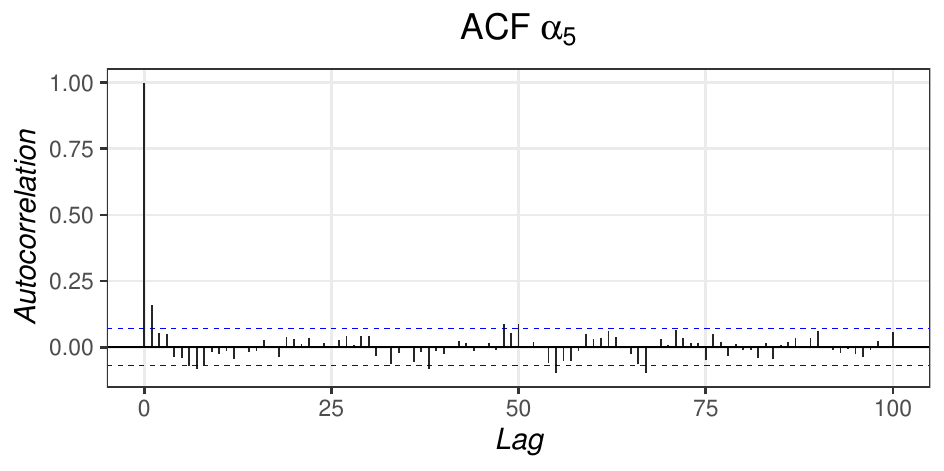} & 
    \includegraphics[trim=0.8cm 0.8cm 0 0, clip, width=0.28\linewidth]{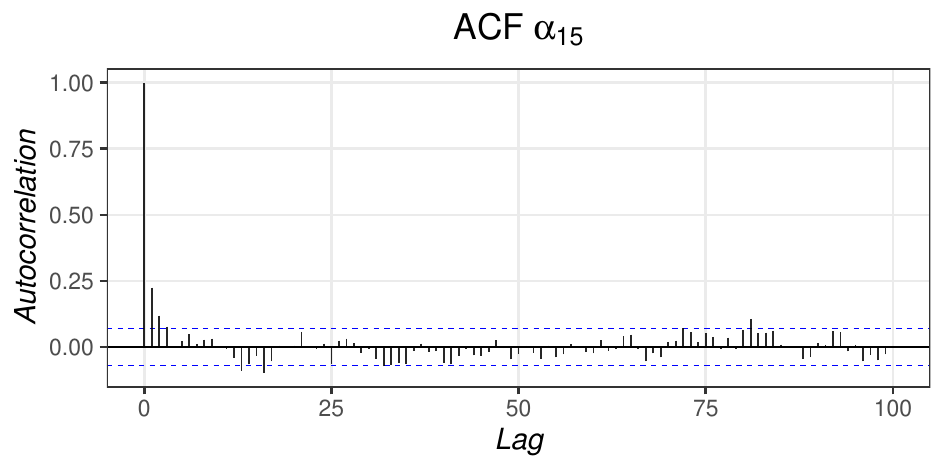} \\
        
    \includegraphics[trim=0.8cm 0.8cm 0 0, clip, width=0.28\linewidth]{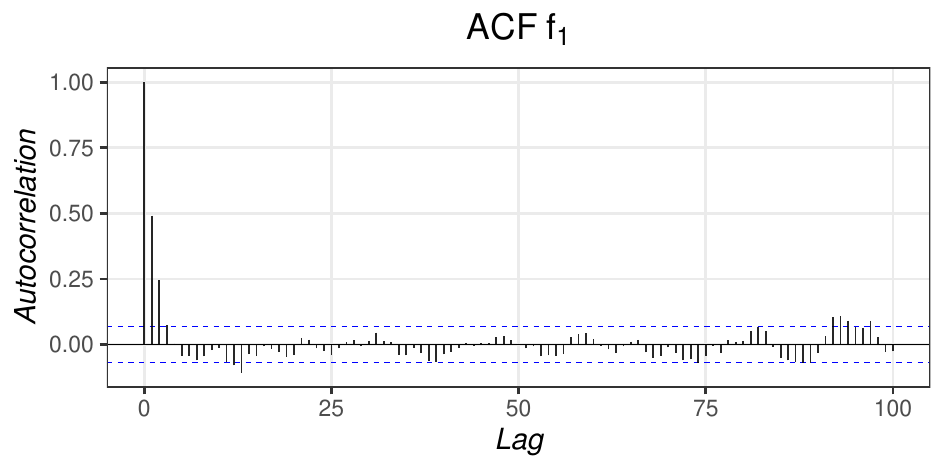} & 
    \includegraphics[trim=0.8cm 0.8cm 0 0, clip, width=0.28\linewidth]{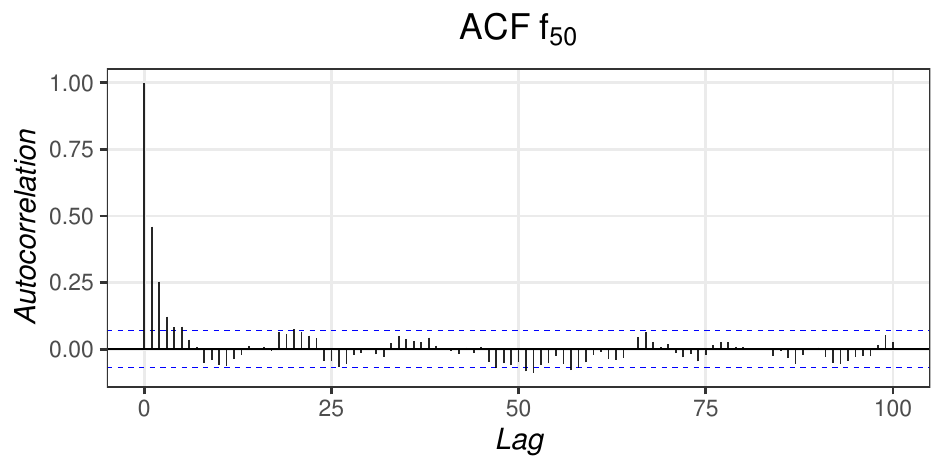} & 
    \includegraphics[trim=0.8cm 0.8cm 0 0, clip, width=0.28\linewidth]{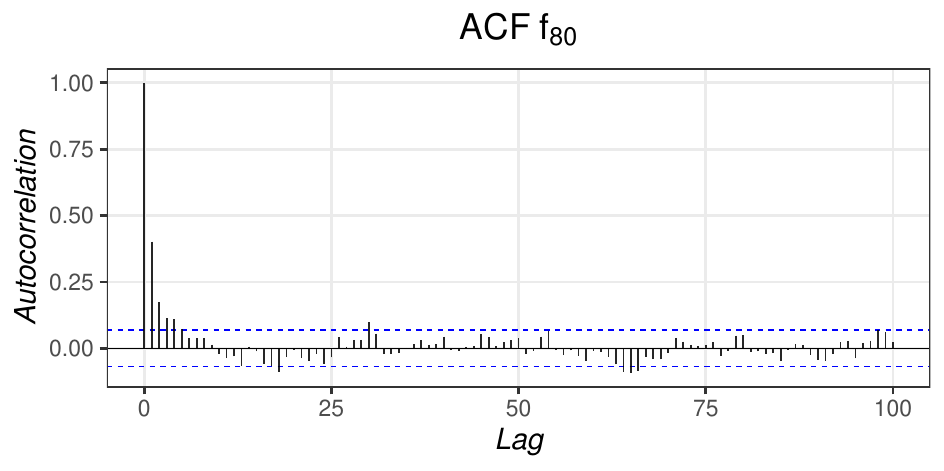}
\\

    \multicolumn{3}{c}{\small Model $\mathcal{M}_2$ (Autoregressive Model)}\\
    \includegraphics[trim=0.8cm 0.8cm 0 0, clip, width=0.28\linewidth]{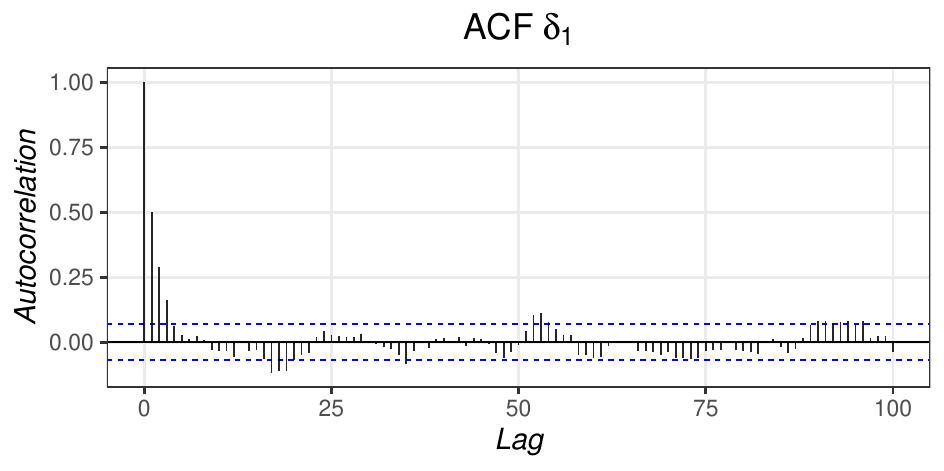} & 
    \includegraphics[trim=0.8cm 0.8cm 0 0, clip, width=0.28\linewidth]{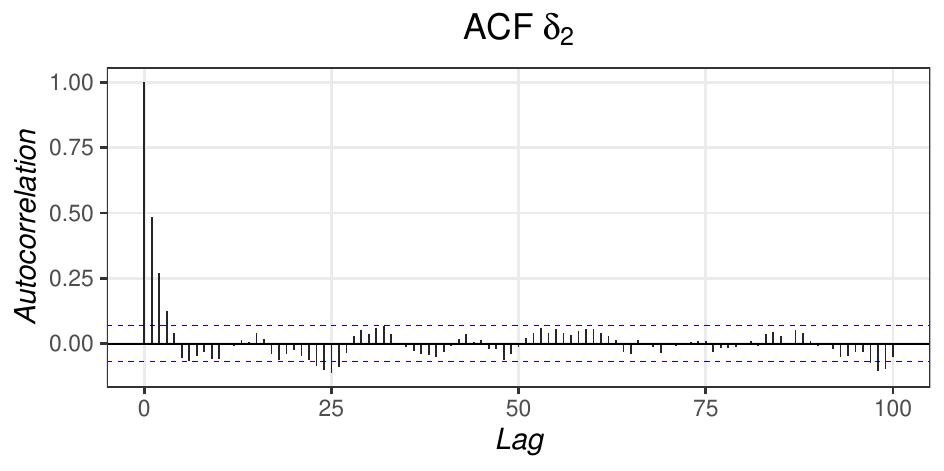} & 
    \includegraphics[trim=0.8cm 0.8cm 0 0, clip, width=0.28\linewidth]{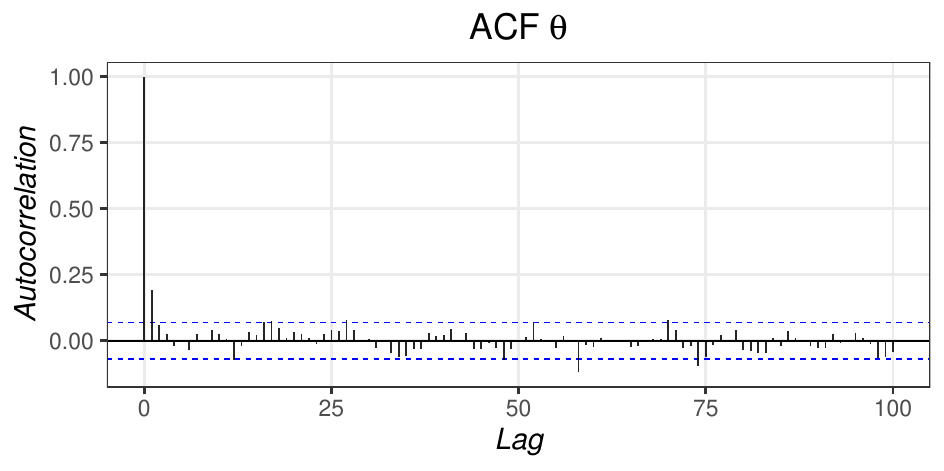} \\
    
     \includegraphics[trim=0.8cm 0.8cm 0 0, clip, width=0.28\linewidth]{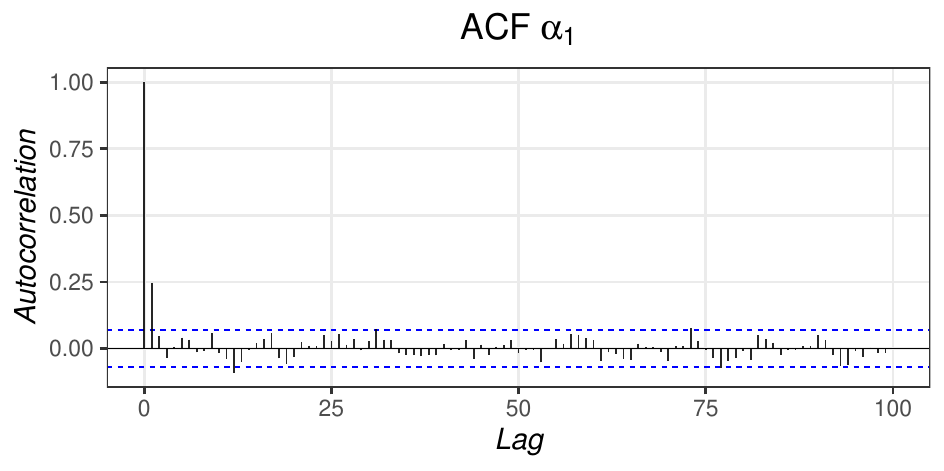} & 
    \includegraphics[trim=0.8cm 0.8cm 0 0, clip, width=0.28\linewidth]{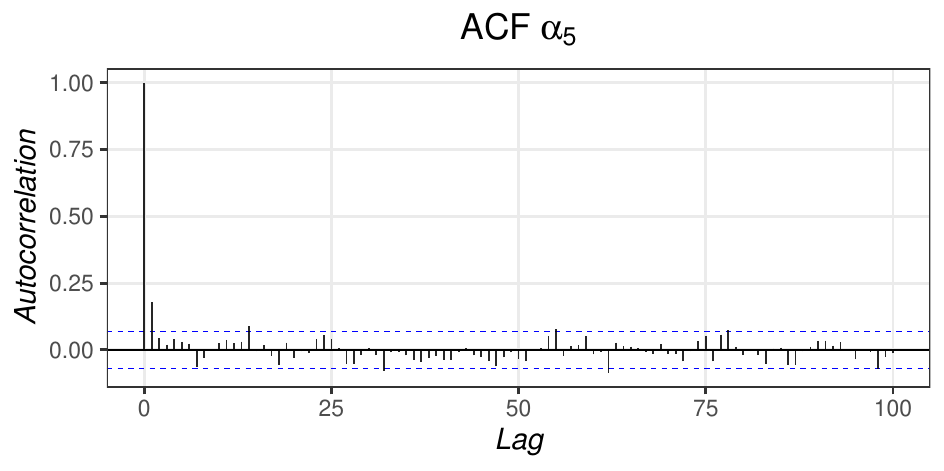} & 
    \includegraphics[trim=0.8cm 0.8cm 0 0, clip, width=0.28\linewidth]{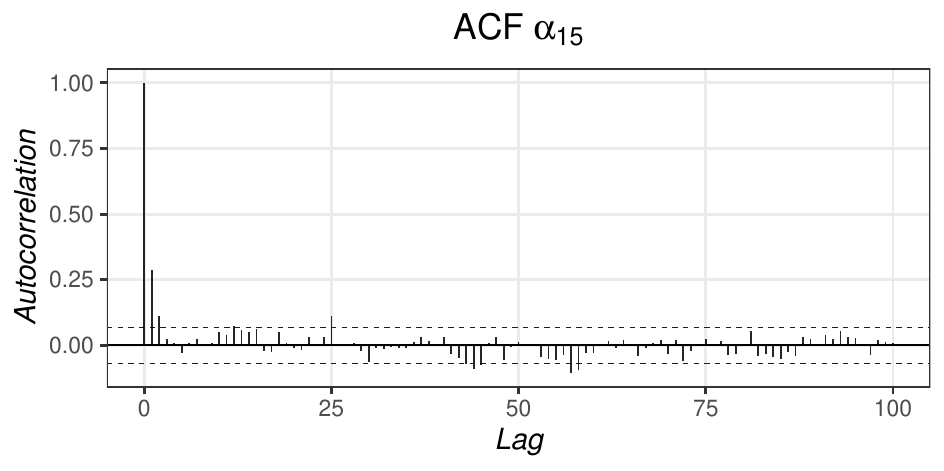} \\

\multicolumn{3}{c}{\small Model $\mathcal{M}_3$ (Latent Space Model)}\\
\includegraphics[trim=0.8cm 0.8cm 0 0, clip, width=0.28\linewidth]{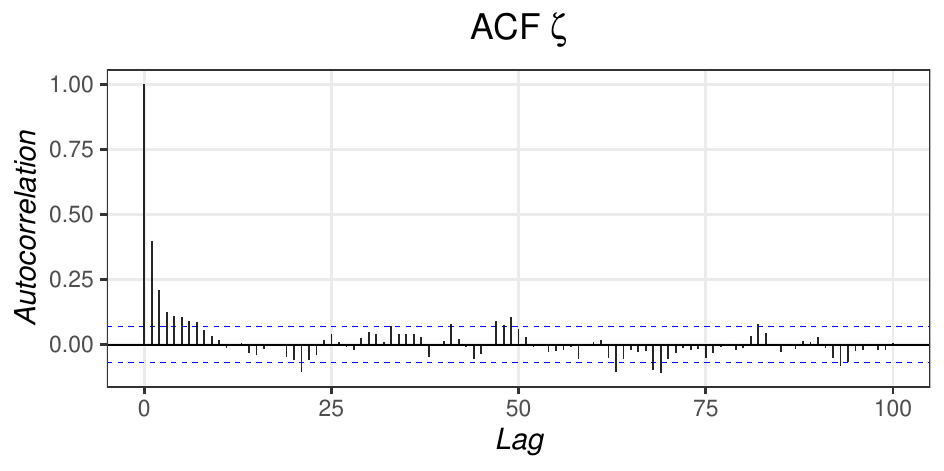} & 
    \includegraphics[trim=0.8cm 0.8cm 0 0, clip, width=0.28\linewidth]{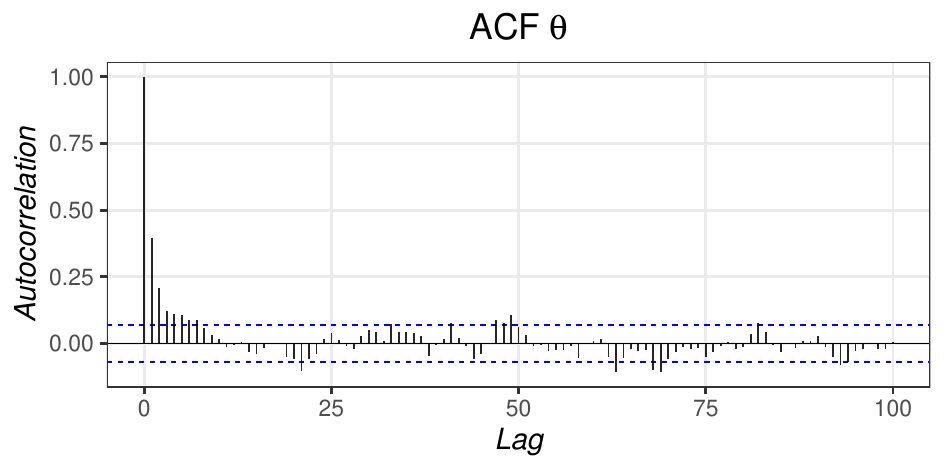} & 
    \includegraphics[trim=0.8cm 0.8cm 0 0, clip, width=0.28\linewidth]{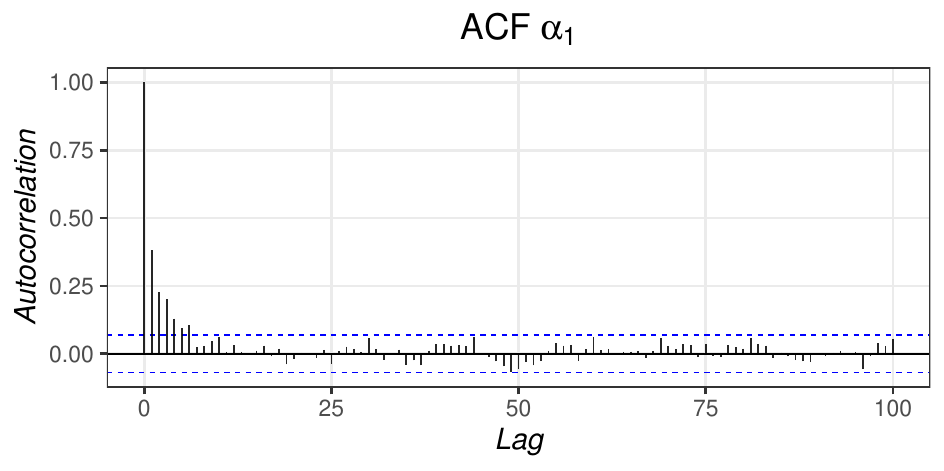} \\
    
     \includegraphics[trim=0.8cm 0.8cm 0 0, clip, width=0.28\linewidth]{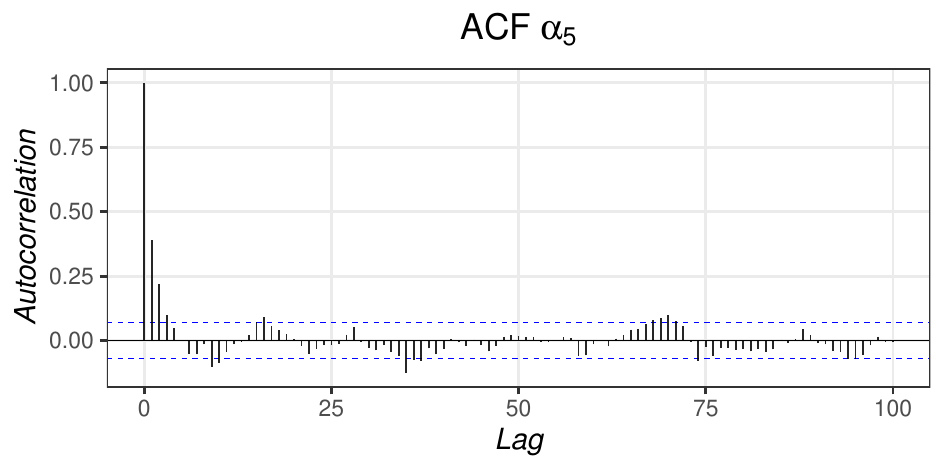} & 
    \includegraphics[trim=0.8cm 0.8cm 0 0, clip, width=0.28\linewidth]{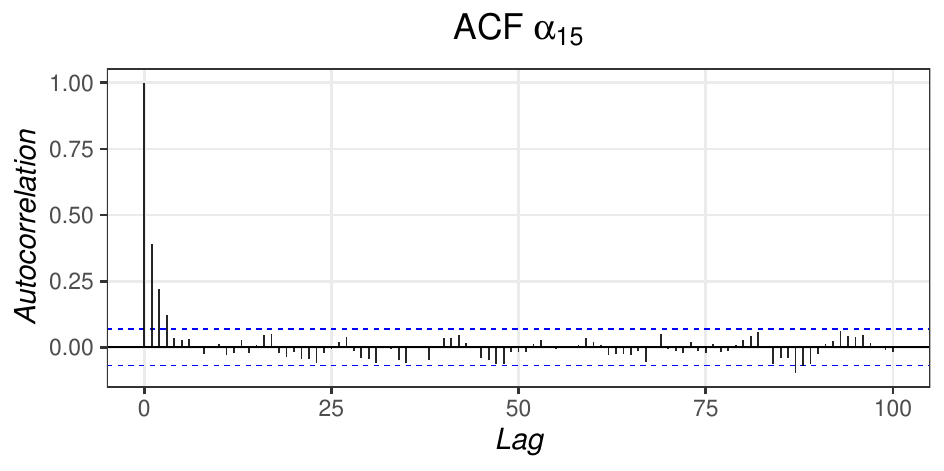} & 
    \includegraphics[trim=0.8cm 0.8cm 0 0, clip, width=0.28\linewidth]{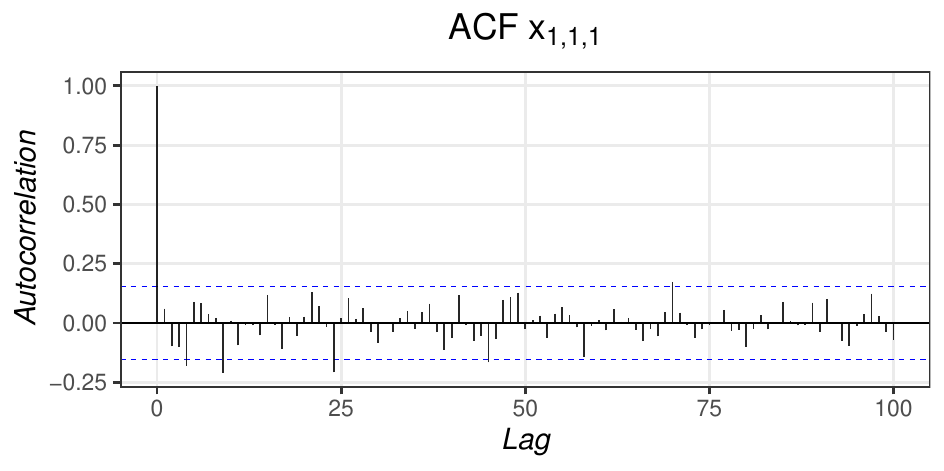} \\
    
    \includegraphics[trim=0.8cm 0.8cm 0 0, clip, width=0.28\linewidth]{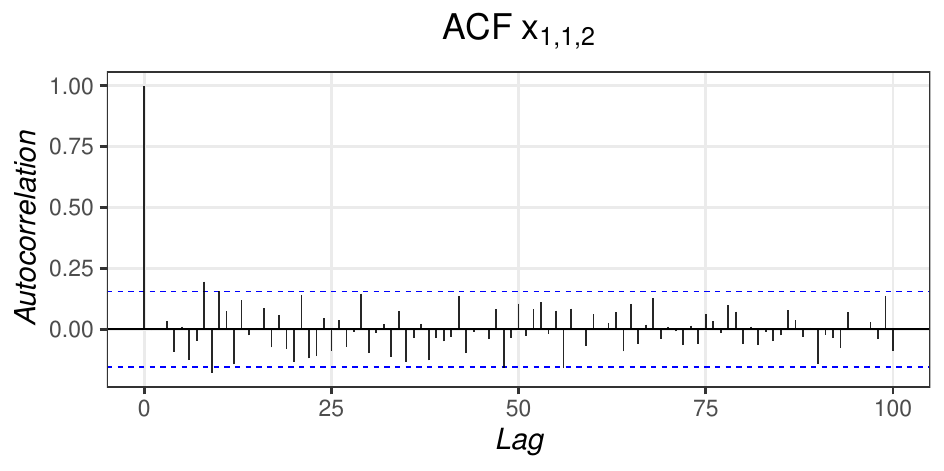} & 
    \includegraphics[trim=0.8cm 0.8cm 0 0, clip, width=0.28\linewidth]{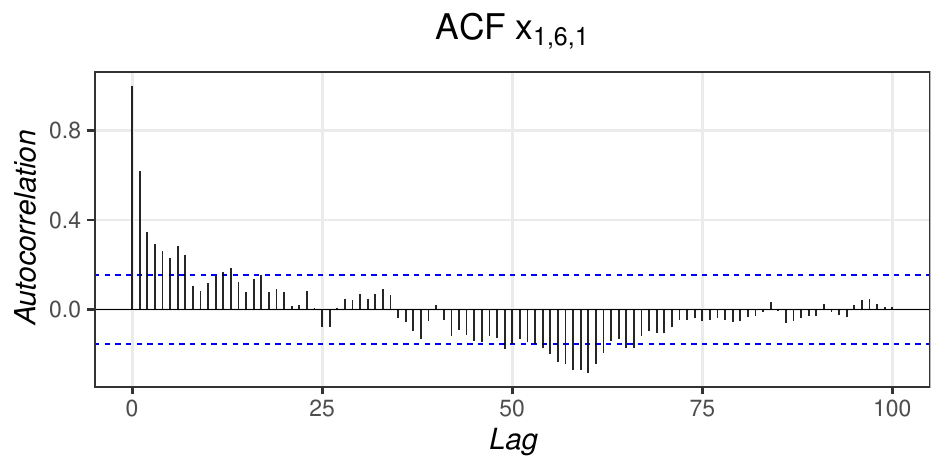} & 
    \includegraphics[trim=0.8cm 0.8cm 0 0, clip, width=0.28\linewidth]{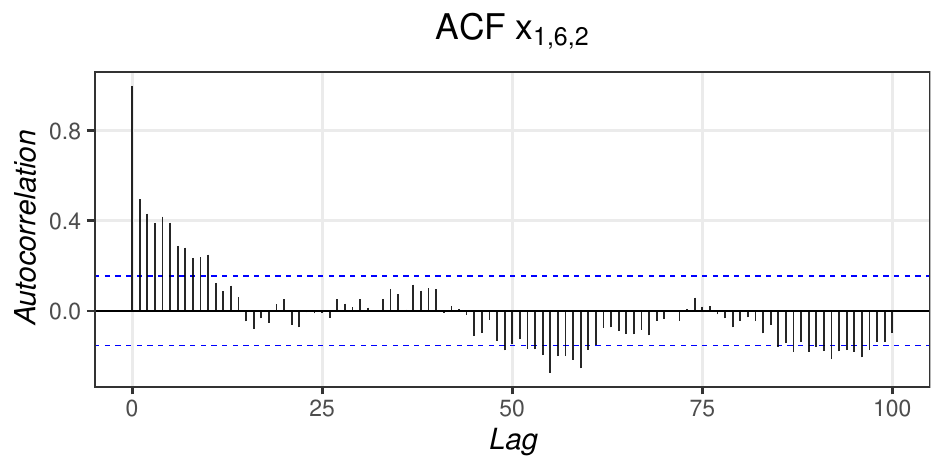}
    \end{tabular}}
    \caption{\textbf{Autocorrelation plots for $\mathcal{M}_1$ (factor model),  $\mathcal{M}_2$ (autoregressive model) and  $\mathcal{M}_3$ (latent space):} Posterior samples are reported in Autocorrelation at different lags are reported in black and White confidence intervals are reported in blue. Sampling with 5'000 iterations and 1'000 iterations as burn-in with thinning every 5 iterations.}  
    \label{fig:acf}
\end{figure}

\begin{figure}[t]
    \centering
    \begin{tabular}{c}
    (a) Model $\mathcal{M}_1$\\
       \includegraphics[height=93pt,width=\linewidth]{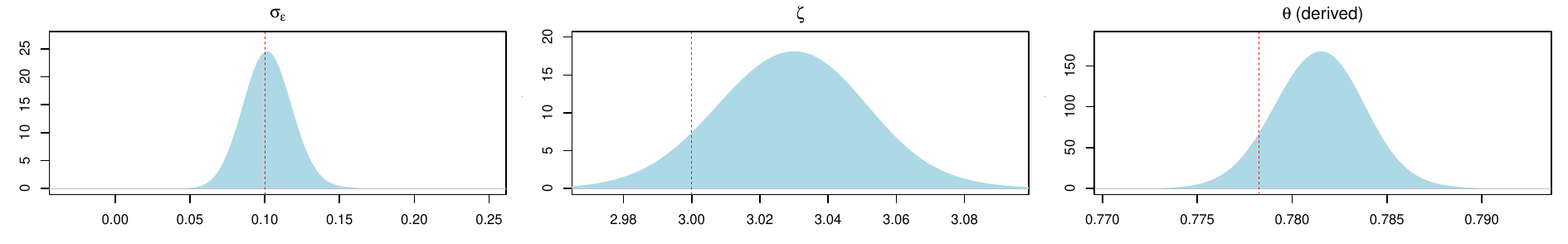}   \\
       (b) Model $\mathcal{M}_2$\\
       \includegraphics[height=93pt,width=\linewidth]{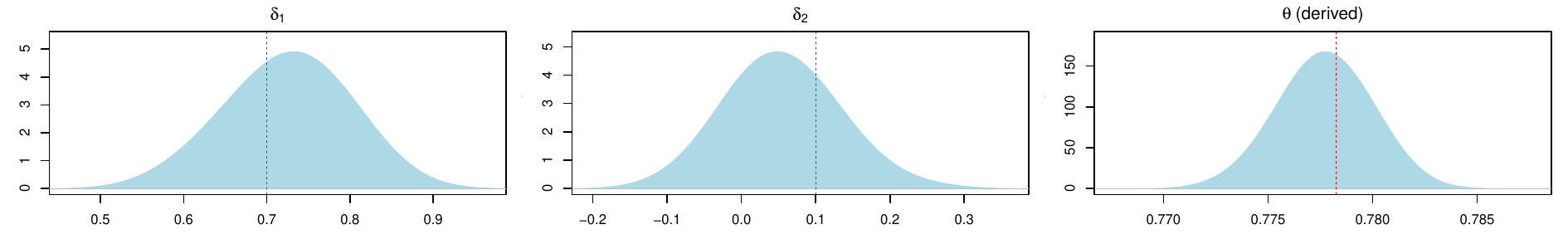}   \\
       (c) Model $\mathcal{M}_3$\\
\includegraphics[height=93pt,width=\linewidth]{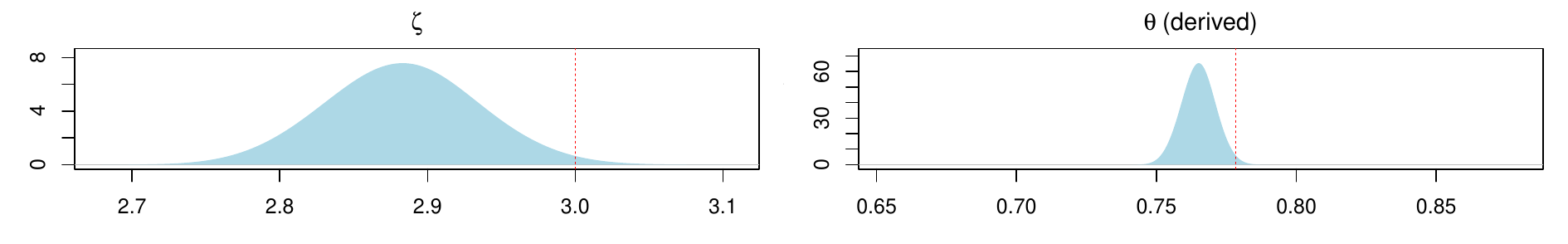}
    \end{tabular}
    \caption{\textbf{Simulation Exercise Posteriors:} Posterior density kernel estimates (\usebox{\bluerectanglebox}) and true values (\usebox{\lineRed}) of $\sigma^2_\epsilon$, $\zeta$, and, $\theta$ in $\mathcal{M}_1$ (panel a), of $\delta_1$, $\delta_2$, and, $\theta$ in $\mathcal{M}_2$ (panel b) and .}
    \label{fig:sim_ex_m2}
\end{figure}

\begin{figure}[t]
    \centering

    \resizebox{0.7\textwidth}{!}{
    \begin{tabular}{ccc}
 Overdispersion $\mathcal{M}_1$ &   Overdispersion $\mathcal{M}_2$ &    Overdispersion $\mathcal{M}_3$\\    
\includegraphics[trim={1cm 1.5cm 1cm 1.5cm}, clip, width=0.30\linewidth]
{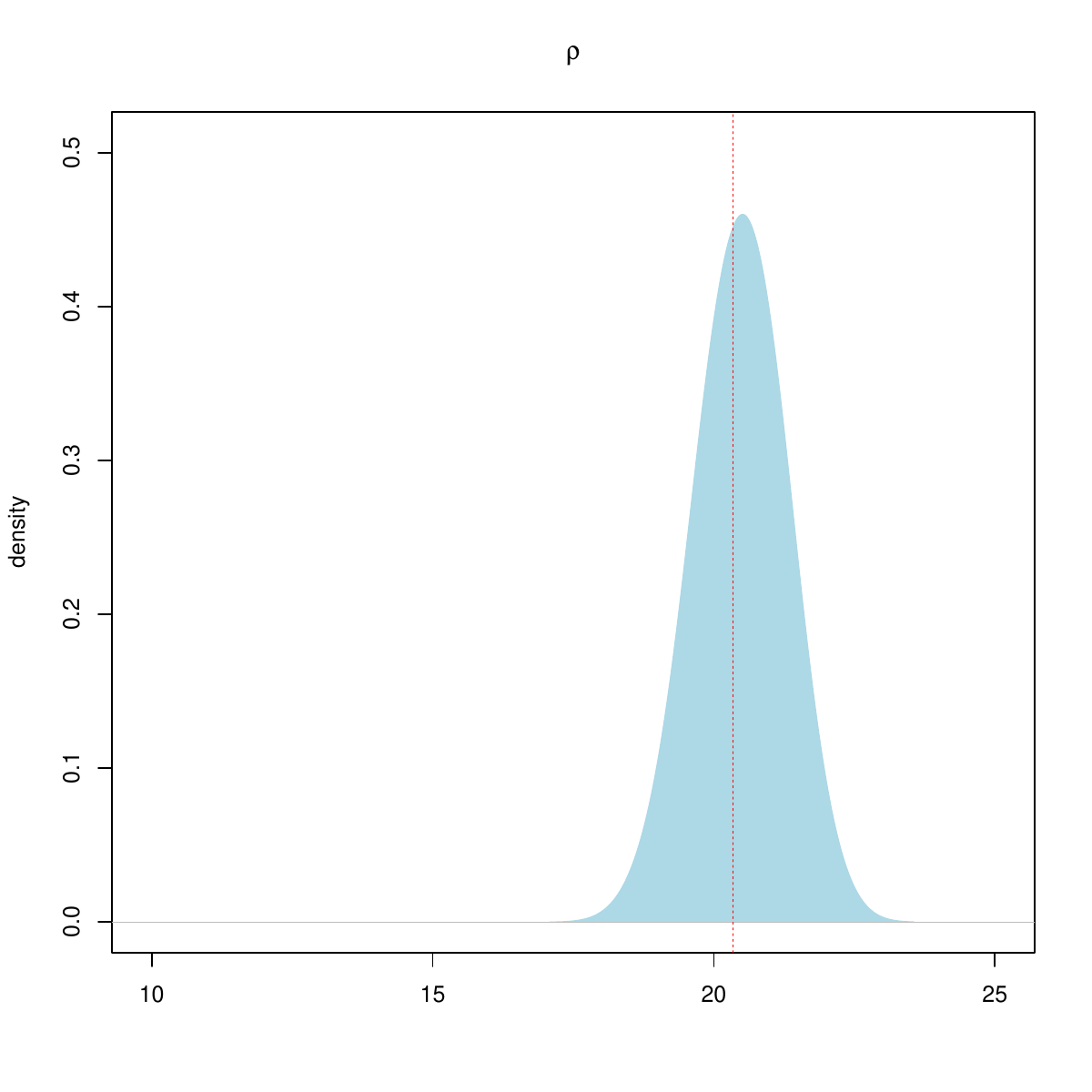} & 
    \includegraphics[trim={1cm 1.5cm 1cm 1.5cm}, clip, width=0.30\linewidth]{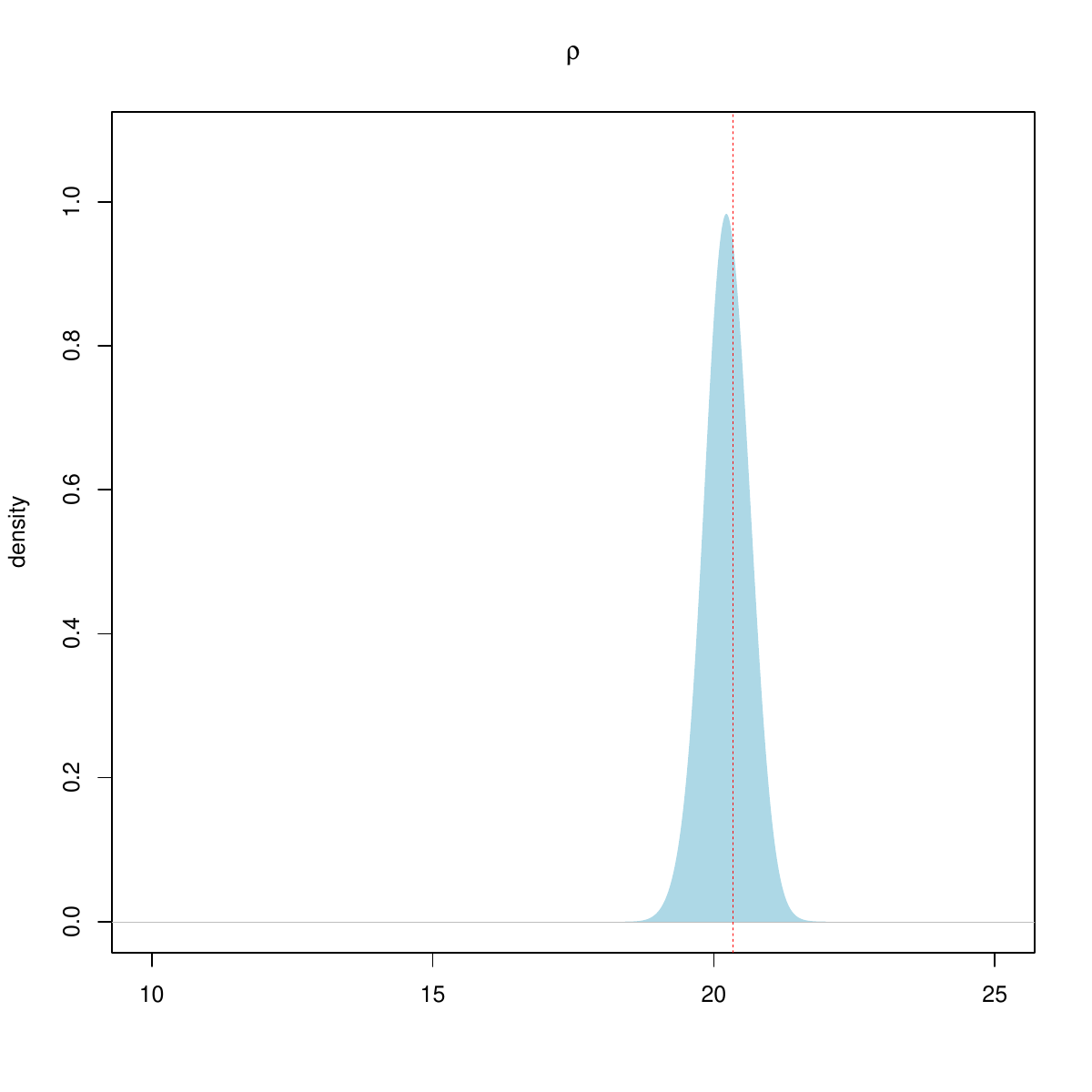}& 
    \includegraphics[trim={1cm 1.5cm 1cm 1.5cm}, clip, width=0.30\linewidth]{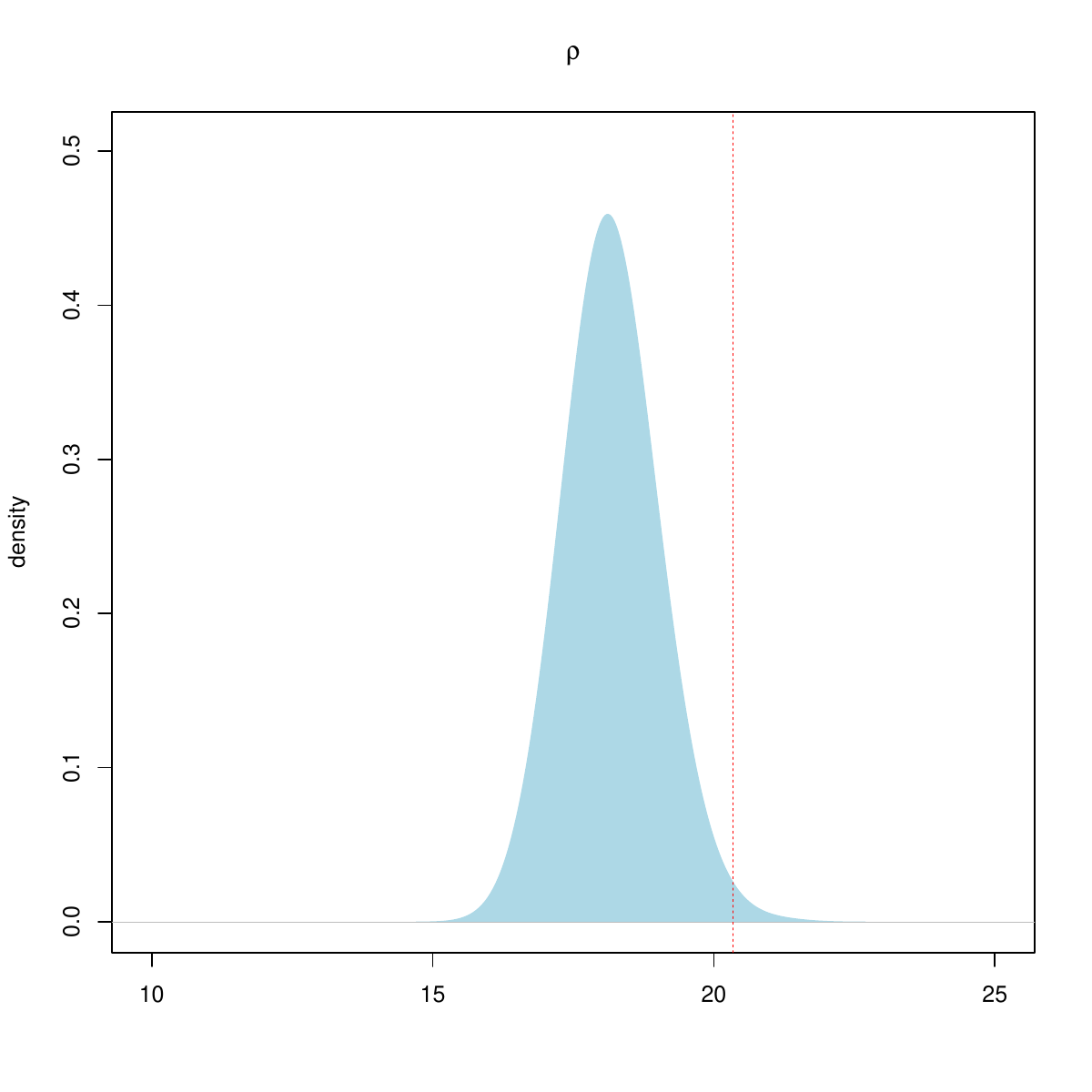}\\
 Underdispersion $\mathcal{M}_1$ &   Underdispersion $\mathcal{M}_2$ &    Underdispersion $\mathcal{M}_3$\\ 
\includegraphics[trim={1cm 1.5cm 1cm 1.5cm}, clip, width=0.30\linewidth]
{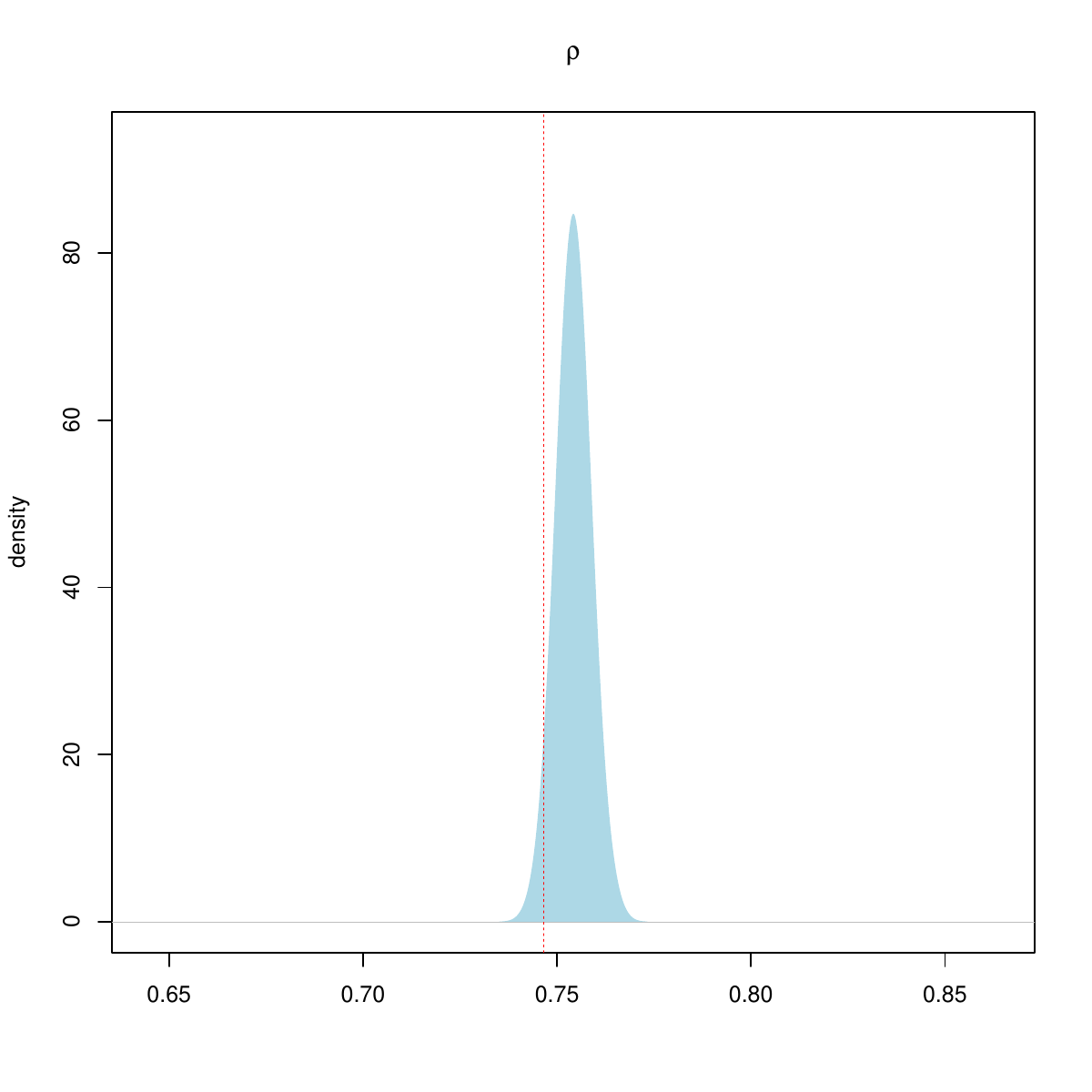} & 
    \includegraphics[trim={1cm 1.5cm 1cm 1.5cm}, clip, width=0.30\linewidth]{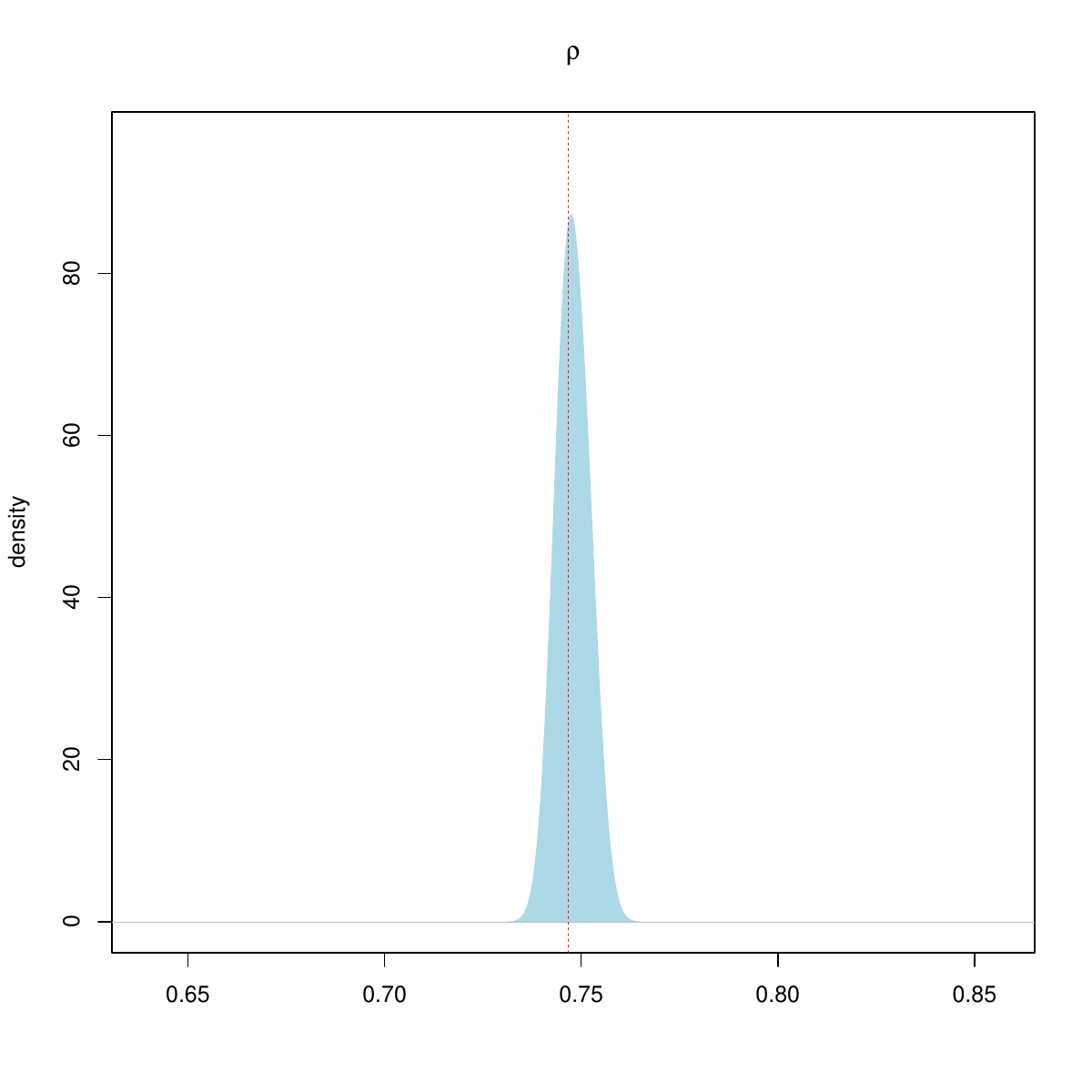}& 
    \includegraphics[trim={1cm 1.5cm 1cm 1.5cm}, clip, width=0.30\linewidth]{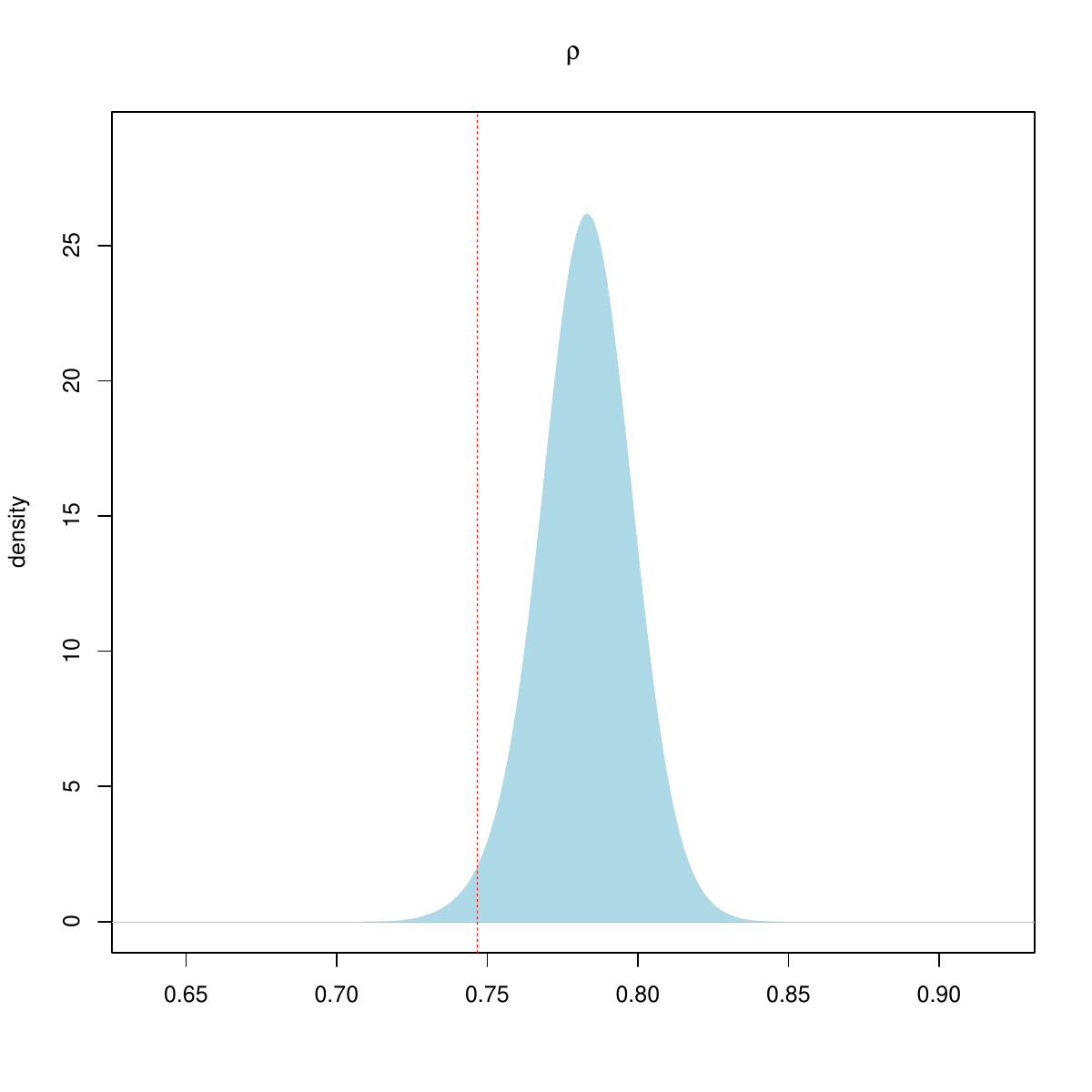}
    \end{tabular}}
    \caption{\textbf{Simulation Exercise Posteriors:} Posterior distribution of the dispersion parameter $\rho$ in case of underdispersion. Posterior densities (blue curve) and true dispersion parameter (dashed red line).}
    \label{fig:underdispersion}
\end{figure}

\begin{figure}[t]
    \centering
    \resizebox{0.9\textwidth}{!}{
    \begin{tabular}{c}
    (a) Model $\mathcal{M}_1$\\
       \includegraphics[width=\linewidth]{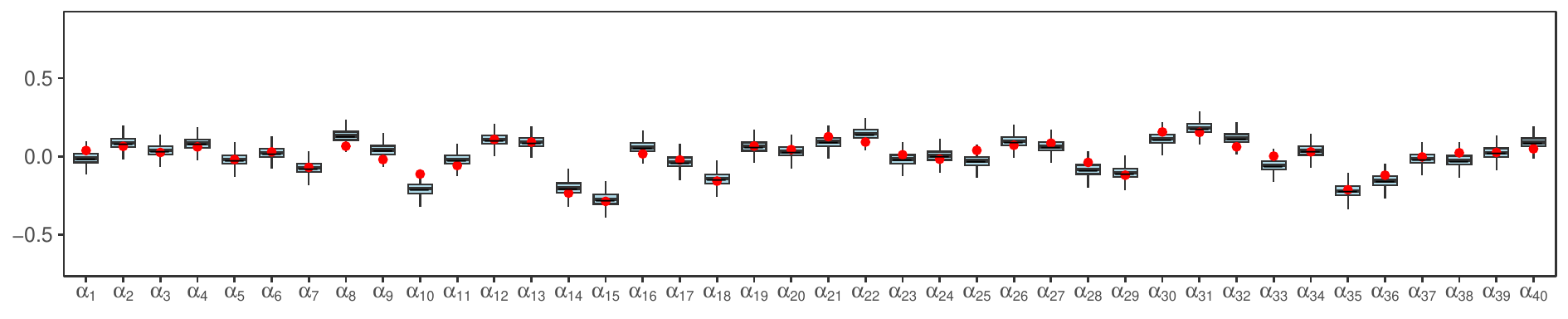}\\
       \includegraphics[width=\linewidth]{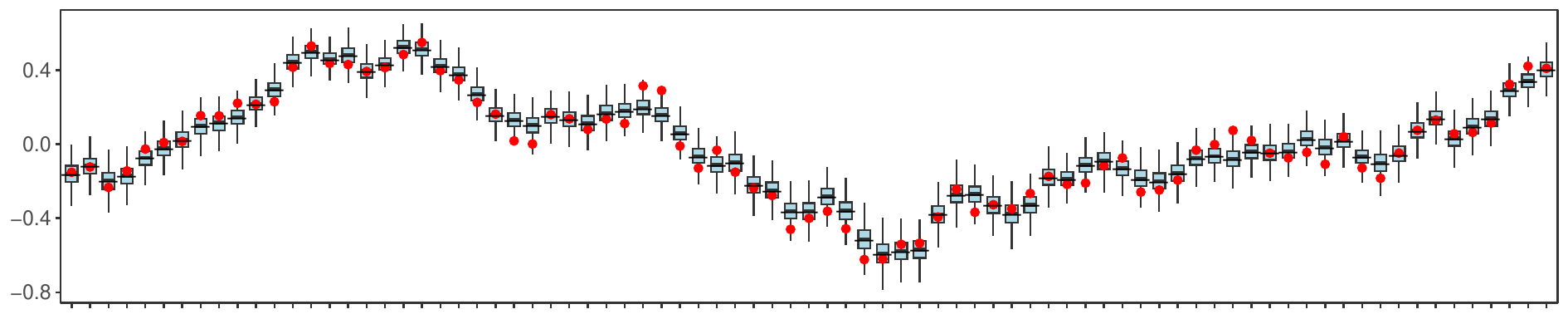}\\
    (b) Model $\mathcal{M}_2$\\
       \includegraphics[width=\linewidth]{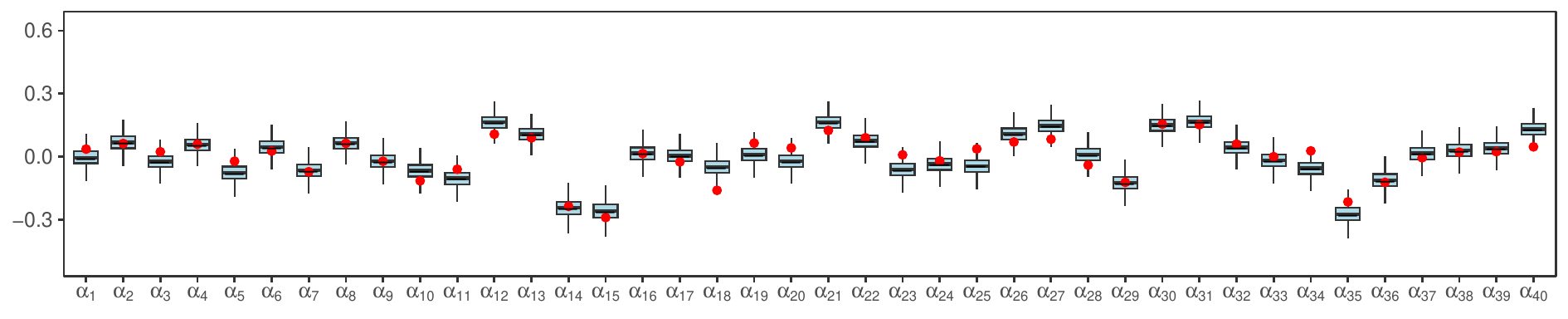}   \\
   (c) Model $\mathcal{M}_3$\\
       \includegraphics[width=\linewidth]{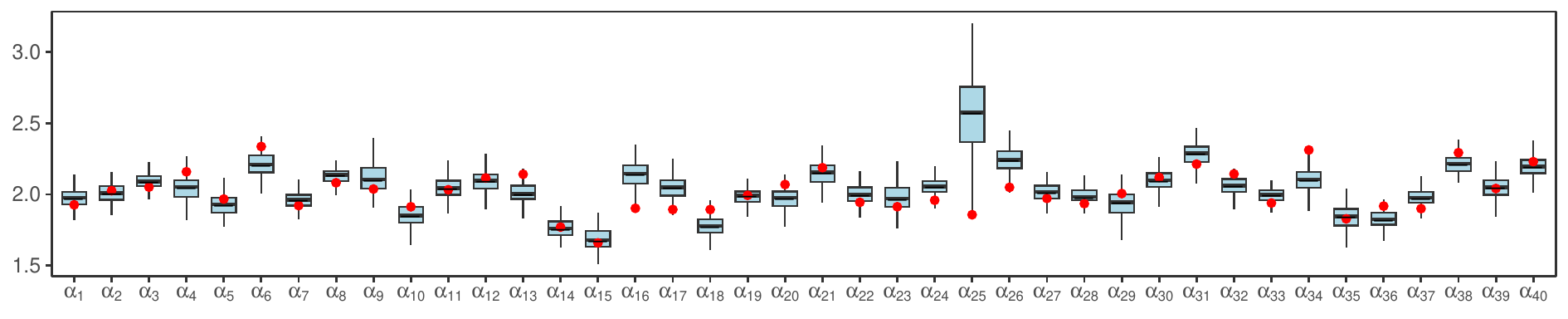}\\
       \includegraphics[width=\linewidth]{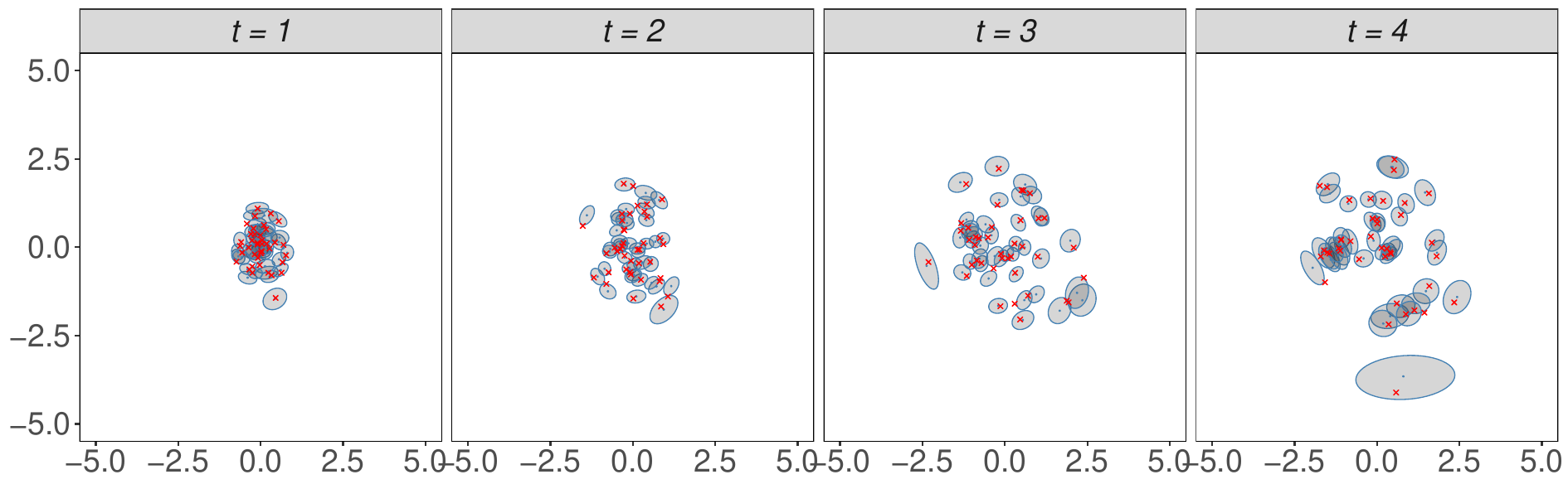}
    \end{tabular}}
    \caption{\textbf{Simulation Exercise Posteriors:} Posterior median (\usebox{\lineBlack}), 95\% HPD (\usebox{\bluerectanglebox}) and true values (\usebox{\DotRed}) of $\alpha_i$, $i = 1, \ldots, N$ and $f_t$ $t = 1, \ldots, T$ in $\mathcal{M}_1$ (panel a), $\alpha_i$ for $i = 1, \ldots, N$ in $\mathcal{M}_2$ (panel b), and  $\alpha_i$ for $i = 1, \ldots, N$, $\mathbf{x}_{it}$  $i = 1, \ldots, N$, $t = 1, \ldots, 4$ and D = 2 in $\mathcal{M}_3$ (panel c).}
    \label{fig:sim_ex_m1}
\end{figure}

\clearpage

\newpage

\renewcommand{\theequation}{C.\arabic{equation}}
\renewcommand{\thefigure}{C.\arabic{figure}}
\renewcommand{\thetable}{C.\arabic{table}}
\setcounter{equation}{0}
\setcounter{figure}{0}
\setcounter{table}{0}

\section{Further empirical results}\label{app:emp}
\vspace{0pt}

\subsection{Citibike Application}

\begin{table}[h]
\resizebox{\textwidth}{!}{
\begin{tabular}{ccccccc}
\hline\hline
1st Semester     & Jan-2019      & Feb-2019      & Mar-2019      & Apr-2019      & May-2019      & Jun-2019      \\ \hline
Mean             & 23601.48      & 23277.08      & 33067.77      & 44562.59      & 49123.64      & 54340.62      \\
Variance         & 745894705.59  & 720626908.91  & 1373513626.58 & 2461623180.31 & 2924056487.6  & 3274853349.51 \\
Dispersion Index & 31603.73      & 30958.64      & 41536.32      & 55239.68      & 59524.43      & 60265.29      \\ \hline
2nd Semester     & Jul-2019      & Aug-2019      & Sep-2019      & Oct-2019      & Nov-2019      & Dec-2019      \\ \hline
Mean             & 55156.85      & 59518.43      & 62554.79      & 53388.33      & 37297.25      & 23489.15      \\
Variance         & 3569834055.46 & 4245141728.18 & 4781624302.37 & 3710476822.66 & 1852269086.76 & 734066259.76  \\
Dispersion Index & 64721.5       & 71324.83      & 76438.98      & 69499.78      & 49662.36      & 31251.29      \\ \hline\hline
\end{tabular}}
\caption{Mean, Variance, and Dispersion Index of the weighted degree distribution at monthly and NTA aggregation in the Citi Bike dataset \cite{citibike2019}. The weighted degree distribution is the distribution of $Y_{it} =\sum_{j=1,j \neq i}^NY_{ijt}$.}
\label{tab:degree}
\end{table}

\begin{figure}[h!]
    \centering
\resizebox{1.0\textwidth}{!}{
    \begin{tabular}{c}
    \includegraphics[width= \linewidth]{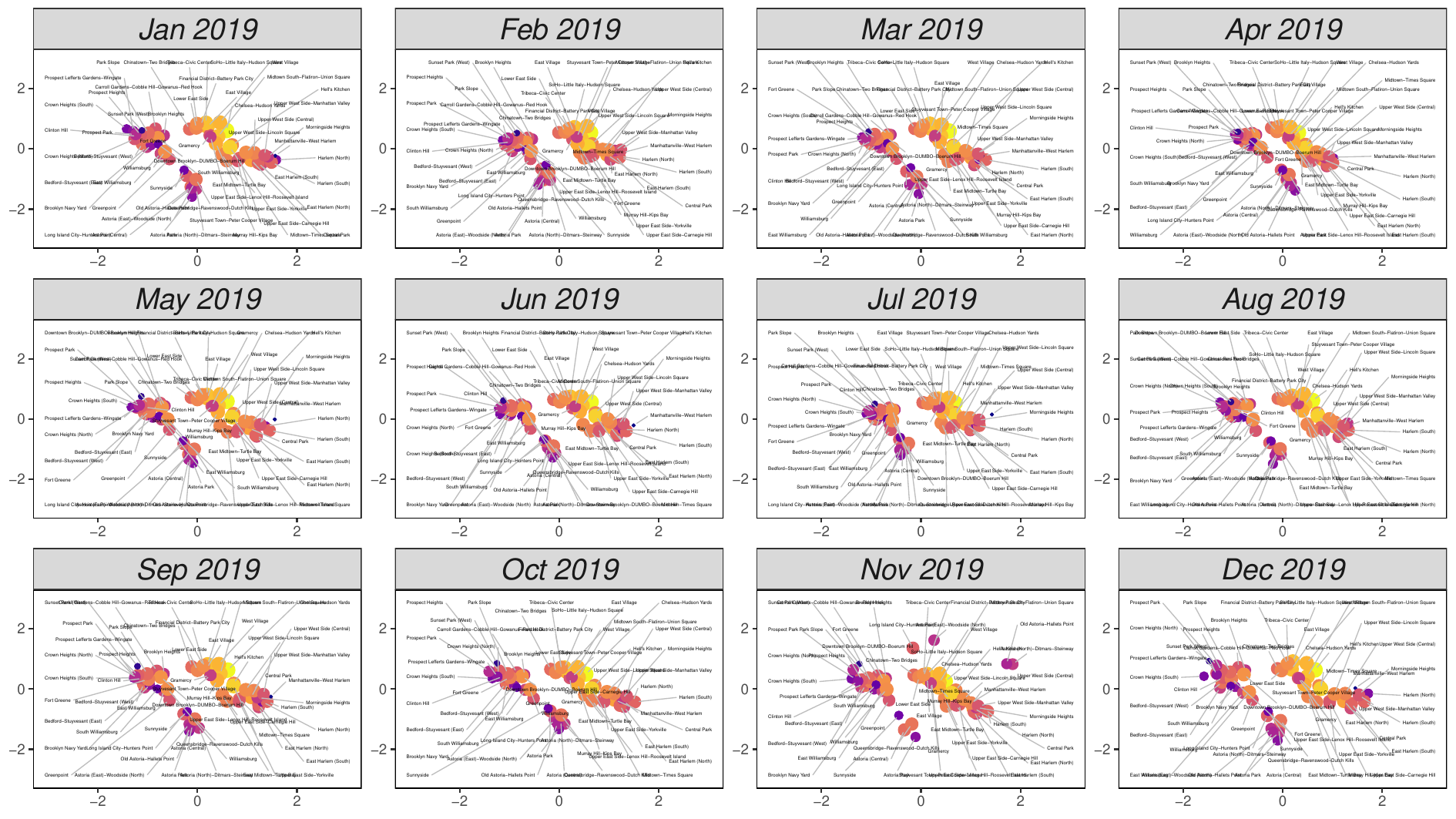} 
    \end{tabular}}
    \caption{\textbf{Citi Bike Network, Model $\mathcal{M}_3$ Results Results:} 
    Latent space representation of the neighborhoods with $d = 2$ over the year 2019. The size of the nodes is proportional to $\alpha_i$. The bigger the dot size of each NTA, the more prominent the neighborhood in the Citi Bike network (\usebox{\purplesmalldot} - \usebox{\yellowbigdot}).}
    \label{fig:application2}
\end{figure}

\begin{figure}[h!]
    \centering
    \begin{tabular}{cc}
    \includegraphics[width=0.25\linewidth]{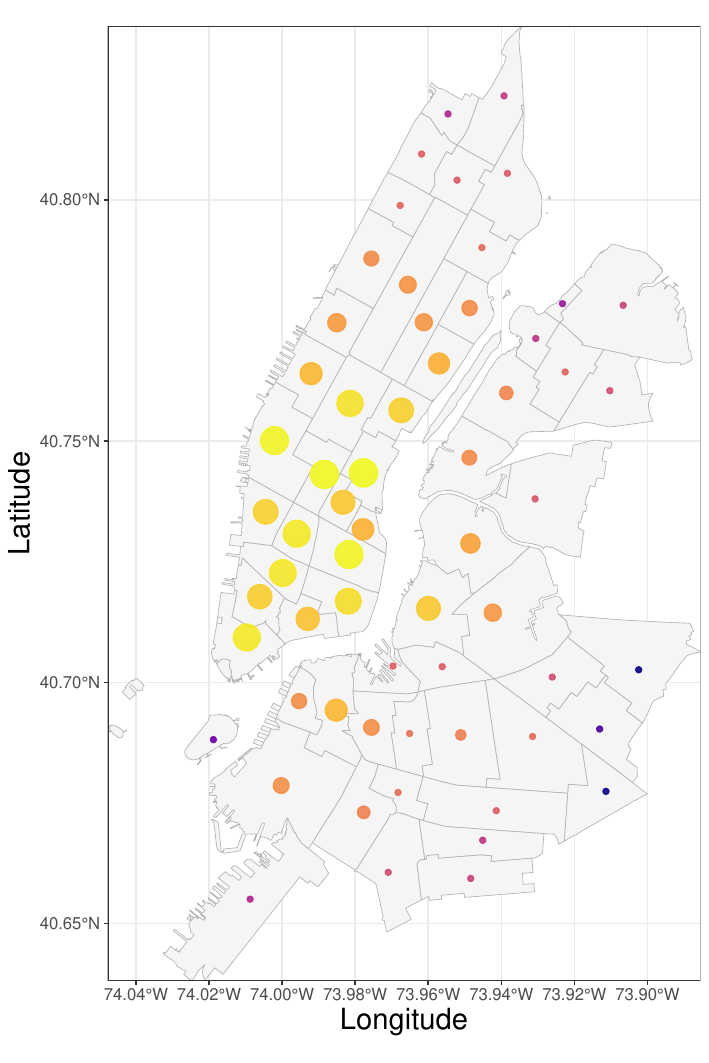}  &  \includegraphics[width=0.5\linewidth]{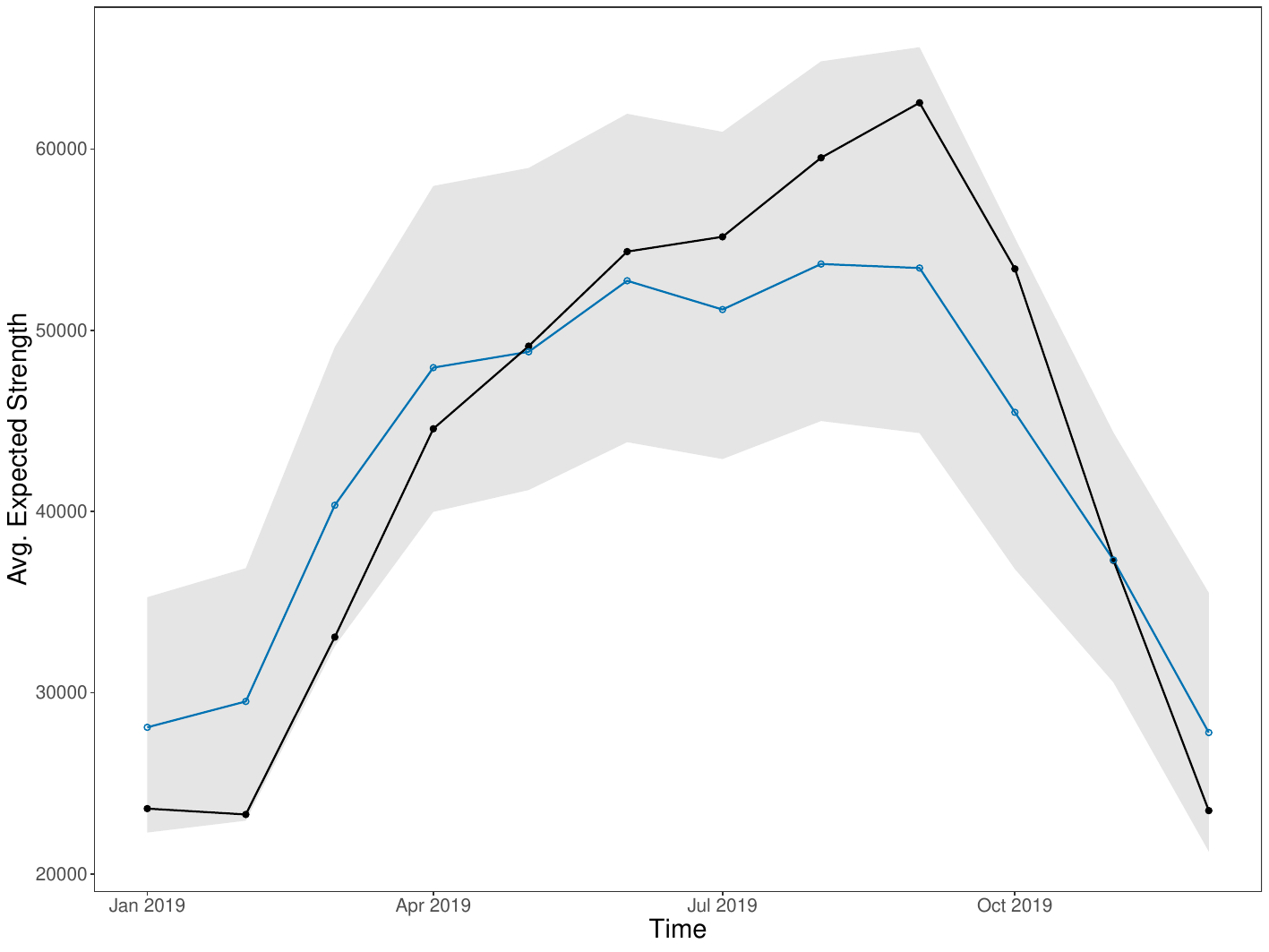}
    \\
    \includegraphics[width=0.25\linewidth]{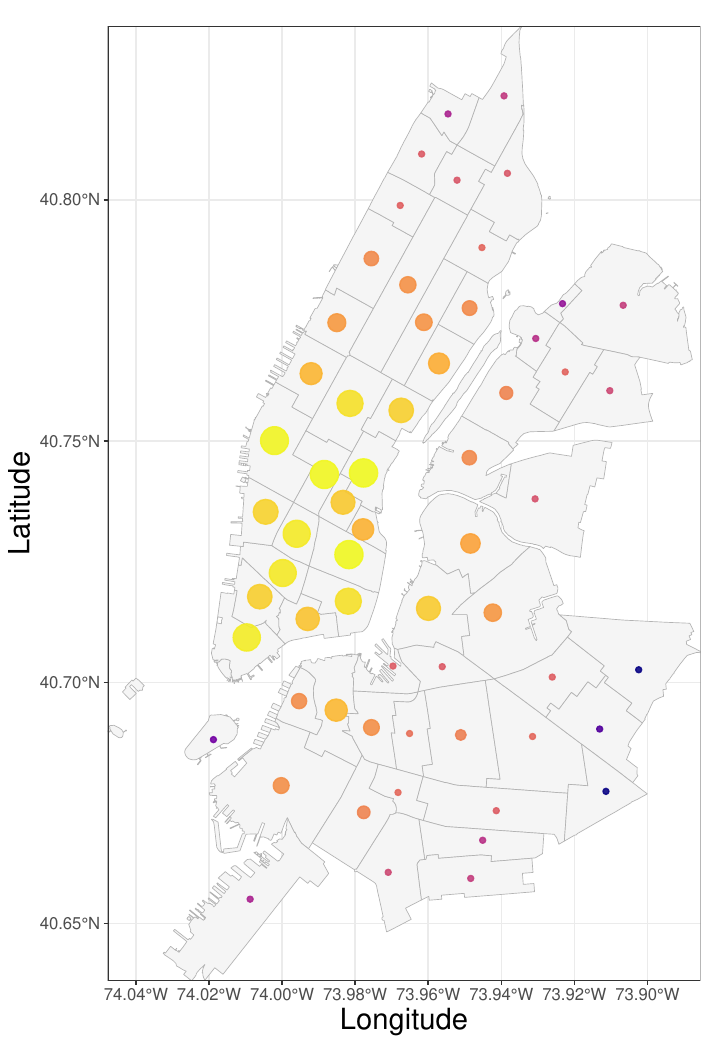}  &  \includegraphics[width=0.5\linewidth]{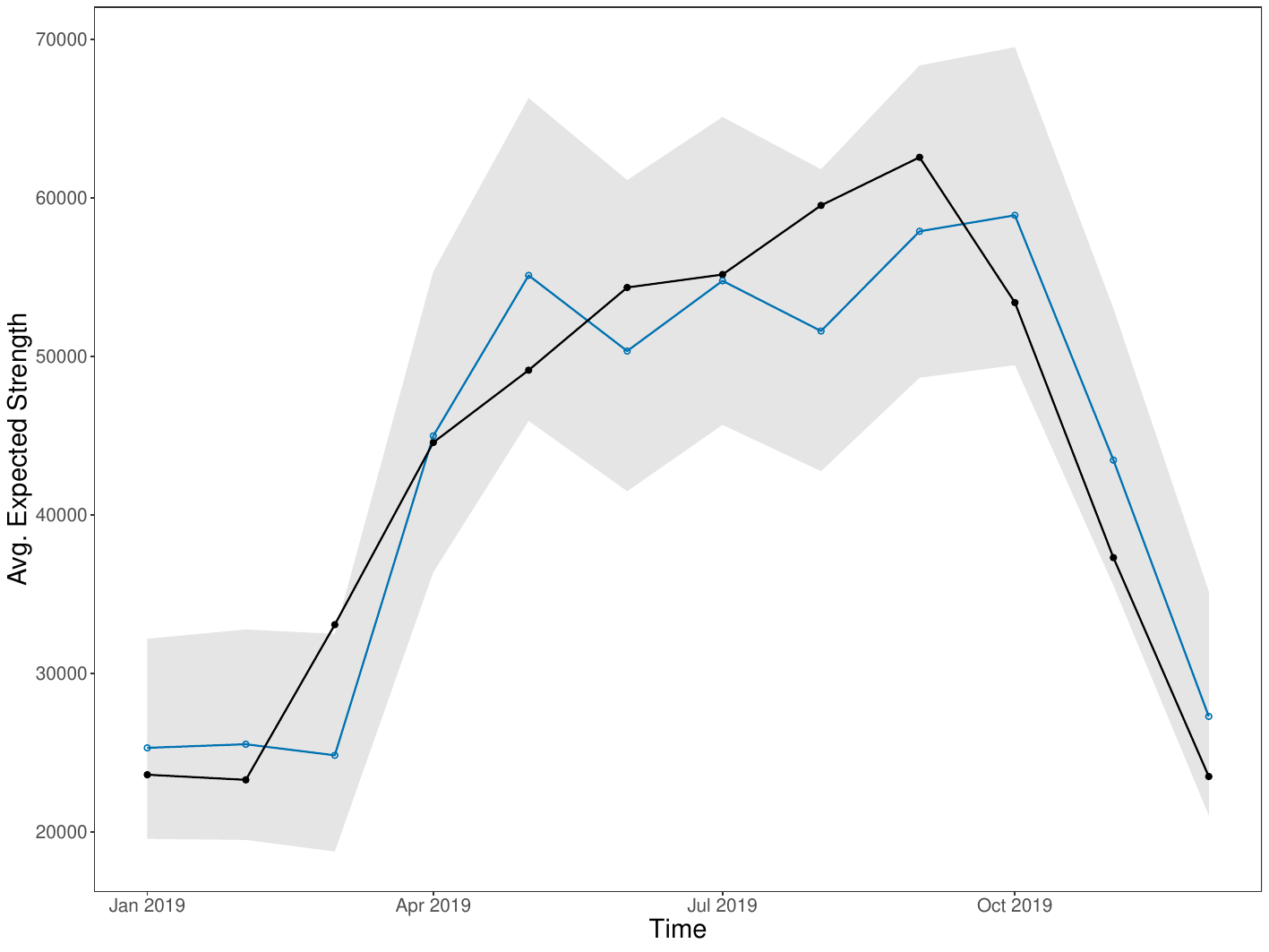} \\
    \includegraphics[width=0.25\linewidth]{Figures/LS/CentralityMap_m3.pdf}  &  \includegraphics[width=0.5\linewidth]{Figures/LS/Application_GP_PPC_M3.pdf} 
    \end{tabular}
    \caption{\textbf{Citi Bike Network, Model $\mathcal{M}_1$, $\mathcal{M}_2$ and $\mathcal{M}_3$ Results:} Left panel, posterior mean of the $\alpha_i$ parameter associated to each of the 61 neighborhoods. The $\alpha_i$ can be interpreted as node centrality. The larger the NTA dot size, the more prominent the neighborhood in the Citi Bike network (\usebox{\purplesmalldot} - \usebox{\yellowbigdot}). Right panel, posterior predictive of the average expected strength over time. Time series of the observed average strength (\usebox{\linewithcirclesbox}) against the posterior mean (\usebox{\linewithcirclesboxBlue}) with its 95\% credible interval (\usebox{\grayrectanglebox}) }
    \label{fig:application}
\end{figure}

\clearpage

\subsection{Media network Application}
\label{sec:media_sup}
\begin{table}[h]
\resizebox{\textwidth}{!}{

\begin{tabular}{ccccccc}
\hline\hline
1st Semester 2015 & Jan-2015 & Feb-2015 & Mar-2015 & Apr-2015 & May-2015 & Jun-2015 \\ \hline
Mean & 3146.72 & 2537.66 & 3280.47 & 2528.30 & 2905.70 & 3522.26 \\
Variance & 20614579.29 & 13016417.71 & 20825976.91 & 11779641.30 & 14702164.04 & 21019835.41 \\
Dispersion Index & 6551.13 & 5129.30 & 6348.48 & 4659.12 & 5059.76 & 5967.72 \\ \hline

2nd Semester 2015 & Jul-2015 & Aug-2015 & Sep-2015 & Oct-2015 & Nov-2015 & Dec-2015 \\ \hline
Mean & 4519.91 & 3363.70 & 4606.77 & 4465.32 & 4847.28 & 3715.62 \\
Variance & 34219608.64 & 19296964.17 & 38811178.01 & 35407776.83 & 39486641.77 & 23338643.15 \\
Dispersion Index & 7570.85 & 5736.82 & 8424.82 & 7929.51 & 8146.15 & 6281.23 \\ \hline

1st Semester 2016 & Jan-2016 & Feb-2016 & Mar-2016 & Apr-2016 & May-2016 & Jun-2016 \\ \hline
Mean & 6448.09 & 5020.89 & 5120.34 & 4714.30 & 4375.62 & 5201.91 \\
Variance & 69618319.34 & 42842136.23 & 47438586.88 & 37897043.78 & 33897981.81 & 48626877.12 \\
Dispersion Index & 10796.74 & 8532.77 & 9264.73 & 8038.75 & 7747.02 & 9347.88 \\ \hline

2nd Semester 2016 & Jul-2016 & Aug-2016 & Sep-2016 & Oct-2016 & Nov-2016 & Dec-2016 \\ \hline
Mean & 6190.72 & 5923.36 & 5403.57 & 5158.00 & 6661.83 & 6885.62 \\
Variance & 69600689.42 & 61848388.71 & 49731199.81 & 48240089.35 & 75803273.41 & 84116133.68 \\
Dispersion Index & 11242.74 & 10441.43 & 9203.39 & 9352.48 & 11378.75 & 12216.21 \\ \hline\hline
\end{tabular}

}
\caption{Mean, Variance, and Dispersion Index of the weighted degree distribution at monthly aggregation in the German Media network dataset \cite{schmidt2018polarization}. The weighted degree distribution is the distribution of $Y_{it} =\sum_{j=1,j \neq i}^NY_{ijt}$.}
\label{tab:degree_media_de}
\end{table}

\begin{table}[h]
\resizebox{\textwidth}{!}{

\begin{tabular}{ccccccc}
\hline\hline
1st Semester 2015 & Jan-2015 & Feb-2015 & Mar-2015 & Apr-2015 & May-2015 & Jun-2015 \\ \hline
Mean & 8144.68 & 5510.68 & 6728.65 & 6082.10 & 6014.81 & 7272.87 \\
Variance & 195450287.01 & 87333153.89 & 139543720.30 & 121108860.58 & 104929048.55 & 150575068.25 \\
Dispersion Index & 23997.30 & 15847.99 & 20738.75 & 19912.35 & 17445.12 & 20703.66 \\ \hline

2nd Semester 2015 & Jul-2015 & Aug-2015 & Sep-2015 & Oct-2015 & Nov-2015 & Dec-2015 \\ \hline
Mean & 8124.00 & 6226.97 & 9009.68 & 6361.61 & 8842.77 & 8217.39 \\
Variance & 183356749.38 & 115425098.59 & 248206913.80 & 110120533.81 & 209929264.74 & 163131530.24 \\
Dispersion Index & 22569.76 & 18536.33 & 27548.92 & 17310.16 & 23740.20 & 19852.00 \\ \hline

1st Semester 2016 & Jan-2016 & Feb-2016 & Mar-2016 & Apr-2016 & May-2016 & Jun-2016 \\ \hline
Mean & 6754.13 & 6563.03 & 7047.06 & 8331.97 & 9000.87 & 8846.71 \\
Variance & 106079754.90 & 99203632.79 & 114529407.73 & 167116316.85 & 202117660.15 & 212860187.88 \\
Dispersion Index & 15705.91 & 15115.52 & 16252.07 & 20057.24 & 22455.34 & 24060.94 \\ \hline

2nd Semester 2016 & Jul-2016 & Aug-2016 & Sep-2016 & Oct-2016 & Nov-2016 & Dec-2016 \\ \hline
Mean & 9962.58 & 10445.61 & 11255.29 & 11168.87 & 13813.48 & 11905.19 \\
Variance & 288907977.40 & 288066222.41 & 342247056.21 & 350018235.43 & 529992259.11 & 385391864.55 \\
Dispersion Index & 28999.31 & 27577.72 & 30407.66 & 31338.73 & 38367.75 & 32371.74 \\ \hline\hline
\end{tabular}

}
\caption{Mean, Variance, and Dispersion Index of the weighted degree distribution at monthly aggregation in the French Media network dataset \cite{schmidt2018polarization}. The weighted degree distribution is the distribution of $Y_{it} =\sum_{j=1,j \neq i}^NY_{ijt}$.}
\label{tab:degree_media_fr}
\end{table}

\begin{table}[h]
\resizebox{\textwidth}{!}{
\begin{tabular}{ccccccc}
\hline\hline
1st Semester 2015 & Jan-2015 & Feb-2015 & Mar-2015 & Apr-2015 & May-2015 & Jun-2015 \\ \hline
Mean & 11995.42 & 12218.84 & 13125.64 & 13749.38 & 14941.64 & 15876.98 \\
Variance & 383058333.29 & 379733169.77 & 446752281.42 & 497134294.69 & 590226080.64 & 656962579.61 \\
Dispersion Index & 31933.71 & 31077.67 & 34036.60 & 36156.86 & 39502.08 & 41378.31 \\ \hline

2nd Semester 2015 & Jul-2015 & Aug-2015 & Sep-2015 & Oct-2015 & Nov-2015 & Dec-2015 \\ \hline
Mean & 17040.80 & 11789.82 & 12850.27 & 10419.69 & 10198.67 & 8484.44 \\
Variance & 694968931.12 & 348651559.33 & 416096188.20 & 266355630.81 & 258720652.32 & 175363266.53 \\
Dispersion Index & 40782.65 & 29572.25 & 32380.35 & 25562.72 & 25368.09 & 20668.80 \\ \hline

1st Semester 2016 & Jan-2016 & Feb-2016 & Mar-2016 & Apr-2016 & May-2016 & Jun-2016 \\ \hline
Mean & 8806.13 & 8487.24 & 7374.13 & 7931.82 & 7621.87 & 8703.91 \\
Variance & 195648863.44 & 186524546.78 & 141712167.85 & 159082399.06 & 146355198.80 & 196096361.45 \\
Dispersion Index & 22217.34 & 21977.04 & 19217.47 & 20056.22 & 19202.02 & 22529.68 \\ \hline

2nd Semester 2016 & Jul-2016 & Aug-2016 & Sep-2016 & Oct-2016 & Nov-2016 & Dec-2016 \\ \hline
Mean & 9500.80 & 9244.04 & 11043.64 & 10453.51 & 11790.93 & 12930.00 \\
Variance & 230012184.57 & 217565998.59 & 327209654.92 & 287653349.39 & 363213699.52 & 449497650.27 \\
Dispersion Index & 24209.77 & 23535.80 & 29628.77 & 27517.39 & 30804.49 & 34763.93 \\ \hline\hline
\end{tabular}}
\caption{Mean, Variance, and Dispersion Index of the weighted degree distribution at monthly aggregation in the Italian Media network dataset \cite{schmidt2018polarization}. The weighted degree distribution is the distribution of $Y_{it} =\sum_{j=1,j \neq i}^NY_{ijt}$.}
\label{tab:degree_media}
\end{table}

\begin{table}[t]
\resizebox{\textwidth}{!}{

\begin{tabular}{ccccccc}
\hline\hline
1st Semester 2015 & Jan-2015 & Feb-2015 & Mar-2015 & Apr-2015 & May-2015 & Jun-2015 \\ \hline
Mean & 5531.49 & 6258.42 & 7310.98 & 5872.33 & 7642.74 & 9456.33 \\
Variance & 53128148.64 & 73187128.68 & 88048937.93 & 66798210.80 & 124347949.19 & 194856900.51 \\
Dispersion Index & 9604.68 & 11694.19 & 12043.39 & 11375.09 & 16270.07 & 20605.98 \\ \hline

2nd Semester 2015 & Jul-2015 & Aug-2015 & Sep-2015 & Oct-2015 & Nov-2015 & Dec-2015 \\ \hline
Mean & 9658.37 & 7027.63 & 7774.70 & 5873.77 & 6053.58 & 8469.77 \\
Variance & 185805904.48 & 97758492.10 & 134993961.74 & 65034942.04 & 69092229.25 & 144620043.71 \\
Dispersion Index & 19237.81 & 13910.60 & 17363.24 & 11072.10 & 11413.45 & 17074.85 \\ \hline

1st Semester 2016 & Jan-2016 & Feb-2016 & Mar-2016 & Apr-2016 & May-2016 & Jun-2016 \\ \hline
Mean & 7588.37 & 6598.98 & 7145.07 & 6829.63 & 6626.79 & 7335.63 \\
Variance & 105888823.43 & 85816504.64 & 101136588.54 & 87514670.29 & 81010675.93 & 111756925.38 \\
Dispersion Index & 13954.09 & 13004.52 & 14154.74 & 12813.97 & 12224.72 & 15234.81 \\ \hline

2nd Semester 2016 & Jul-2016 & Aug-2016 & Sep-2016 & Oct-2016 & Nov-2016 & Dec-2016 \\ \hline
Mean & 6021.58 & 6069.07 & 7235.95 & 8123.02 & 6394.09 & 5279.35 \\
Variance & 63464352.30 & 67193193.97 & 101961831.62 & 135283964.31 & 75520870.85 & 46575264.71 \\
Dispersion Index & 10539.48 & 11071.42 & 14091.00 & 16654.39 & 11811.04 & 8822.16 \\ \hline\hline
\end{tabular}
}
\caption{Mean, Variance, and Dispersion Index of the weighted degree distribution at monthly aggregation in the Spanish Media network dataset \cite{schmidt2018polarization}. The weighted degree distribution is the distribution of $Y_{it} =\sum_{j=1,j \neq i}^NY_{ijt}$.}
\label{tab:degree_media_sp}
\end{table}

\begin{figure}[t]
\centering
\resizebox{0.65\textwidth}{!}{
\begin{tabular}{c}
    \includegraphics[width = 0.7\textwidth]{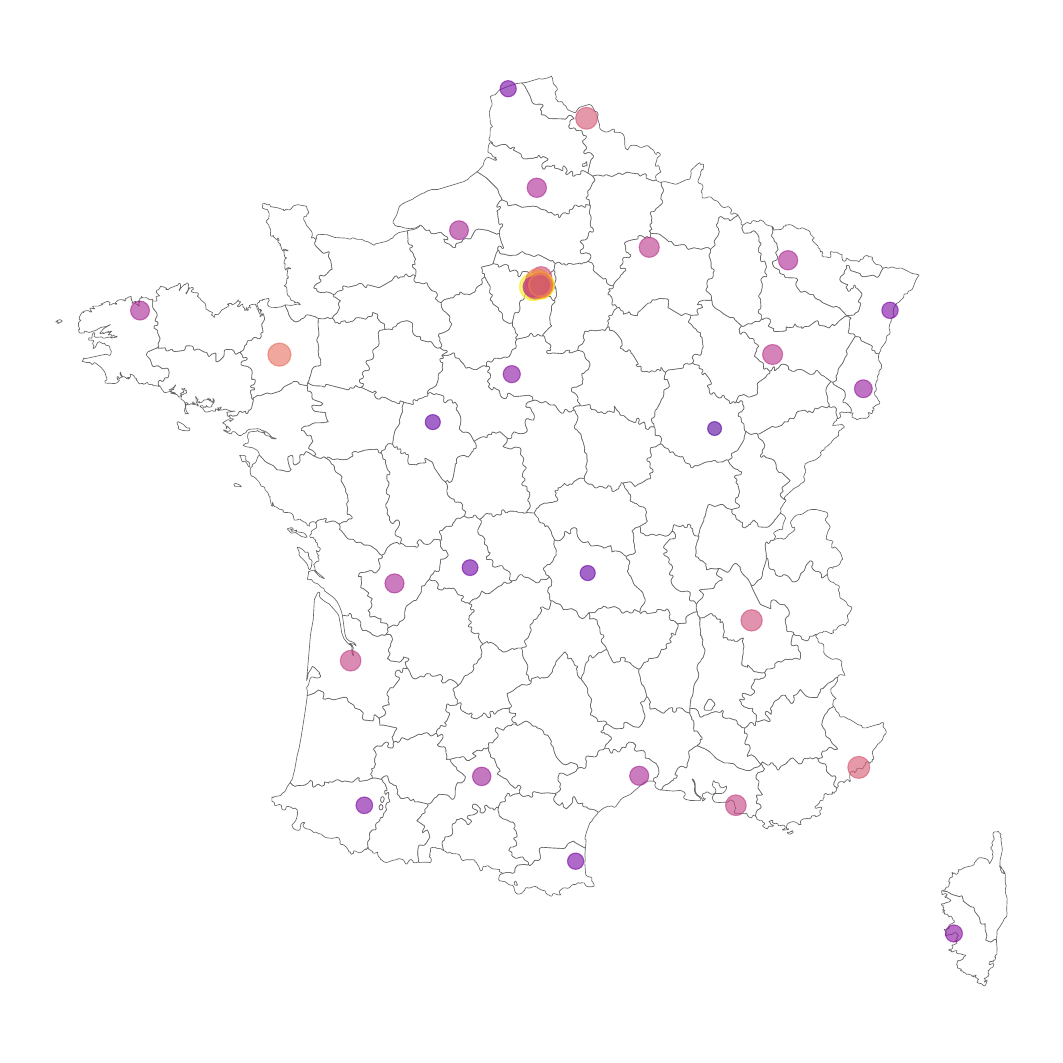} \\
    \includegraphics[width = 0.7\textwidth]{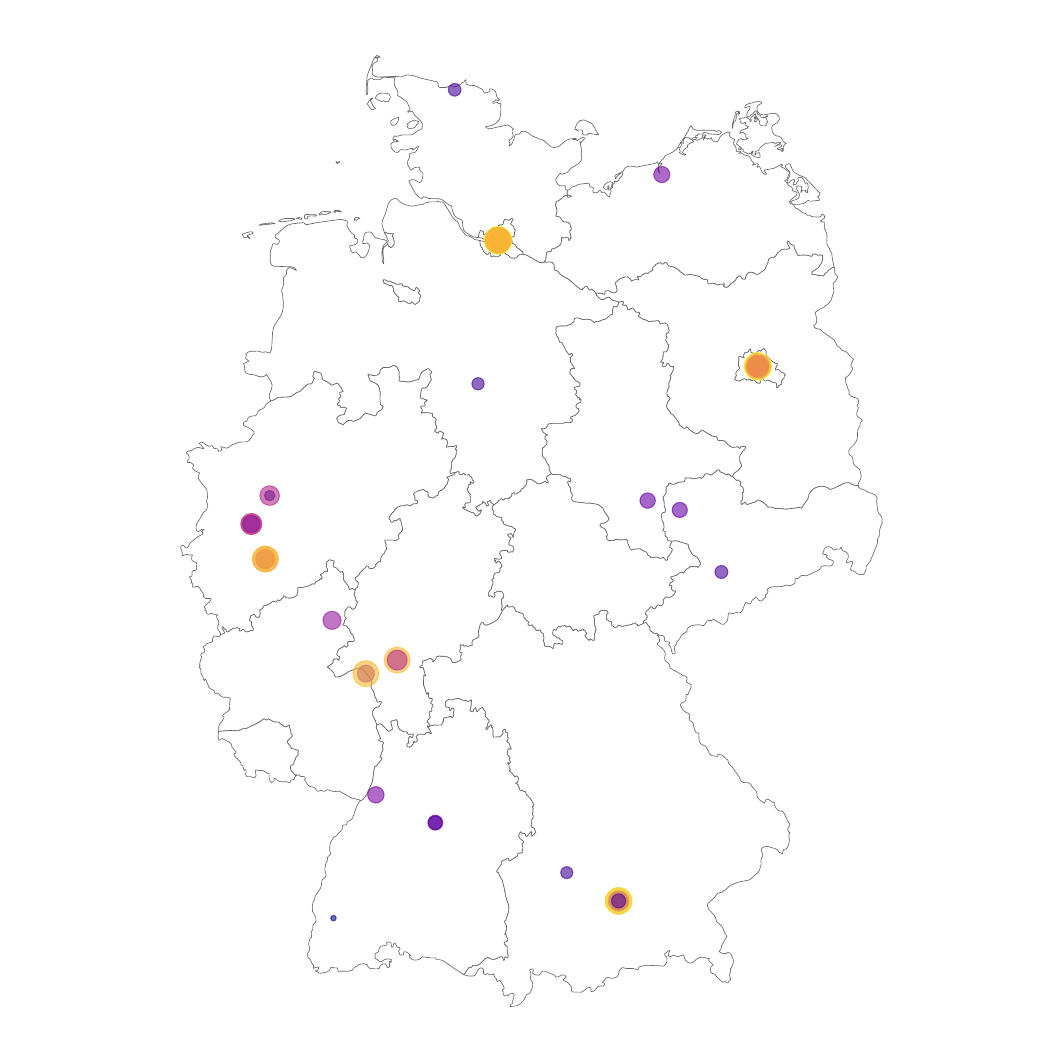}
\end{tabular}}
\caption{\textbf{Media Network, Model $\mathcal{M}_3$ Results:}. Geographical distribution of the news outlets for France (top) and Germany (bottom) with node color and size proportional to the posterior mean of $\alpha_i$ (\usebox{\purplesmalldot} - \usebox{\yellowbigdot}).}
\label{fig:media_bias_de_frMAP}
\end{figure}

\vfill 
\newpage 

\begin{figure}[htbp]
\centering
\resizebox{0.65\textwidth}{!}{
\begin{tabular}{c}
    \includegraphics[width = 0.7\textwidth]{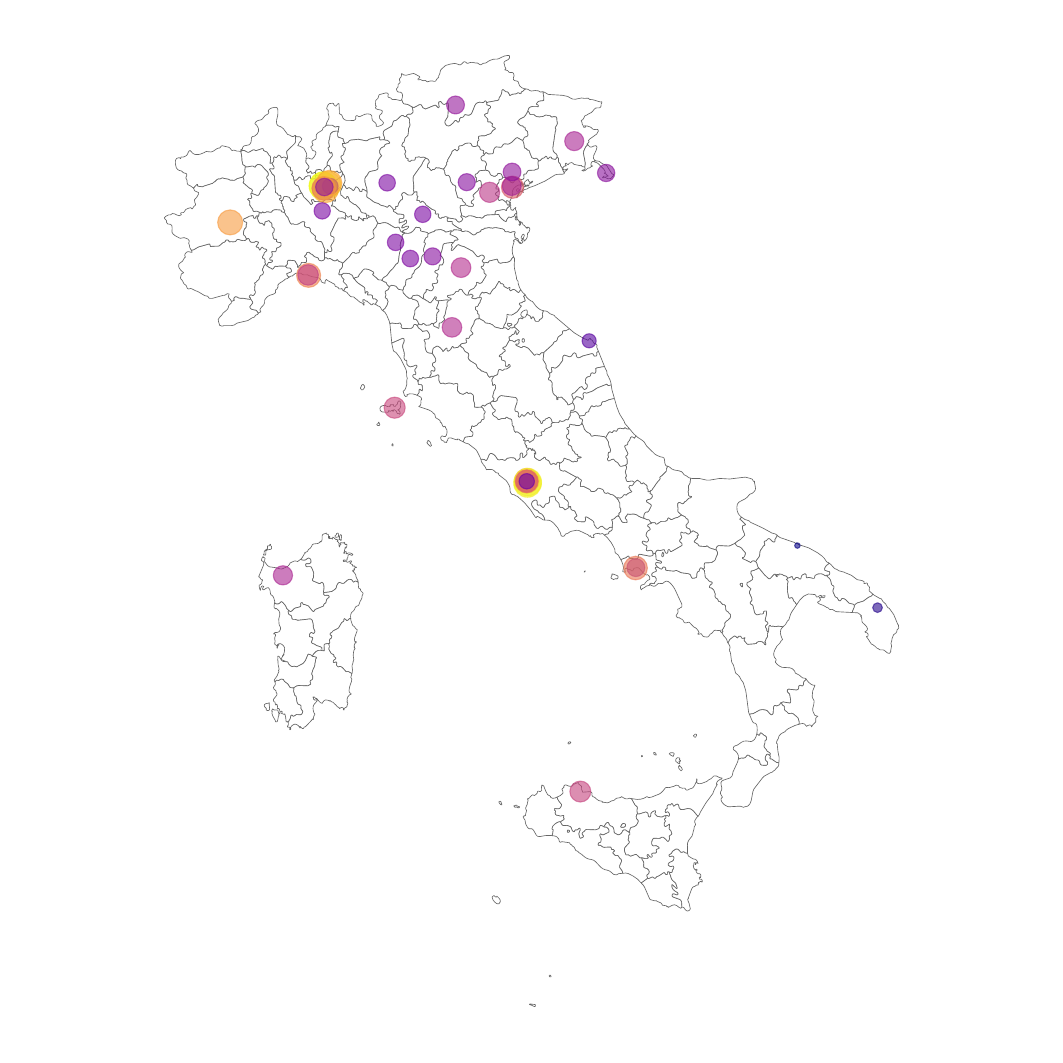} \\
    \includegraphics[width = 0.7\textwidth]{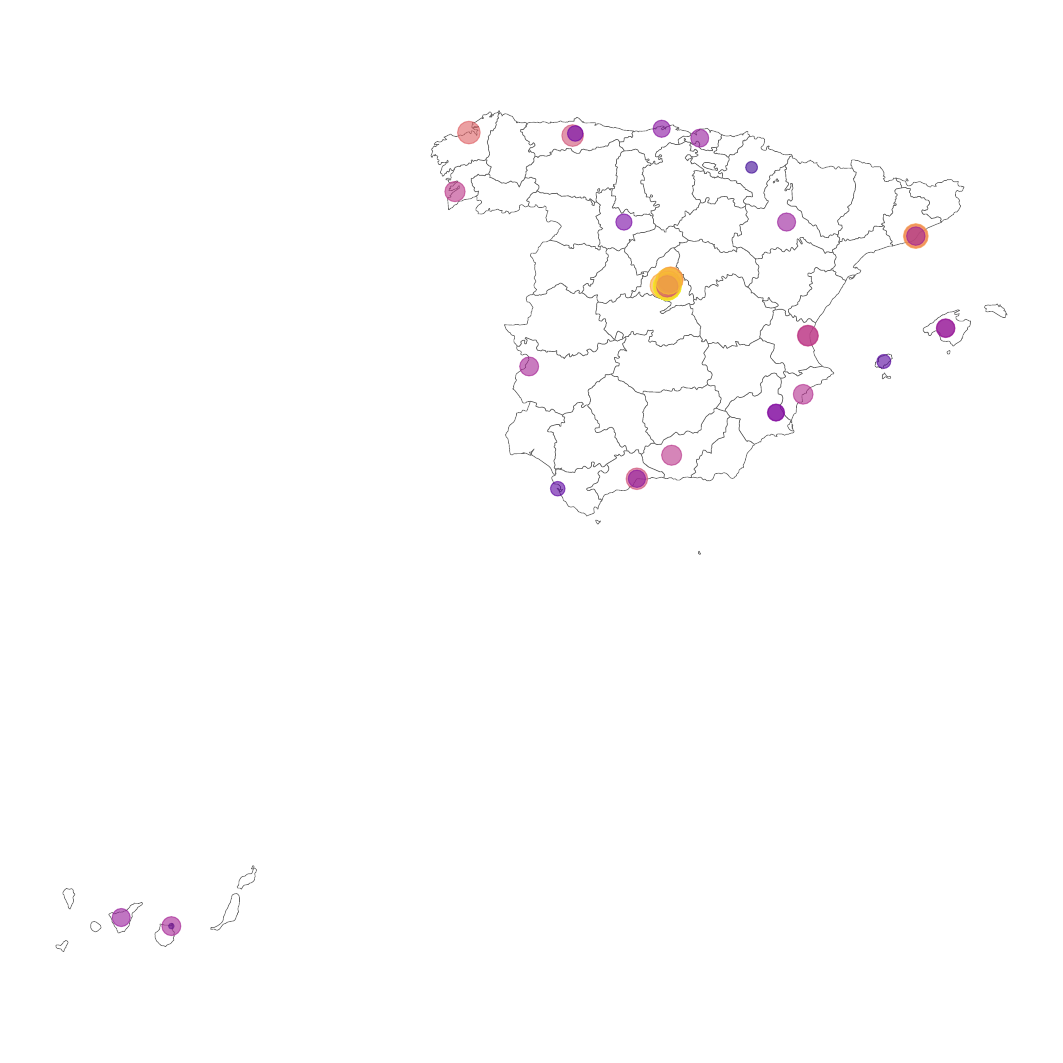}
\end{tabular}}
\caption{\textbf{Media Network, Model $\mathcal{M}_3$ Results:}. Geographical distribution of the news outlets for Italy (top) and Spain (bottom) with node color and size proportional to the posterior mean of $\alpha_i$ (\usebox{\purplesmalldot} - \usebox{\yellowbigdot}).}
\label{fig:media_bias_it_spMAP}
\end{figure}

\vfill 
\newpage 

\begin{figure}[t]
\centering
\begin{tabular}{c}
    \includegraphics[width = 0.98\textwidth]{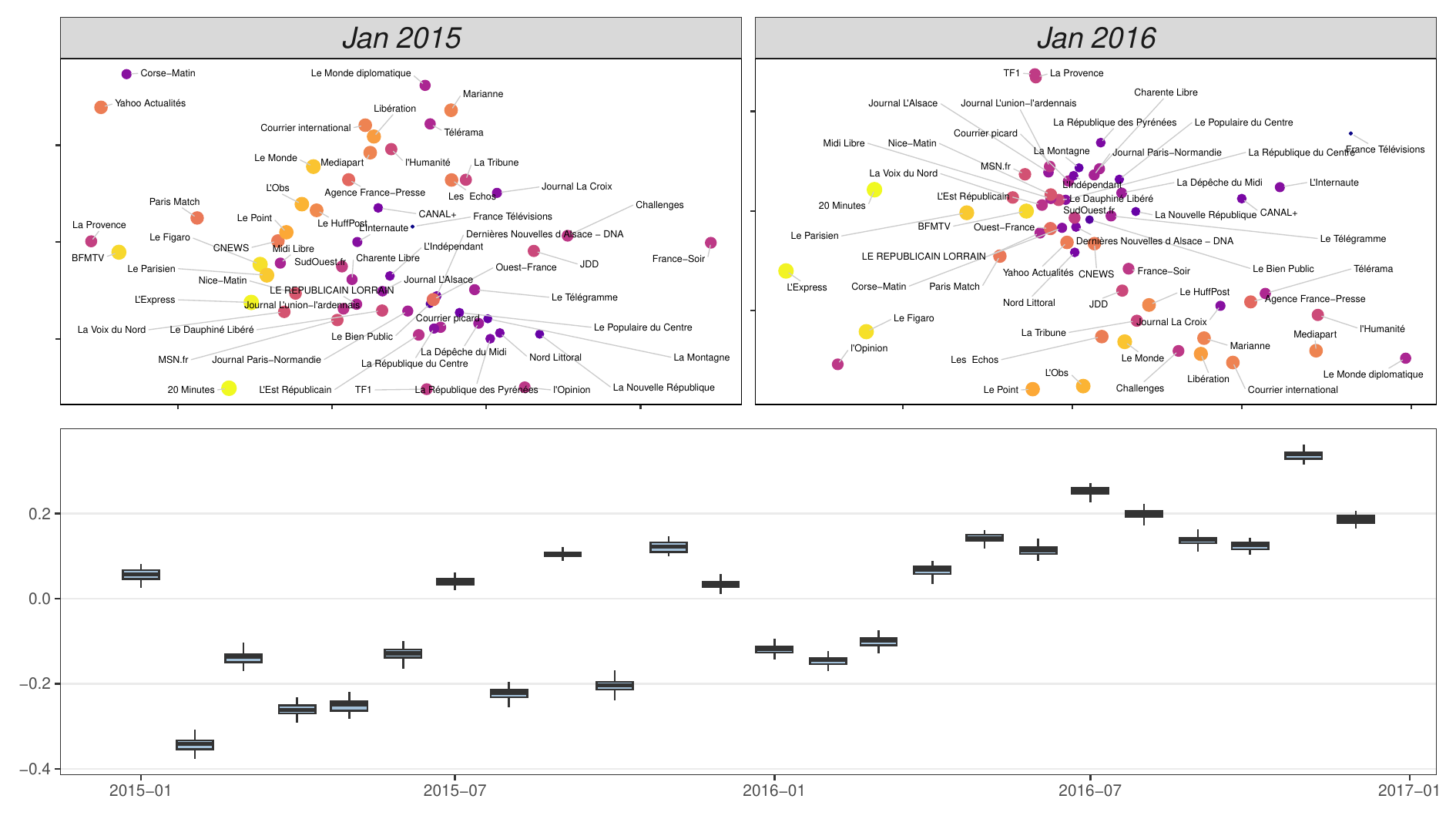} \\
    \includegraphics[width = 0.98\textwidth]{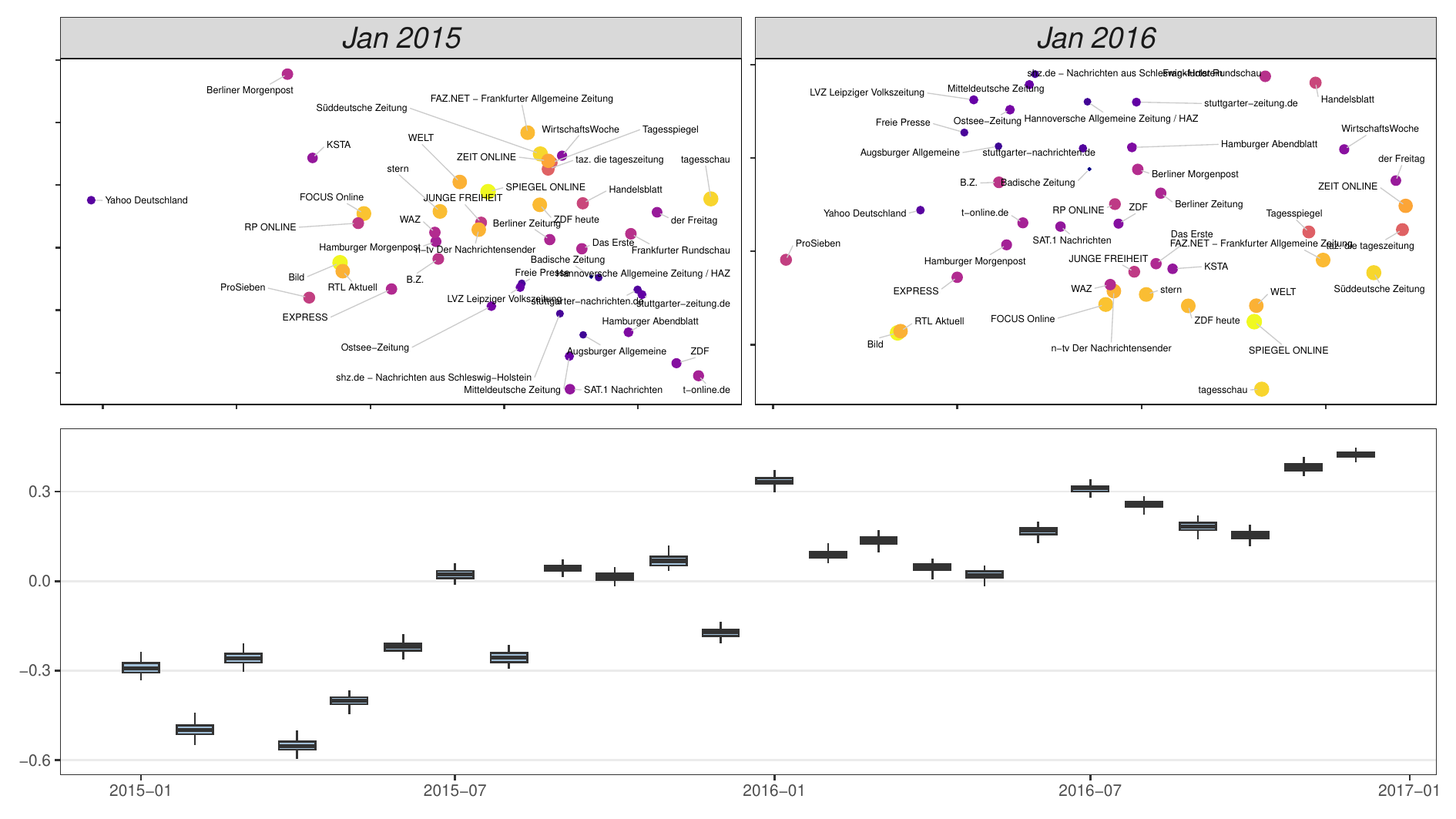}
\end{tabular}
\caption{\textbf{Media Network, Model $\mathcal{M}_3$ Results:}.
Posterior mean of the latent coordinates with node color proportional to the posterior mean of $\alpha_i$ (January 2015 and January 2016) and posterior estimates of the time-varying latent factor $f_t$ for France (top) and Germany (bottom).}
\label{fig:media_bias_de_fr}
\end{figure}

\vfill 
\newpage 

\begin{figure}[htbp]
\centering
\begin{tabular}{c}
    \includegraphics[width = 0.98\textwidth]{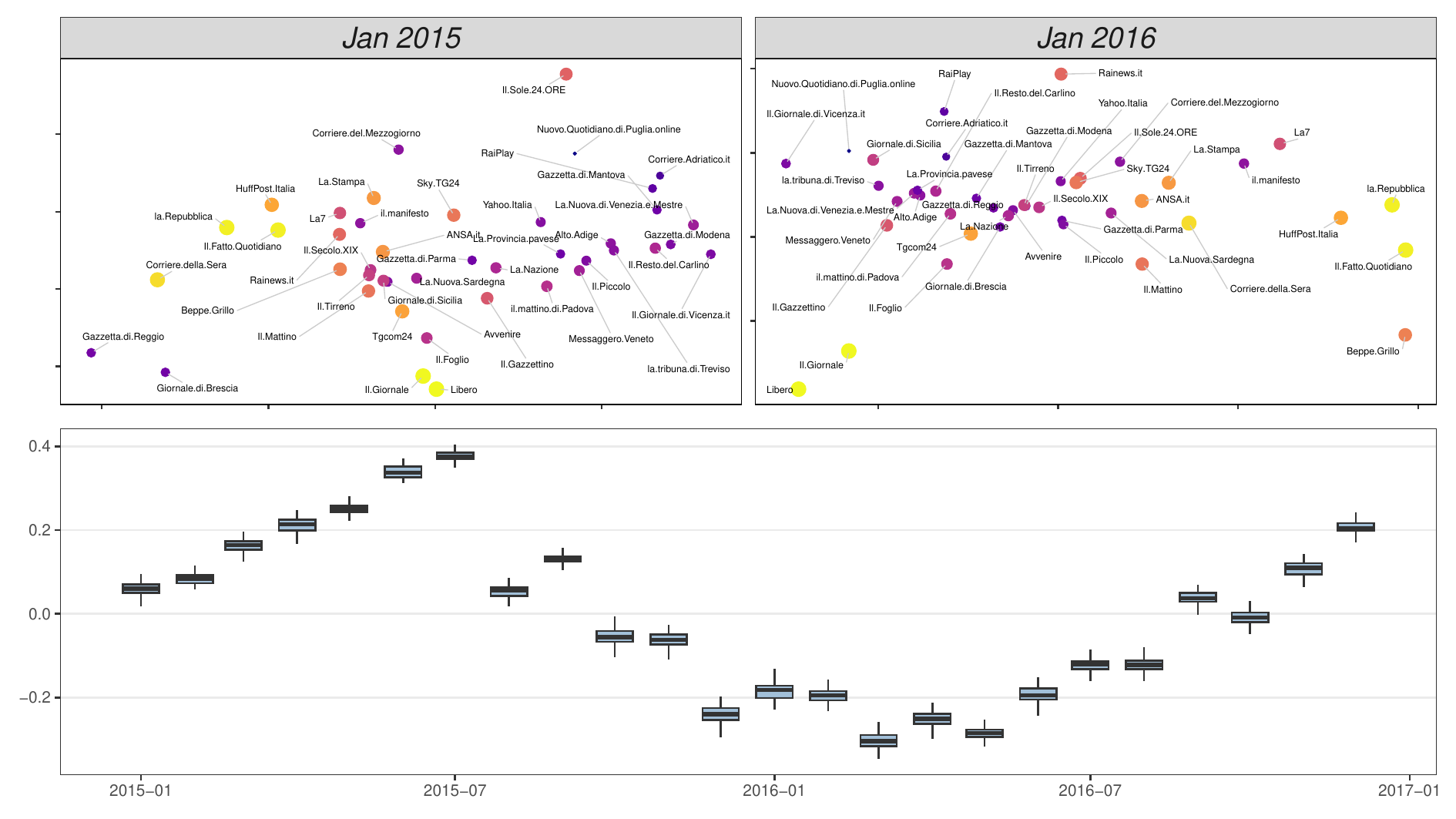} \\
    \includegraphics[width = 0.98\textwidth]{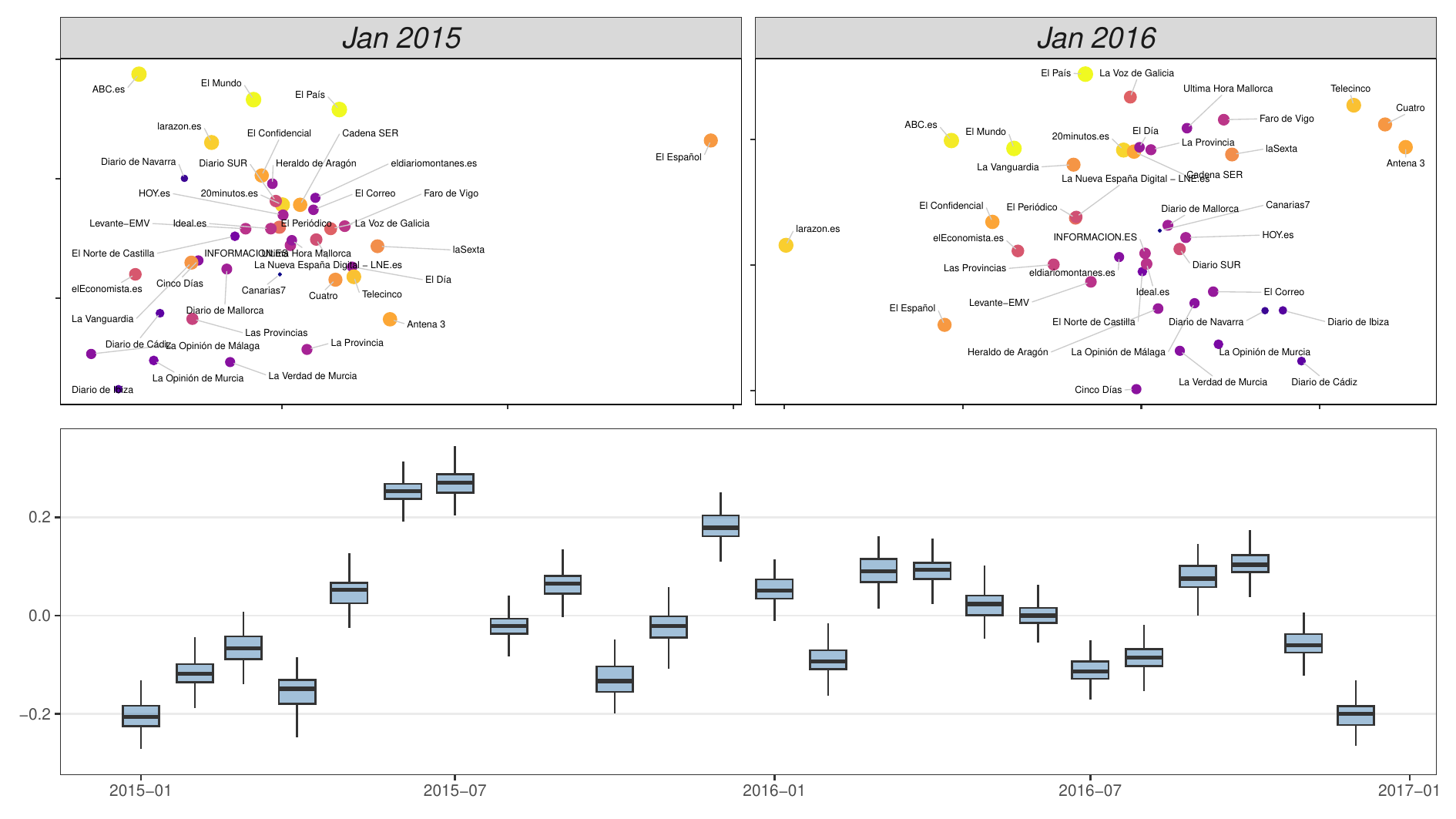}
\end{tabular}
\caption{\textbf{Media Network, Model $\mathcal{M}_3$ Results:}. Geographical distribution of the news outlets for Italy (top) and Spain (bottom) with node color and size proportional to the posterior mean of $\alpha_i$ (\usebox{\purplesmalldot} - \usebox{\yellowbigdot}).}
\label{fig:media_bias_it_sp}
\end{figure}

\end{appendix}

\end{document}